\documentclass[12pt,a4paper]{article} 
\usepackage[utf8]{inputenc}
\usepackage[left=3cm,right=3cm,top=3cm,bottom=3cm]{geometry}

\usepackage{authblk} % Für Affiliation

\author[1]{Eva-Maria Maier}
\author[1,2]{Almond St\"ocker}
\author[3]{Bernd Fitzenberger}
\author[1]{Sonja Greven}
\affil[1]{Chair of Statistics, School of Business and Economics, Humboldt-Universit\"at zu Berlin, Germany}
\affil[2]{Institute of Mathematics, EPFL, Lausanne, Switzerland}
\affil[3]{IAB (Institute for Employment Research), Nuremberg, Germany}

\date{} % Kein Datum

\usepackage[english]{babel}

\usepackage{fancyhdr}

\usepackage{enumerate}

\usepackage{amsmath}
\usepackage{mathtools}
\usepackage{amsfonts}
\usepackage{amssymb}
\usepackage{amsthm}
\usepackage{esvect}

\usepackage{nicefrac}

\usepackage{dsfont}
\usepackage{bbm}

\usepackage{float}
\usepackage[rflt]{floatflt}
\usepackage{graphicx}
\usepackage{xcolor}

\newcommand{\bpl}{B^p(\mu)}
\newcommand{\B}{B^2(\mu)}
\newcommand{\Bl}{B^2(\lambda)}
\newcommand{\Bd}{B^2(\delta^\bullet)}
\newcommand{\Ln}{L^2_0(\mu)}
\newcommand{\Lnl}{L^2_0(\lambda)}
\newcommand{\Lnd}{L^2_0(\delta^\bullet)}
\newcommand{\Ltwo}{L^2(\mu)}

\newcommand{\lpl}{L^p(\mu)}
\newcommand{\lpnl}{L^p_0(\mu)}
\newcommand{\fnu}{f_\nu}

\DeclareMathOperator{\clr}{clr}
\newcommand{\rangleb}{\rangle_{\B}}

\newcommand{\Vertb}{\Vert_{\B}}
\newcommand{\Vertbl}{\Vert_{\Bl}}
\newcommand{\Vertbd}{\Vert_{\Bd}}

\newcommand{\fh}{\hat{f}}

\newcommand{\gh}{\hat{g}}

\newcommand{\isb}{=_{\Bcal}}
\newcommand{\rhoy}{\rho_{y_i}}
\newcommand{\rhof}{\rho_{f_i}}
\newcommand{\Jc}{\iota_{\mathrm{c}}}
\newcommand{\Jd}{\iota_{\mathrm{d}}}
\newcommand{\fc}{f_{\mathrm{c}}}
\newcommand{\fd}{f_{\mathrm{d}}}
\newcommand{\Jtc}{\tilde{\iota}_{\mathrm{c}}}
\newcommand{\Jtd}{\tilde{\iota}_{\mathrm{d}}}
\newcommand{\ftc}{\tilde{f}_{\mathrm{c}}}
\newcommand{\ftd}{\tilde{f}_{\mathrm{d}}}

\newcommand{\Sfc}{\Scal_\lambda (f_{\mathrm{c}})}

\newcommand{\SfT}{\Scal_{\Tcal}(f)}
\newcommand{\SfTn}{\Scal_{\Tcal_0}(f)}

\newcommand{\ootimes}{\unitlength1ex
\begin{picture}(3,2)
\put(1.5,0.6){\circle{2}}\put(1.5,0.6){\makebox(0,0){$\otimes$}}
\end{picture}}

\newtheorem{thm}{Theorem}[section]

\newtheorem{prop}[thm]{Proposition}

\newtheoremstyle{cited}{}{}{\itshape}{}{%\bfseries
}{\textbf{.}}{.5em}{\thmname{\textbf{#1}}\thmnumber{ \textbf{#2}} #3%
}

\theoremstyle{cited}

\theoremstyle{definition}

\newcommand{\ra}{\rightarrow}

\newcommand{\Lra}{\Longrightarrow}
\newcommand{\Ra}{\Rightarrow}

\newcommand{\La}{\Leftarrow}

\newcommand{\Llra}{\Longleftrightarrow}
\newcommand{\Lera}{\Leftrightarrow}

\newcommand{\Ebb}{\mathbb{E}}

\newcommand{\Nbb}{\mathbb{N}}
\newcommand{\Pbb}{\mathbb{P}}

\newcommand{\Rbb}{\mathbb{R}}

\newcommand{\bfe}{\mathbf{b}}

\newcommand{\Cf}{\mathbf{C}}

\newcommand{\eh}{\hat{e}}

\newcommand{\hh}{\hat{h}}

\newcommand{\If}{\mathbf{I}}

\newcommand{\Pf}{\mathbf{P}}

\newcommand{\Qf}{\mathbf{Q}}

\newcommand{\yh}{\hat{y}}

\newcommand{\xf}{\mathbf{x}}
\newcommand{\Zf}{\mathbf{Z}}

\newcommand{\betah}{\hat{\beta}}
\newcommand{\epsh}{\hat{\varepsilon}}
\newcommand{\epst}{\tilde{\varepsilon}}

\newcommand{\gammaf}{\boldsymbol{\gamma}}

\newcommand{\gammafh}{\hat{\boldsymbol{\gamma}}}

\newcommand{\thetaf}{\boldsymbol{\theta}}

\newcommand{\zetaf}{\boldsymbol{\zeta}}

\newcommand{\zetafh}{\hat{\boldsymbol{\zeta}}}

\newcommand{\Id}{\mathbf{I}}

\newcommand{\ft}{\tilde{f}}

\newcommand{\dmu}{\mathrm{d}\mu}
\newcommand{\dmumu}{\mathrm{d}(\mu \otimes \mu)}
\newcommand{\dnu}{\mathrm{d}\nu}
\newcommand{\dlamb}{\mathrm{d}\lambda}
\newcommand{\ddel}{\mathrm{d}\delta}
\newcommand{\ddelb}{\mathrm{d}\delta^\bullet}
\newcommand{\dt}{\mathrm{d}t}

\newcommand{\dups}{\mathrm{d}\upsilon}

\newcommand{\Cov}{\mathrm{Cov}}

\newcommand{\SSE}{\mathrm{SSE}}

\newcommand{\rMSE}{\mathrm{relMSE}}

\DeclareMathOperator{\id}{id}

\DeclareMathOperator*{\argmin}{argmin}

\newcommand{\OR}{O\! R}

\newcommand{\Acal}{\mathcal{A}}
\newcommand{\Bcal}{\mathcal{B}}
\newcommand{\Ccal}{\mathcal{C}}
\newcommand{\Dcal}{\mathcal{D}}

\newcommand{\Mcal}{\mathcal{M}}

\newcommand{\Pcal}{\mathcal{P}}
\newcommand{\Scal}{\mathcal{S}}
\newcommand{\Tcal}{\mathcal{T}}

\newcommand{\Ycal}{\mathcal{Y}}

\usepackage[
	backend=biber, 
    giveninits=true, % kürzt Vornamen durch Initiale ab
    style=authoryear,	
	autocite=inline,
	mincitenames=1, % Anzahl der Autoren, die bei der Zitation angegeben werden, wenn mit et al. abgekürzt wird
	maxcitenames=2, % Anzahl der Autoren, bis zu der nicht mit et al. abgekürzt wird
	maxbibnames=99, % Anzahl der Autoren, die in der Bibliographie ausgegeben wird
	uniquename=false,
	uniquelist=minyear, % Wichtig, damit bei Zitation mincitenames nicht überschrieben wird, wenn derselbe Autor in verschiedenen Jahren veröffentlicht hat (eindeutige Zuordnung der Zitation)
	labeldateparts=true, % Muss bei uniquelist=minyear gesetzt werden
	doi=false,
	url=false,
	eprint=false,
	isbn=false,
	natbib=true, % Ermöglicht \citep und \citet
	date=year
	]{biblatex}
	
\DeclareNameAlias{sortname}{last-first} % Anordnung der Namen in der Bibliographie

\renewbibmacro{in:}{} % Supress "in"
\DeclareFieldFormat{pages}{#1}
\DeclareFieldFormat{pagetotal}{#1~pages}
\DeclareFieldFormat[article]{volume}{\mkbibbold{#1}} % Volume bold
\DeclareFieldFormat[article,inbook,incollection,inproceedings,patent,thesis,unpublished]{title}{#1} %{\mkbibquote{#1\isdot}} % Suppress quotation marks

\AtEveryBibitem{\clearfield{number}}
\AtEveryBibitem{\clearfield{issue}}
\numberwithin{equation}{section}
\numberwithin{table}{section}
\numberwithin{figure}{section}

\usepackage{hyperref}
\usepackage{multirow}
\usepackage{subcaption}
\usepackage{csquotes}

\newcounter{jasa}
\newcounter{aoas}
\begin{document}

% Suggestions for a shorter title:
% was: Boosting Flexible Regression Models for Compositional Data and Probability Densities in Bayes Hilbert Spaces
% 1. Flexible Regression Models for Probability Densities in Bayes Hilbert Spaces
% 2. Additive Density-on-Scalar Regression in Bayes Hilbert Spaces
% 3. Additive Density-on-Scalar Regression in Bayes Hilbert Spaces with an application to Gender Economics
\title{Additive Density-on-Scalar Regression in Bayes Hilbert Spaces with an Application to Gender Economics}

\maketitle

\begin{refsection}
% Zusammenfassung
\begin{abstract}
\noindent
Motivated by research on gender identity norms and the distribution of the woman’s share in a couple’s total labor income, we consider functional additive regression models for probability density functions as responses with scalar covariates. 
To preserve nonnegativity and integration to one under vector space operations, we formulate the model for densities in a Bayes Hilbert space, % with respect to some finite measure. This
which allows to not only consider continuous densities, but also, e.g., discrete or mixed densities. 
Mixed ones occur in our application, as the woman’s income share is a continuous variable having discrete point masses at zero and one for single-earner couples. 
Estimation is based on a gradient boosting algorithm, allowing for potentially numerous flexible covariate effects and model selection. 
We develop properties of Bayes Hilbert spaces related to subcompositional coherence, yielding (odds-ratio) interpretation of effect functions and simplified estimation for mixed densities via an orthogonal decomposition. 
Applying our approach to data from the German Socio-Economic Panel Study (SOEP) shows a more symmetric distribution in East German than in West German couples after reunification and a smaller child penalty comparing couples with and without minor children. 
These West-East differences become smaller, but are persistent over time.
%between couples with and without minor children, as well as trends over time. 
%Supplemental materials are available online.

\vspace{0.5cm}
\noindent%
\end{abstract}

% Inhaltsverzeichnis
%\tableofcontents

%auto-ignore
\section{Introduction}

In the core of their discussion of economic consequences of gender identity, %\citet{bertrand2014} and 
\citet{bertrand2015} consider % the \emph{female share distribution}, 
the distribution of a wife's share % $s$ 
in the total labor income of a wife-husband couple in the U.S., represented by the density. % $f: [0,1] \rightarrow \mathbb{R}^+$ of $s$. % estimated over the population. 
%In their case, the population corresponds to U.S.\ couples. 
They focus on the hypothesis that the distribution exhibits a distinct drop at $0.5$, which is attributed to gender identity norms according to which a husband should earn more than his wife. %with men being averse to a situation where their female partners make more money than themselves. 
Subsequent studies on couples in Germany show a mixed picture with respect to this drop (e.g., \citealp{sprengholz2020, kuehnle2021}),
%\\
% but indicate that the female share distribution clearly varies with covariates.
while also indicating that distributions differ in West compared to East Germany. % other than in West Germany, the distributions in East Germany are relatively symmetric (see also Figure \ref{original_densities_south_east}, bottom). %, contradicting the drop-hypothesis. % , leading us to focus on a different question of interest. 
%Inspired by these works, w
% Establishing the female share distribution more systematically as statistical object in an empirical gender-economic analysis, we consider it in different socio-politic contexts, i.e.\ the distribution of $s$ conditional on different covariates, for different geographic regions and over time. 
% This is of great interest, e.g., when taking the child situation in the household into account, since employment and earnings of female partners show a strong childhood penalty \citep{kleven2019,fitzenberger2013} while social norms change over time towards higher employment of females, with part-time employment becoming more prevalent, especially in the presence of children.
% Despite its relevance, most of the literature does not consider changes in the female share distribution, so far -- potentially also because of a lack of systematic frameworks for their analysis.
Furthermore, employment and earnings of female partners show a strong childhood penalty \citep{kleven2019,fitzenberger2013} while social norms change over time towards higher employment of females, with part-time employment becoming more prevalent, especially in the presence of children.
Thus, it is of great interest to take the child situation in the household and time trends into account.
This highlights the relevance of analyzing female share distributions depending on covariates, which is not done systematically so far -- potentially also because of a lack of convenient frameworks. 
We aim to fill this gap, introducing a regression approach to analyze probability densities given scalar covariates. 
%
% This motivates our development of a systematic framework to analyze to analyze probability density functions conditional on scalar covariates.
% Furthermore, it is of great interest to take the child situation in the household into account, since employment and earnings of female partners show a strong childhood penalty \citep{kleven2019,fitzenberger2013} while social norms change over time towards higher employment of females, with part-time employment becoming more prevalent, especially in the presence of children.
% Despite its relevance, most of the literature does not consider changes in the female share distribution, so far -- potentially also because of a lack of systematic frameworks for their analysis. 
% We aim to fill this gap, introducing a regression approach to analyze probability density functions conditional on scalar covariates. 
\\ 
Densities $f_i$ reflecting distributions in different (sub-)populations $i=1,\dots, N$ preserve more information than scalar statistics like the mean, % (like the mean),
enabling more in-depth investigations and insights. % compared to classical regression approaches for scalar response (e.g. modeling the mean of the response). % as in the motivating studies above. 
In particular, they give a more fine-grained picture of often multi-modal income share distributions (Figure \ref{original_densities_south_east}, top) and show individual variability in the population, avoiding over-simplification. %that might occur in, say, scalar mean regression.
Understanding the density functions as genuine object of analysis, however, demands for suitable statistical methodology: 
We will model densities in dependence on scalar covariates, which we refer to as \emph{density-on-scalar regression} along the lines of function-on-scalar regression in functional data analysis \citep{morris2015}. %\citep{ramsay2005}. 
Although densities have been modeled via traditional functional regression models in $L^2$ spaces in the past \citep[e.g.][]{park2012}, this is problematic, as it does not reflect the particular properties of density functions (nonnegativity and integration to one). %multiplying a density with a negative scalar or adding two densities in the classical sense immediately violates the non-negativity and integrate-to-one constraints of densities.
Instead, we consider the $f_i$ as elements of a \emph{Bayes Hilbert space} \citep{egoz2006,vdb2014}, based on an alternative normed vector space structure for densities.
The additive density-on-scalar regression model framework we introduce extends the range of available covariate effects compared to the linear density-on-scalar model of \citet{talska2017} by non-linear effects and complements the %kernel-based 
additive regression model for general Hilbert space responses of \citet{jeon2020additive}, who utilize backfitting with a Nadaraya-Watson-type estimator for smooth main effects of continuous covariates, 
%only discuss compositional data as response presenting the finite (in their case three) dimensional analogue of density functions in a Bayes Hilbert space, 
and do not provide a comparably modular and ready-to-use framework for statistical modeling as well as no implementation. 
Moreover, in contrast to earlier gender-economic and Bayes Hilbert space literature, we consider mixed continuous-discrete distributions where the densities $f_i$ are not solely with respect to the Lebesgue measure,  but have additional point masses, here at $0$ and $1$, where either the woman or the man has zero income.
\\  
Apart from the Bayes Hilbert space approach, analysis of densities (or respective distributions) has been based on different alternative mathematical representations. 
%\citet{park2012} discuss %transformation-free
%density-on-density regression without any positivity constraints, performing linear regression directly on the deviations from the mean density.
Wasserstein distances have been employed, e.g., by \citet{petersen2019} for linear density-on-scalar regression and by \citet{ghodrati2022distribution} for specialized density-on-density regression. 
The Fisher-Rao metric as another option was recently used by \citet{zhao2023density} for a similarly specialized density-on-density approach.
Log hazard and log quantile transformations have been proposed to represent a distribution in an $L^2$ space, which was used by \citet{han2020} to apply additive functional regression models for density-on-scalar regression. 
%\citet{happ2019} show that for both, numerical instabilities may occur and finally prefer the clr transformation to both.
%, considering only the continuous case.
%In contrast to \citet{han2020}, which only considers the transformed densities for modeling, our Bayes Hilbert space approach provides an entire conceptual framework that allows to embed the densities and specify the model in a vector space structure.
%The clr transformation, an isometric isomorphism, allows an equivalent formulation in the $L^2$, which %simplifies estimation but allows for
%enables appealing odds-ratio-type interpretations on the original density-level.
Compared to these alternatives, statistical analysis in Bayes Hilbert spaces, besides the mathematical convenience of modeling densities in a linear space, offers a major advantage: 
generalizing the Aitchison geometry \citep{aitchison1986} to infinite dimensions, 
we may expect \emph{subcompositional coherence}, a central principle in compositional data analysis, to carry over to analyses of density data.
Translated to probability distributions, the principle states that an analysis conditioning on a subdomain of the densities must not be contrary to analyzing the whole densities. 
For our application, for example, this translates to consistency of analyses for all and for double-earner couples only.
Besides being generally desirable, this is of practical relevance, as it allows to reduce results to smaller regions for a detailed interpretation, which, as we will show, corresponds to familiar odds-ratio interpretations in scalar logit regression.
Moreover, we show how restriction of the density to a subdomain can be viewed as an orthogonal projection, implying the property of \emph{subcompositional dominance}  known from compositional data analysis % While this cannot be expected in other, non-linear geometries, the related property that  
\citep[distances of densities should be smaller or equal when restricted to a subdomain;][]
{egozglahn2011BasicConceptsProcedures} 
to also hold in Bayes Hilbert spaces. These properties do not hold % which is not the case 
for the other approaches mentioned above. 
\\
There is a variety of less directly connected approaches in the literature which, 
instead of modeling a conditional mean density of a sample of density functions, model the conditional distribution of a scalar response variable beyond scalar mean regression. 
These include generalized additive models for location, scale and shape (GAMLSS) modeling multiple distribution parameters, also referred to as distributional regression \citep[e.g.,][]{rigby2005}, conditional transformation models \citep[e.g.,][]{hothorn2014}, 
quantile regression \citep[e.g.,][]{koenker2005} and distribution function regression \citep[e.g.,][]{hall1999}, 
as well as various approaches to conditional density estimation \citep[e.g.,][]{gu1995conditionallogsplinedensity,macEachern1999,li2021}.
Although related, they address a fundamentally different problem from the one we focus on here.
\\ 
The contributions of this paper go well beyond our motivating analysis of the female income share distribution to express gender-based income differences,
an important issue of major interest:
I.\ We establish the (estimated) female share distribution itself as object of statistical analysis beyond its previous descriptive use. 
II.\ For its analysis, we propose an additive density-on-scalar regression framework. Models are fitted via gradient boosting, which we III.\ formulate for responses in Bayes Hilbert spaces.
We integrate the approach and its implementation into the modular functional boosting framework provided by the \texttt{R} package \texttt{FDboost} \citep{brh2020}. 
The component-wise fitting facilitates specification of parameter-intense functional effects %, as often desirable in high-dimensional, functional data, 
and avoids over-fitting via early stopping based on density-wise cross-validation. 
This also yields inherent model selection, which enables identifying relevant variables as an alternative to statistical testing.
IV.\ We consider continuous densities, discrete probability mass functions (compositional data), and, unlike previous work, also mixtures of both within one unified framework. 
This is motivated by the nature of the female share distribution and based on Bayes Hilbert spaces for general finite measures \citep{vdb2014}.  
%While theorems on the basic structure of the space are adopted from the literature or only marginally extended, we 
V.\ We derive useful properties for Bayes Hilbert  spaces  related to the principle of subcompositional coherence, %of density decomposition 
facilitating detailed analytic interpretations in such spaces that are relevant also beyond regression. These include point/interval-wise odds ratio interpretations of differences and model effects (Proposition \ref{thm:oddsratio}), conditional densities as orthogonal projections (Proposition \ref{thm_subcompositional_coherence}) and orthogonal decomposition of mixed densities into continuous and discrete components (Proposition B.1). %\ref{thm_decomposition_bayes}).
Using I.-V.\, we then VI.\ investigate gender-specific income differences in German couples based on the Socio-Economic Panel (SOEP, \citealp{soep2019}), 
clearly illustrating different share distributions depending on the child status, and for East vs. West Germany, with some assimilation occurring after reunification but also differences persisting over time. 
Due to its history of two different political systems, the case of Germany is particularly interesting and nicely shows the usefulness of the proposed approach. 
A simulation study based on the SOEP data confirms good estimation quality.\\
%
%After a short introduction to Bayes Hilbert spaces in Section \ref{chapter_bayes_hilbert_space}, our additive density-on-scalar regression approach is introduced in Section \ref{chapter_regression}. 
We introduce our additive density-on-scalar regression approach in Section \ref{chapter_regression}.
In Section \ref{chapter_subcomp}, we discuss decomposability properties useful for model interpretation and partly also for estimation. 
We model female share distributions based on the SOEP data in Section \ref{chapter_application} and present a simulation study based thereon in Section \ref{chapter_simulation}, before a final discussion in Section \ref{chapter_conclusion}.

% \input{2_bayes_hilbert_space_short.tex} % Included as first subsection of the regression section
%auto-ignore
\section{Density-on-scalar regression}\label{chapter_regression}
% We consider regression models with a density as response and scalar covariates. 
% More precisely, the response is an element of a Bayes Hilbert space \citep{vdb2010}.
To formulate regression models with probability densities~$f$ as response, we will consider $f$ % $\in B^2(\mu)$ 
as an element of a Bayes Hilbert space % with reference measure $\mu$ 
\citep{vdb2014}. 
Thus, we first briefly introduce Bayes Hilbert spaces in Section~\ref{chapter_bayes_hilbert_space}, before formulating our structured additive regression models therein in Section~\ref{chapter_regression_model} and presenting our boosting algorithm for estimation in Section~\ref{chapter_estimation_bayes}.
% Note that a more detailed discussion of Bayes Hilbert spaces including proofs is provided in Appendix~\ref{appendix_bayes_hilbert_space} taking a slightly different point of view compared to \citet{vdb2010, vdb2014}.

%auto-ignore
\subsection{The Bayes Hilbert space}\label{chapter_bayes_hilbert_space}

%Kürzpotenzial: Subsections raus

% To formulate a regression model for probability densities $y$ as response, we will consider $y\in B^2(\mu)$ element of a Bayes Hilbert space \citep{vdb2014} with reference measure $\mu$. It 
A Bayes Hilbert space $\B$ is constructed somewhat analogously to $L^2(\mu)$, % spaces, 
but built on the alternative vector space structure of Bayes spaces \citep{vdb2010} grounded on relative rather than absolute differences. 
An isomorphism $\clr: B^2(\mu) \mapsto L^2_0(\mu)$ to the closed subspace $L^2_0(\mu) = \{ \ft \in L^2(\mu) \mid \int \ft \, \dmu = 0 \} \subset L^2(\mu)$ of square integrable functions integrating to zero allows carrying out many computations effectively in $L^2(\mu)$.
% We briefly introduce Bayes Hilbert spaces in the following.
The formal construction is summarized in the following.
More detailed discussion and proofs are provided in appendix~A. %\ref{appendix_bayes_hilbert_space}, 
% taking a slightly different point of view compared to \citet{vdb2010, vdb2014}.

%\subsection{Introducing the Bayes space}\label{chapter_bayes_space_definition}

Let $(\Tcal, \Acal)$ be a measurable space and $\mu$ a finite measure on it. %, the so-called \emph{reference measure}.
E.g., for income share distributions analyzed in Section~\ref{chapter_application}, consider $\Tcal = [0,1]$, $\Acal$ its Borel $\sigma$-algebra, and $\mu = \lambda + \delta_0 + \delta_1$ with $\lambda$ the Lebesgue measure and $\delta_t$, $t\in\Tcal$, Dirac measures at $t$. 
In the set $\Mcal(\mu) = \Mcal(\Tcal, \Acal, \mu)$ of $\sigma$-finite measures % equivalent to $\mu$, i.e., all measures 
with the same null sets as $\mu$, each measure possesses a $\mu$-almost every\-where ($\mu$-a.e.) positive and unique density $f$ with respect to $\mu$ (Radon-Nikodym derivative). For simplicity, we identify measures in $\Mcal(\mu)$ with their $\mu$-densities.
This notion of densities does not imply a fixed integral of one.
However, considering two densities $f_1, f_2 \in \Mcal (\mu)$ equivalent if they are % $\mu$-a.e. 
proportional, $f_1 \propto f_2$, i.e., if there is a $c > 0$ with $f_1 = c \, f_2$ (here and in the following, pointwise identities have to be understood $\mu$-a.e.), % $f_1 (t) = c \,f_2 (t)$ for $\mu$-almost every $t \in \Tcal$, %including also improper densities. %, where $c \, (+ \infty) = + \infty$.
%Provided the equivalence class of $f$ with respect to $\propto$ contains a finite measure, 
in practice, we choose the probability density $f / \int_\Tcal f \, \dmu$ as representative of a $\propto$-equivalence class (if possible). % with respect to $\propto$, whenever it contains densities with finite integral.
%We now consider the set $\Bcal(\mu) = \Bcal(\Tcal, \Acal, \mu)$ of equivalence classes with respect to $\propto$. 
The set $\Bcal(\mu) = \Bcal(\Tcal, \Acal, \mu)$ of $\propto$-equivalence classes, % with respect to~$\propto$ 
called the \emph{Bayes space (with reference measure $\mu$)}, is a real vector space with addition $\oplus$ and scalar multiplication $\odot$ defined as
%\begin{align*}
%	f_1 \oplus f_2 := f_1 \, f_2 \quad \text{\emph{(perturbation)}} \qquad 
%	\text{and} \qquad \alpha \odot f := (f)^\alpha  \quad \text{\emph{(powering)}}
%\end{align*}
$f_1 \oplus f_2 := f_1 \, f_2$ (\emph{perturbation}) and $\alpha \odot f_1 := (f_1)^\alpha$ (\emph{powering}) for $f_1, f_2 \in \Bcal (\mu)$ and $\alpha \in \Rbb$.%, where $\oplus$ is the as addition and $\odot$ is the scalar multiplication. 
\footnote{We do not distinguish $f\in \Mcal(\mu)$ and its equivalence class $[f]\in \Bcal(\mu)$ in the notation, denoting both by $f$ in the following, but clarify its use whenever not clear from the context.}
To obtain probability densities, resulting representatives have to be re-normalized. 
% Here, the definitions do not depend on the representative and can be directly understood for equivalence classes.
The additive neutral element $0_\Bcal \in \Bcal(\mu)$ is the equivalence class of constant functions % constant $1$ function 
(containing the density of $\mu$), the additive inverse element is $\ominus f := 1 / f$, and the multiplicative neutral element is $1 \in \Rbb$.
For subtraction, we write $f_{1} \ominus f_{2} := f_{1} \oplus (\ominus f_{2})$. 
% Note that up to this point, assuming $\mu$ to be $\sigma$-finite would suffice, but a finite reference measure is needed to proceed.
%
%\subsection{Introducing the Bayes Hilbert space}\label{chapter_bayes_hilbert_space_definition}

%To also From now on, we restrict the reference measure $\mu$ to be finite, progressing to Bayes Hilbert spaces.
%This is similar to \citet{vdb2014} with some details different.
Analogously to $L^p$ spaces, $B^p$ spaces for $1 \leq p < \infty$ are defined as
$ %\begin{align*}
B^p(\mu) 
= B^p(\Tcal, \Acal, \mu) 
% große Klammern
%:= \Bigl\{ f \in \Bcal (\mu) ~\Big|~ \int_{\Tcal} \bigl| \log f\bigr|^p \, \dmu < \infty
:= \{ f \in \Bcal (\mu) ~|~ \int_{\Tcal} | \log f |^p \, \dmu < \infty \}.
$ %\end{align*}
Since $f \in B^p(\mu)$ is equivalent to $\log f \in L^p(\mu)$,
%Or, put differently, we have $\ft \in \lpl$, iff $\exp \ft \in \bpl$. Therefore, we have
we have $B^q(\mu) \subset B^p(\mu)$ for $p,q \in \Rbb$ with $1\leq p<q$.
%\begin{prop}\label{Bp_vectorspace}
Note that for every $p \in \Rbb$ with $1 \leq p < \infty$, the space $\bpl$ is a vector subspace of $\Bcal (\mu)$, see \citet{vdb2014}. %[Proposition 1]
%\end{prop}
The \emph{centered log-ratio (clr) transformation of $f \in B^p(\mu)$} %or \emph{clr image of $\nu$ (with reference measure $\mu$}) 
is
\begin{align}
\clr_{B^p(\Tcal, \Acal, \mu)} [f] 
:= \log f - \Scal_{B^p(\Tcal, \Acal, \mu)} (f), \label{def_clr} % \frac1{\mu (\Tcal)} \, \int_{\Tcal} \log f \, \dmu , 
\end{align}
with $\Scal_{B^p(\Tcal, \Acal, \mu)} (f) := 1 / \mu (\Tcal) 
\, \int_{\Tcal} \log f \, \dmu$ the mean logarithmic integral.
%Note that it is linear on $B^2(\Tcal, \Acal, \mu)$, which is straightforward to show.
We omit the indices $B^p(\Tcal, \Acal, \mu)$ 
or shorten them to $\mu$ or $\Tcal$, if the underlying space is clear from context. 
%Now, consider $\lpnl := \left\{ \ft \in \lpl ~|~ \int_{\Tcal} \ft \, \dmu = 0 \right\}$, which is a closed subspace of $\lpl$.

\begin{prop}[For $p = 1$ in \citealp{vdb2014}]
% Propositions 2, 4 and 5
\label{clr_isomorphism}
For $1 \leq p < \infty$, %the clr trans\-formation 
$\clr : \bpl \ra \lpnl$ is an isomorphism with inverse $\clr^{-1} [\ft] = \exp \ft$.
\end{prop}

%Note that $\lpnl$ is a closed subspace of $\lpl$.
The space $B^2(\mu)$ with inner product 
%\begin{align*}
%	&\langle f_{1} , f_{2} \rangleb := \int_{\Tcal} \clr [f_{1}] \, \clr [f_{2}] \, \dmu, \qquad f_1,f_2 \in B^2(\mu),
%\end{align*}
$ \langle f_{1} , f_{2} \rangleb := \int_{\Tcal} \clr [f_{1}] \, \clr [f_{2}] \, \dmu$, where $f_1,f_2 \in B^2(\mu)$,
is called the \emph{Bayes Hilbert space (with reference measure $\mu$)} and indeed is a Hilbert space \citep{vdb2014}.
% This is indeed an inner product, see appendix~\ref{appendix_proofs_BHS}. Furthermore, 
The induced norm on $\B$ is $\Vert f \Vertb := ({\langle f, f \rangleb})^{1/2}$.
By definition, we have
$ % \begin{align*}
\langle f_1, f_2 \rangleb = \langle \clr [f_1] , \clr [f_2] \rangle_{L^2(\mu)},
$ % \end{align*}
which immediately implies that $\clr: \B \ra \Ln$ is isometric.

Bayes Hilbert spaces enable a variety of different applications.
Usually, $\Tcal \subset \Rbb$ with three common cases: % continuous, discrete and mixed.
The \emph{continuous case} denotes $\Tcal = I$ being a nontrivial interval with $\Acal = \mathfrak{B}$ the Borel $\sigma$-algebra restricted to $I$ and $\mu = \lambda$ the Lebesgue measure. % (on~$\mathfrak{B}$).
The \emph{discrete case} refers to $\Tcal = \Dcal := \{ t_1, \ldots , t_D\} \subset \Rbb, ~ D \in \Nbb$,
%, t_i < t_j$ for $i <  j$, 
with $\Acal = \mathcal{P}(\Tcal)$ the power set of $\Dcal$ and $\mu = \delta := \sum_{d = 1}^D w_d \, \delta_{t_d}$ a weighted sum of Dirac measures, where $w_d > 0$. % for $d = 1, \ldots , D$.
% , where $\delta_t$ denotes the Dirac measure at $t$ and $w_d > 0$ are weights for $d = 1, \ldots , D$.
The \emph{mixed case} is a mixture of both, with $\Tcal = I \cup \Dcal$, % (usually $\Dcal \subset I$), 
$\Acal$ being the smallest $\sigma$-algebra containing all closed subintervals of $I$ and all points of $\Dcal$ ($\Acal = \mathfrak{B}$ if $\Dcal \subset I$), 
%The corresponding reference measure is a mixture of weighted Dirac measures % at $t_1, \ldots , t_D$ 
%and the Lebesgue measure, i.e.,
and $\mu = \delta + \lambda$. % $\sum_{d = 1}^D w_d \, \delta_{t_d} + \lambda$.
Note that the mixed case contains the continuous and discrete cases as special cases, allowing either $\Dcal = \emptyset$ or $I = \emptyset$.
Our application in Section~\ref{chapter_application} gives an example for the mixed case.
%Note that the continuous case is a special case of the mixed case with $\Dcal = \emptyset$.
The corresponding Bayes Hilbert spaces are also denoted as \emph{continuous}, \emph{discrete}, or \emph{mixed}.

Note that due to the construction of Bayes Hilbert spaces, %restriction of $\mu$ being finite on $(\Tcal, \Acal)$ 
$\lambda$ is no valid reference measure for densities on % the whole real line 
$\Tcal = \Rbb$ (with Borel $\sigma$-algebra $\mathfrak{B}_{\Rbb}$).
The probability measure corresponding to the standard normal distribution is an alternative \citep{vdb2014}.
Furthermore, Bayes Hilbert spaces only contain positive densities. % which are exactly zero in parts of~$\mathcal{T}$.
If a density is not directly observed but estimated from an observed sample, density values of zero can be avoided by choosing a density estimation method that yields a positive density. %, e.g., kernel density estimation with positive kernel functions in the continuous case. % like beta kernels as we use in our application in Section~\ref{chapter_application}. 
For discrete sets $\Tcal$, one option is to replace observed density values of zero with small values % as long as they are rounded or counting zeros, 
(e.g., \citealp{pawlowsky2015}).

\subsection{Regression model}\label{chapter_regression_model} % früher: chapter_function_on_scalar
% Density-on-scalar regression is motivated by function-on-scalar regression. % (for modeling the conditional expectation of the response), see, e.g., \citealp{brh2015}. %, and is constructed in an analogous manner. 
% Both regression types are closely related (at least in the continuous case), as density-on-scalar models can be transformed to function-on-scalar models in $\Ln$ via the clr transformation.
Density-on-scalar regression is motivated by and (at least in the continuous case)  closely related to function-on-scalar regression % (for modeling the conditional expectation of the response), see, e.g., \citealp{brh2015}. %, and is constructed in an analogous manner. 
as the clr transformation of density-on-scalar models yields function-on-scalar models in $\Ln$.
%Thus, we introduce density-on-scalar regression emanating from function-on-scalar regression.
% We formulate our model analogously to structured additive function-on-scalar regression models \citep{brh2015}, considering densities in a Bayes Hilbert space $\B$ instead of functions in $L^2(I, \mathfrak{B}, \lambda)$ and using the corresponding operations.
% %
% %For data pairs~$(\tilde{f}_i, \xf_i) \in L^2(\lambda) \times \Rbb^K, ~K \in \Nbb, ~i = 1, \ldots , N, ~ N \in \Nbb$, consider the structured additive (function-on-scalar) regression model
% %\begin{align}
% %\tilde{f}_i = \tilde{h}(\xf_i) + \tilde{\varepsilon}_i = \sum_{j=1}^J \tilde{h}_j (\xf_i) + \tilde{\varepsilon}_i, \label{modelgleichung_fm} 
% %% widetilde for partial effects
% %%\tilde{f}_i = \widetilde{h(\xf_i)} + \tilde{\varepsilon}_i = \sum_{j=1}^J \widetilde{h_j (\xf_i)} + \tilde{\varepsilon}_i, \label{modelgleichung_fm} 
% %\end{align}
% %where $\tilde{\varepsilon}_i \in L^2(\lambda)$ are functional error terms with $\Ebb(\tilde{\varepsilon}_i) = 0$ and $\tilde{h}_j(\xf_i) \in L^2(\lambda)$ are partial effects, $J \in \Nbb$.
% %%All sides of this equation belong to the Hilbert space~$L^2(\lambda)$. 
% %Before enlarging on these, we formulate an analogous density-on-scalar model by replacing $L^2(\lambda)$ with a Bayes Hilbert space $\B$ and using the related operations. % $\oplus$ and $\odot$.
% %
% For data pairs~$(f_i, \xf_i) \in \B \times \Rbb^K, ~K \in \Nbb, ~i = 1, \ldots , N, ~ N \in \Nbb$, this yields the structured additive density-on-scalar regression model
%
Thus, analogously to the function-on-scalar model of \citet{brh2015}, where the response functions are elements of $L^2(I, \mathfrak{B}, \lambda)$, for data pairs~$(f_i, \xf_i) \in \B \times \Rbb^K, ~K \in \Nbb, ~i = 1, \ldots , N, ~ N \in \Nbb$,
we consider the structured additive density-on-scalar regression model
\begin{align}
f_i = h(\xf_i) \oplus \varepsilon_i = \bigoplus_{j=1}^J h_j (\xf_i) \oplus \varepsilon_i, \label{modelgleichung} %\label{modelgleichung_alternativ}
\end{align}
where $\varepsilon_i \in \B$ are functional error terms with $\Ebb(\varepsilon_i) = 0_\Bcal \in \B$ and  $h_j(\xf_i) \in \B$ are $J \in \Nbb$ partial effects.
The expectations of the $\B$-valued random elements~$\varepsilon_i$ are defined via the Bochner integral (e.g., \citealp{hsing2015}). 
Each partial effect $h_j(\xf_i) \in \B$ % in~(\ref{modelgleichung}) 
models an effect of none, one or more covariates in $\xf_i$ and thus $J \neq K$ in general. % = (x_{i,1} , \ldots , x_{i, K}). 

%\ifnum\value{jasa}=1
%{
%%  \spacingset{1}
%\begin{table}[h]
%\caption{Partial effects for density-on-scalar regression.% Interactions of the given effects are possible as well. 
%}
%\begin{center}
%\vspace{-0.5cm}
%\begin{tabular}{lll}
%% \hline \hline
%Covariate(s) & Type of effect & $h_j(\xf)$\\
%\hline
%%\rule{0pt}{13pt}
%%\hspace{-0.28cm} 
%None & Intercept & $\beta_0$ \\
%One scalar covariate $x$ & Linear effect & $x \odot \beta$ \\
% & Flexible effect & $g(x)$ \\
%Two scalar covariates $x_1, x_2$ & Linear interaction & $x_1 \odot (x_2 \odot \beta)$ \\
% & Functional varying coefficient & $ x_1 \odot g(x_2)$ \\
% & Flexible interaction & $g(x_1, x_2)$ \\[0.15cm]
% %
%Grouping variable $k$ & Group-specific intercepts & $\beta_k$ \\
%Grouping variable $k$ and scalar $x$ & Group-specific linear effects & $x \odot \beta_k$ \\
% & Group-specific flexible effects & $g_k(x)$ \\
%% \hline
%\end{tabular}
%\end{center}
%\label{possible_effects}
%\end{table}
%} \else
\ifnum\value{aoas}=1
{
%  \spacingset{1}
\begin{table}[h]
\caption{Partial effects for density-on-scalar regression ($x$ denoting scalar covariates, $\beta$ and $g(\,)$ densities in $\B$).% Interactions of the given effects are possible as well. 
\label{possible_effects}
}
\begin{center}
\begin{tabular}{lll}
\hline
Covariate(s) & Type of effect & $h_j(\xf)$\\
\hline
%\rule{0pt}{13pt}
%\hspace{-0.28cm} 
None & Intercept & $\beta_0$ \\
One scalar covariate $x$ & Linear effect & $x \odot \beta$ \\
 & Flexible effect & $g(x)$ \\
Two scalar covariates $x_1, x_2$ & Linear interaction & $x_1 \odot (x_2 \odot \beta)$ \\
 & Functional varying coefficient & $ x_1 \odot g(x_2)$ \\
 & Flexible interaction & $g(x_1, x_2)$ \\[0.15cm]
Grouping variable $k$ & Group-specific intercepts & $\beta_k$ \\
Grouping variable $k$ and scalar $x$ & Group-specific linear effects & $x \odot \beta_k$ \\
 & Group-specific flexible effects & $g_k(x)$ \\
 \hline
\end{tabular}
\end{center}
\end{table}
} \else
{
\begin{table}[h]
\begin{center}
\begin{tabular}{l|l|l}
Covariate(s) & Type of effect & $h_j(\xf)$\\
\hline
%\rule{0pt}{13pt}
%\hspace{-0.28cm} 
None & Intercept & $\beta_0$ \\
One scalar covariate $x$ & Linear effect & $x \odot \beta$ \\
 & Flexible effect & $g(x)$ \\
Two scalar covariates $x_1, x_2$ & Linear interaction & $x_1 \odot (x_2 \odot \beta)$ \\
 & Functional varying coefficient & $ x_1 \odot g(x_2)$ \\
 & Flexible interaction & $g(x_1, x_2)$ \\
\hline
Grouping variable $k$ & Group-specific intercepts & $\beta_k$ \\
Grouping variable $k$ and scalar $x$ & Group-specific linear effects & $x \odot \beta_k$ \\
 & Group-specific flexible effects & $g_k(x)$ \\
\end{tabular}
\end{center}
\caption{Partial effects for density-on-scalar regression: Scalar covariates are denoted by $x$, densities in $\B$ by $\beta$ and $g(\,)$.% Interactions of the given effects are possible as well. 
\label{possible_effects}}
\end{table}
\vspace{-0.4cm}
} \fi

Table~\ref{possible_effects} gives an overview of possible partial effects, inspired by Table~1 in \citet{brh2015}.
The upper part shows effects for up to two different scalar covariates. 
In the lower part, group-specific effects for categorical variables are presented.
%the lower part group-specific effects for categorical variables.
Interactions of the given effects are possible as well.
% Scalar covariates are denoted by $x$, densities in $\B$ by $\beta$ and $g(\,)$.
%
Note that constraints are necessary to obtain identifiable models. 
For a model with an intercept $\beta_0$, this is obtained by centering the partial effects:
\begin{align}
\frac1{N} \odot \bigoplus_{i=1}^N h_j(\xf_i) = 0. \label{identifiability_constraint}
\end{align}
This constraint can be included based on \citet[Section 1.8.1]{wood2017} as in appendix~A of \citealp{brh2015} for function-on-scalar regression models.
Similarly, interaction effects can be centered around the main effects \citep[see appendix~A of][]{stoecker2019}.
% More details about how to include this constraint in a functional linear array model for function-on-scalar regression can be found in appendix~A of \citet{brh2015}.
% A similar procedure can be used to obtain a centering of interaction effects around the main effects, see appendix~A of \citet{stoecker2019}.
% Both approaches are based on \citet[Section 1.8.1]{wood2017} and can be transferred straightforwardly to density-on-scalar regression.

\subsection{Estimation by Gradient Boosting}\label{chapter_estimation_bayes}

To estimate the function $h (\xf_i) \in \B$ in Equation~(\ref{modelgleichung}), the aim is to minimize the sum of squared errors
\begin{align}
\SSE(h) := 
\sum_{i=1}^N \Vert \varepsilon_i \Vertb^2 = 
\sum_{i=1}^N \Vert f_i \ominus h(\xf_i) \Vertb^2
= \sum_{i=1}^N \rho_{f_i}\left(h(\xf_i)\right).
\label{sse}
\end{align}
Here, $\rho_{f_i}: \B \ra \Rbb, \, f \mapsto \Vert f_i \ominus f \Vertb^2$ is the quadratic loss functional.
%To simplify the minimization problem and to determine the type of an effect, compare Table~\ref{possible_effects}, 
We consider a basis representation for each partial effect:
\begin{align}
h_j(\xf_i) = \left( \bfe_j(\xf_i) \ootimes \bfe_Y \right)^\top \thetaf_j 
= \bigoplus_{n=1}^{K_j} \bigoplus_{m=1}^{K_Y} b_{j, n}(\xf_i) \odot b_{Y, m} \odot \theta_{j, n, m}
, \label{basisfunktionen_kronecker}
\end{align}
where $\bfe_j = (b_{j, 1}, \ldots , b_{j, K_j})^\top: \Rbb^K \ra \Rbb^{K_j}$ is a vector of basis functions describing the covariate effect, e.g., splines for smooth non-linear effects, and $\bfe_Y = ( b_{Y, 1}, \ldots, b_{Y, K_Y})^\top$ $\in \B^{K_Y}$ is a vector of basis functions in the response space.
A suitable choice of this tensor product basis thus allows to linearize flexible covariate effects on the response densities.
With~$\ootimes$, we denote the Kronecker product of a real-valued %matrix 
with a $\B$-valued matrix.
It is defined like the Kronecker product %$\otimes$ 
of two real-valued matrices, using %powering 
$\odot$ instead of the usual multiplication.
Similarly, matrix multiplication of a real-valued %matrix 
with a $\B$-valued matrix is defined % like the usual matrix multiplication, replacing sums with perturbations $\oplus$ and products with powerings $\odot$.
by replacing sums with %perturbations 
$\oplus$ and products with %powerings 
$\odot$ in the usual matrix multiplication.
Our goal is to estimate the coefficient vector~$\thetaf_j = (\theta_{j, 1, 1}, \ldots , %\theta_{j,1,K_Y}, \ldots , 
\theta_{j, K_j, K_Y}) \in \Rbb^{K_j \, K_Y}$.
To allow sufficient flexibility for $h_j$, the product $K_j \, K_Y$ can be chosen to be large.
The necessary regularization can then be accomplished with a Ridge-type penalty term $\thetaf_j^\top \Pf_{j,Y}\thetaf_j$. 
For a basis representation as in equation~(\ref{basisfunktionen_kronecker}), an anisotropic penalty matrix % can be used~\citep[Section~5.6]{wood2017}:
$
\Pf_{j,Y} = \lambda_j (\Pf_j \otimes \If_{K_Y}) + \lambda_Y (\If_{K_j} \otimes \Pf_Y)
%\label{penalty}
$
can be used. % \citep{brh2020}.
%Here, $\Pf_j \in \Rbb^{K_j \times K_j}$ is a suitable penalty matrix for $\bfe_j$, while $\Pf_Y \in \Rbb^{K_Y \times K_Y}$ is a suitable penalty matrix for $\bfe_Y$ 
Here, $\Pf_j \in \Rbb^{K_j \times K_j}$ and $\Pf_Y \in \Rbb^{K_Y \times K_Y}$ are suitable penalty matrices for $\bfe_j$ and $\bfe_Y$, respectively,
and $\lambda_j,\lambda_Y \geq 0$ are smoothing parameters in the respective directions. 
Alternatively, a simplified isotropic penalty matrix
$
\Pf_{j,Y} = \lambda_j ( (\Pf_j \otimes \If_{K_Y}) + (\If_{K_j} \otimes \Pf_Y) ) % \label{penalty2}
$
with only one smoothing parameter is possible \citep{brh2020}.
The penalized basis representation allows for very flexible modeling of effects, in analogy to established additive models for scalar data \citep{wood2017}.

%To fit model~(\ref{modelgleichung}), we use a component-wise gradient boosting algorithm, where the expected loss is minimized step-wise along the steepest gradient descent, based on~\citet{buehlmann2007}, which was modified by \citet{brh2015} for functional regression using \citet{hothorn2014}. 
% We now adapt the algorithm for density-on-scalar regression.
We fit model~(\ref{modelgleichung}) using a component-wise gradient boosting algorithm, where the (empirical) expected loss is minimized step-wise along the steepest gradient descent.
It is an adaption of the algorithm presented in \citet{brh2015}, which was built on that in \citet{hothorn2014}.
Advantages of this approach are that it can deal with a large number of covariates, it performs variable selection, and includes regularization. \citet{Buehlmann2003L2boosting} discuss theoretical properties of gradient boosting with respect to sum of squares errors, which is typically referred to as $L^2$-Boosting, for scalar responses. They show -- simplifying to a single learner -- that bias decays exponentially fast while estimator variance increases in exponentially small steps over the boosting iterations, which supports the general practice of stopping the algorithm early before it eventually reaches the standard (penalized) least squares estimate. \citet{LutzBuehlmann2006BoostingHighMultivariate} show consistency of component-wise $L^2$-Boosting for linear regression with both high-dimensional multivariate response and predictors. Similar to these predecessors, our $L^2$-Boosting algorithm for Bayes Hilbert spaces simplifies to repeated re-fitting of residuals -- which, however, present densities in our case.
%In this work, we model the conditional expectation, which corresponds to minimizing the quadratic loss. 
%However, boosting can also handle other loss functions, which enable to model, e.g., the median or some quantile. 
%See~\citet{brh2015} for a general framework with an arbitrary loss function for functional regression.

\ifnum\value{jasa}=1 {\vspace{-0.3cm}} \else{} \fi
\subsubsection*{Algorithm: Bayes space $L^2$-Boosting for density-on-scalar models}
\begin{enumerate}
\item\ifnum\value{jasa}=1 {\vspace{-0.28cm}} \else{} \fi %\vspace{-0.28cm}
Select vectors of basis functions~$\bfe_Y, \bfe_j$, the starting coefficient vector $\thetaf_j^{[0]} \in \Rbb^{K_j \, K_Y}$, % see~(\ref{basisfunktionen_kronecker}), 
and penalty matrices~$\Pf_{j,Y}, j = 1, \ldots, J$. 
% Furthermore, define integration weights $\Delta_i(t_g)$ for $i = 1, \ldots , N$ and $g = 1, \ldots , G$. 
Choose the step-length $\kappa \in (0,1)$ and the stopping iteration $m_{\text{stop}}$ and set the iteration number $m$ to zero.
We comment on a suitable selection of these quantities below.
\item\ifnum\value{jasa}=1 {\vspace{-0.28cm}} \else{} \fi %\vspace{-0.28cm}
Calculate the negative gradient of the empirical risk with respect to the Fr\'{e}chet differential (see appendix~B % \ref{appendix_proofs} 
for the proof of this equation)
\begin{align}
U_i := 
\ominus \nabla \rhof (f) \big|_{{f = \hh^{[m]}(\xf_i)}}
= 2 \odot \left( f_i \ominus \hh^{[m]}(\xf_i)\right) %\in \B 
, \label{gradient}
\end{align}
where $\hh^{[m]}(\xf_i) = \bigoplus_{j=1}^J \left( \bfe_j(\xf_i)^\top \ootimes \bfe_Y^\top \right) \thetaf_j^{[m]}$. 
Fit the base-learners
\begin{align}
\gammafh_j 
&= \argmin_{\gammaf \in \Rbb^{K_j \, K_Y}} \sum_{i=1}^N \bigl\Vert U_i \ominus\left( \bfe_j(\xf_i)^\top \ootimes \bfe_Y ^\top \right) \gammaf \bigr\Vert_{\B}^2 + \gammaf^\top \Pf_{jY} \gammaf  \label{fit_baselearner}
\end{align}
for $j = 1, \ldots , J$ and select the best base-learner
\begin{align}
j^\ast 
%= \argmin_{j = 1, \ldots , J} \sum_{i=1}^N \sum_{g=1}^G \Delta_i(t_g) \, \left( U_i(t_g) - \left( \bfe_j(\xf_i)^\top \otimes \tilde{\bfe}_Y (t_g) ^\top \right) \gammafh_j\right)^2 \\
= \argmin_{j = 1, \ldots , J} \sum_{i=1}^N \bigl\Vert U_i \ominus\left( \bfe_j(\xf_i)^\top \ootimes \bfe_Y ^\top \right) \gammafh_j \bigr\Vert_{\B}^2.  \label{select_baselearner}
\end{align}
\item\ifnum\value{jasa}=1 {\vspace{-0.28cm}} \else{} \fi %\vspace{-0.28cm}
The coefficient vector corresponding to the best base-learner is updated, the others stay the same: $\thetaf_{j^\ast}^{[m+1]} := \thetaf_{j^\ast}^{[m]} + \kappa \, \gammafh_{j^\ast}, ~\thetaf_j^{[m+1]} := \thetaf_j^{[m]} \quad \text{for}~j \neq j^\ast$.
\item\ifnum\value{jasa}=1 {\vspace{-0.28cm}} \else{} \fi %\vspace{-0.28cm}
While $m<m_{\text{stop}}$, increase $m$ by one and go back to step~2. Stop otherwise.
\end{enumerate}

The resulting estimator of model~(\ref{modelgleichung}) is
%\begin{align*}
$
\hat{f}_i = \hat{\Ebb}(f_i ~|~ %\Xf = 
\xf_i) = \bigoplus_{j=1}^J \hh_j^{[m_{\text{stop}}]} (\xf_i),
$
%\end{align*}
with $\hh_j^{[m_{\text{stop}}]} (\xf_i) = ( \bfe_j(\xf_i)^\top \ootimes \bfe_Y ^\top ) \thetaf_j^{[m_{\text{stop}}]}$.
%Note that $\rhoy$ is Fr\'echet differentiable at $\fnu \in \B$ and $\B$ is a Hilbert space. Thus, the gradient in step~2 exists. 
%We prove this as well as the equality in Equation~(\ref{gradient}) in appendix~\ref{appendix_proofs_regression}.
%We prove Equation~(\ref{gradient}) in appendix~\ref{appendix_proofs_regression}.
%
In the following, we discuss the selection of parameters in step 1, see also \citet{brh2015, brh2020}.

The choice of vectors of basis functions % $\bfe_Y$ and 
$\bfe_j$ and penalty matrices % $\Pf_Y$ and 
$\Pf_j$ depends on the desired partial effect $h_j(\xf)$. 
% Regarding the covariate basis functions $\bfe_j$, a 
A suitable choice for flexible nonlinear effects is, e.g., B-splines with a difference penalty. %, e.g., second order differences. 
For a linear effect of one covariate, % the vector of basis functions is chosen as 
set $\bfe_j = (1, \id): \Rbb \ra \Rbb^2, ~x \mapsto (1, x)$, yielding the design matrix of a simple linear model, with, e.g., a Ridge penalty, $\Pf_j = \Id_2$.
%The resulting matrices $\bfe_j(\xf_i)$ and $\Pf_j$ have to be transformed to respect the sum-to-zero constraint~(\ref{identifiability_constraint}) as described in online appendix~A of~\citet{brh2015}.
% Penalty Matrizen müssen ebenfalls transformiert werden! Z.B. $\Zf_j^\top \Pf_j \Zf_j$ statt $\Pf_j$ siehe Online appendix~A in \citet{brh2015}
% Für bolsc in Timeformula sind das dann die Gitterpunkte (in der 2. Spalte der Matrix; die erste enthält nur 1en, vgl. LM); Penalty-Matrix ist die Identität, was einer Ridge-Penalisierung entspricht
%
%For $\bfe_Y \in \B^{K_Y}$, the selection depends on the measure space $(\Tcal, \Acal, \mu)$. 
A basis $\bfe_Y \in \B^{K_Y}$ can be obtained from a suitable basis $\bar{\bfe}_Y \in L^2(\mu)^{K_Y + 1}$ % as follows. 
% Transforming $\bar{\bfe}_Y$ % from $L^2(\mu)^{K_Y + 1}$ 
by transforming $\bar{\bfe}_Y$
to $\Ln^{K_Y}$ %yields a basis $\tilde{\bfe}_Y \in \Ln^{K_Y}$.
%Note that the penalty matrix has to be transformed accordingly, as well.
%The respective transformation matrix is constructed in appendix~\ref{appendix_constraint}.
(see appendix~C % \ref{appendix_constraint} 
for % the construction of the transformation matrix)
details)
%The approach to obtain the transformation matrix for this purpose is similar to the one to ensure the constraint in~(\ref{identifiability_constraint}) and is derived in appendix~\ref{appendix_constraint}.
and applying the inverse clr transformation component-wise. % on each component of $\tilde{\bfe}_Y$ gives the desired basis $\bfe_Y$.
%Now, we discuss the selection of basis functions $\bar{\bfe}_Y \in \Ltwo^{K_Y + 1}$ (yielding $\bfe_Y\in \B^{K_Y}$) for the continuous and discrete cases as introduced at the beginning of this section. 
%
For the continuous case, a reasonable choice for $\bar{\bfe}_Y \in L^2(\lambda)^{K_Y + 1}$ is a B-spline basis with a difference penalty, allowing flexible modeling of the response densities. 
For the discrete case, a suitable selection is $\bar{\bfe}_Y = (\mathbbm{1}_{\{t_1\}}, \ldots, \mathbbm{1}_{\{t_D\}}) \in L^2 (\sum_{d=1}^D w_d \, \delta_{t_d})^D$, where $\mathbbm{1}_{A}$ is the indicator function of $A \in \Acal$.
Again, a difference penalty can be used to control variability of the estimates, if smoothness across $t_1, \ldots , t_D$ is a reasonable assumption.
% This is motivated by the fact that $\bar{\bfe}_Y$ is equivalent to the design matrix received by a B-spline basis of degree 1 with knots $t_1, \ldots , t_D$, evaluated at the knots, which is the identity matrix $\If_D$.
The mixed case is not as straightforward. We show in Section~\ref{chapter_mixed_model} that it can be decomposed into a continuous and a discrete component.
I.e., it is not necessary to explicitly select basis functions $\bfe_Y \in B^2(\mu)^{K_Y}$ for the mixed case, as they can be obtained by concatenating the basis functions of the continuous and the discrete components.

%An important point when choosing the penalty matrices is the selection of the smoothing parameters. 
Selecting the smoothing parameters is also important for regularization.
They are specified such that the degrees of freedom  are equal for all base-learners, to ensure a fair base-learner selection in each iteration of the algorithm. 
Otherwise, selection of more flexible base-learners is more likely than that of less flexible ones, see~\citet{hofner2011}.
However, the effective degrees of freedom of an effect after $m_{\text{stop}}$ iterations will in general differ from those preselected for the base learners in each single iteration.
They are successively adapted to the data.
The starting coefficient vectors $\thetaf_j^{[0]}$ are usually all set to zero, enabling variable selection as an effect that is never selected stays at zero. % $0 \in \B$.
Like in functional regression, a suitable offset can be used for the intercept to improve the convergence rate of the algorithm, e.g., the mean density of the responses in $\B$.
Note that a constant scalar offset, which is another common choice in functional regression, equals zero $0_\Bcal$ in the Bayes Hilbert space and thus corresponds to no offset.
The optimal number of boosting iterations $m_{\text{stop}}$ can be found with cross-validation, sub-sampling or bootstrapping, with samples generated on the level of elements of $\B$. 
The early-stopping avoids overfitting.
Finally, the value $\kappa = 0.1$ for the step-length is suitable in most applications for %our case of a 
a quadratic loss function \citep{brh2020}. 
A smaller step-length usually requires a larger value for~$m_{\text{stop}}$. While the in-bag risk reduction provides a variable importance measure, further validation out-of-sample is straight-forwardly possible via an outer cross-validation or bootstrap. 
\\%[0.2cm]
Note that the estimation problem can also be solved in $\Ln$ based on the clr transformed model, with the estimates in $\B$ obtained by applying the inverse clr transformation, as proposed by~\citet{talska2017} for functional linear models on closed intervals. 
For our functional additive models, gradient boosting can be performed in $\Ln$ % based on the clr transformed model 
analogously to the algorithm described above.
The results of both algorithms are equivalent via the clr transformation, which we show in appendix~D. % \ref{appendix_boosting_equivalence}.
In the continuous case, this yields the functional boosting algorithm of \citet{brh2015} with the modification that the basis functions $\bfe_Y$ are constrained to be elements of $L^2_0(\lambda)$ instead of $L^2(\lambda)$. % to ensure $h_j(\xf_i) \in L^2_0(\lambda)$ for all $j = 1, \ldots , J$. 
%In general, gradient boosting in $\B$ and $\Ln$ are equivalent:
%As the clr transformation is an isometric isomorphism, applying gradient boosting in $\Ln$ with basis functions $\clr[\bfe_Y]$ yields the clr transformation of the estimates obtained with gradient boosting in the Bayes Hilbert space with basis functions $\bfe_Y$. 
%% I.e., the so-obtained estimates of model~(\ref{modelgleichung_clr}) correspond to the clr transformed estimates of~(\ref{modelgleichung}).

%auto-ignore
\section{Divide and conquer: subcompositional coherence and related properties}
\label{chapter_subcomp}
Understanding the whole density as genuine object of interest is fundamental to object oriented data analysis \citep{marron2021object}. Being able to focus on %coherently reduce it to 
parts of the density in a way coherent with the overall analysis, in analogy to the analysis of subvectors in Euclidean spaces, %in analogy to low dimensional representation of higher dimensional Euclidean vectors in suitable coordinates, 
is however a major advantage for interpretation and potentially for computations. % numerical treatment. 
In this section, we discuss different properties of Bayes Hilbert spaces that allow to focus analysis of densities on selected parts of interest and aid in interpretations. 
%In this section, we develop several properties of Bayes Hilbert spaces, emphasizing that formulating our regression models therein provides an overall consistent framework.
All properties are related to the principle of \emph{subcompositional coherence} \citep[e.g.,][]{pawlowsky2015}, which (translated directly from compositional data analysis) states that 
any analysis of densities $f_1, \dots, f_N \in B^2(\Tcal, \Acal, \mu)$ should be coherent with a corresponding analysis of $f_1|_{\tilde{\Tcal}}, \dots, f_N|_{\tilde{\Tcal}}$ restricted to a subset $\tilde{\Tcal} \in \Acal$ of the domain $\Tcal$. 
From a probabilistic perspective, we may think of the restriction as probability density $f_i( \cdot \mid \tilde{\Tcal}) \propto f_i|_{\tilde{\Tcal}}$ conditional on the event $\tilde{\Tcal}$. %\footnote{Here, we employ `conditional probability' in a straight-forward, non-measure-theoretic notion, restricting to cases where the conditional densities can be computed by re-normalization, such as for $\mu(\tilde{\Tcal})>0$, where we can sensibly consider $f_i|_{\tilde{\Tcal}} \in B^2(\tilde{\Tcal}, \Acal \cap \tilde{\Tcal}, \mu)$, or other standard discrete/continuous cases.}
Accordingly, a probabilistic principle of subcompositional coherence can be phrased as:
\emph{Comparison of two probability distributions conditional on an event $\tilde{\Tcal}$ should not depend on their distribution outside of $\tilde{\Tcal}$.}
This is desirable for at least two reasons: 
1) In many data scenarios, observed and analyzed distributions are in fact restricted to a certain part of a potential set of outcomes due to practical feasibility. 
Their analysis should be compatible with a potential more comprehensive study. 
2) For detailed analysis, one might want to focus on certain aspects, reducing the attention to parts of the domain.
This should be compatible with the whole analysis.  
E.g., in the setting of our application on income share distributions (Section~\ref{chapter_application}),  an analysis only considering double-income households should yield compatible results to an analysis additionally including single-earner households.

In the following, we make more precise in which sense Bayes Hilbert spaces feature subcompositional coherence. 
We show how differences between densities in a Bayes Hilbert space are naturally understood in terms of odds ratios (Section~\ref{chapter_oddsratio_differences}) and how this allows
for local model interpretation (Section~\ref{chapter_interpretation}).
Then, we show how restriction to a subdomain $\tilde{\Tcal}$ can be interpreted as a projection onto a 
subspace (Section \ref{chapter_subcomp_projection}) as in compositional data analysis.
Such a projection is used for decomposing a mixed density into its discrete and continuous parts, discussed in Section~\ref{chapter_mixed_model} and later used to simplify estimation in the analysis of mixed female income share densities in Section \ref{chapter_application}.
All proofs are provided in appendix~B. % \ref{appendix_proofs}.
 
%SOMETHING ABOUT SUBCOMPOSITIONAL DOMINANCE? 
%In particular, this implies that, based on the same reference measure, the distance between conditional densities $\tilde{f}_1$ and $\tilde{f}_2$ two densities $f_1, f_2 \in $, the distance between  

% ---------------------------------------------------------------------------------------------------------------------------

\subsection{Odds ratio interpretation of differences}\label{chapter_oddsratio_differences}
The distance induced by the norm on $\B$ as defined in Section~\ref{chapter_bayes_hilbert_space}
% Instead of using $\clr$-transforms as in Section~\ref{chapter_bayes_hilbert_space}, the distance between two probability densities $f_1,f_2 \in B^2(\mu)$ %of probability measures $\mathbb{P}_1,\mathbb{P}_2$ 
can (similar to \citet{egoz2006}, but written in terms of odds ratios) also be formulated as
\begin{equation*}
	% \label{eq_distance_oddsratio}
	\| f_1 \ominus f_2 \|_{B^2(\Tcal)} = \big( \frac{1}{2 \, \mu(\Tcal)} \int_{\Tcal} \int_{\Tcal} \big(\log \frac{f_1(s) / f_1(t)}{f_2(s) / f_2(t)} \big)^2 \, \dmu(s) \, \dmu(t) \big)^{1/2},
\end{equation*}

which reveals the strong connection of the Bayes Hilbert space geometry and odds ratios. %, well-known for comparison of probabilities. 
The distance essentially aggregates (infinitesimal) odds ratios $\OR(s,t) := \frac{f_1(s) / f_1(t)}{f_2(s) / f_2(t)}$ of
the odds for observing values at $s$ versus at $t$ according to $f_1$ %(or $\mathbb{P}_1$ more specifically) 
over the corresponding odds according to $f_2$. %(or $\mathbb{P}_2$, respectively).
Accordingly, the distance is similarly locally driven to $L^2$-distances, only that it is based on the relation between two points $s$ and $t$.
Due to their relative nature, odds ratios can be easily restricted to $\OR|_{\tilde{\Tcal} \times \tilde{\Tcal}}$ when considering (re-normalized) densities $f_1|_{\tilde{\Tcal}}, f_2|_{\tilde{\Tcal}}$ on a subset $\tilde{\Tcal} \subset \Tcal$.
%(We will make this understanding of odds ratios more precise below.)
%of $\mathbb{P}_1$ over $\mathbb{P}_2$ for , which will be made more precise below. 
As well-established tool for comparison of probabilities, well-known e.g. from logistic regression, odds ratios can thus serve as a key tool for subcompositionally coherent interpretation of differences $f_1 \ominus f_2$ between densities (or probability distributions), also in our application in Section~\ref{chapter_application}, quantifying local differences including %their 
direction.

To make this more precise, we point out the relation of $\OR(s,t)$ to usual odds ratios formulated for probabilities rather than densities, 
where $\mathbb{P}_1$ and $\mathbb{P}_2$ denote the probability measures corresponding to $f_1$ and $f_2$, respectively.
%For discrete probability measures (with $\Acal$ the powerset of $\Tcal$ and $\mu$ strictly positive without loss of generality), 
In the discrete case as introduced in Section~\ref{chapter_bayes_hilbert_space}, % i.e., $\Tcal = \{t_1, \ldots , t_D\}$ and $\mu = \sum_{d=1}^D w_d \delta_{t_d}$,
the correspondence is immediate and $\OR(s,t) = \frac{\mathbb{P}_1(\{s\}) / \mathbb{P}_1(\{t\})}{\mathbb{P}_2(\{s\}) / \mathbb{P}_2(\{t\})} = \frac{\mathbb{P}_1(\{s\} \mid \{s,t\}) / (1 - \mathbb{P}_1(\{s\} \mid \{s,t\}))}{\mathbb{P}_2(\{s\} \mid \{s,t\}) / (1 - \mathbb{P}_2(\{s\} \mid \{s,t\}))}$ is the odds ratio 
for two (of potentially more) outcomes, corresponding also to the most common binary odds ratio when conditioning the outcome on being either $s$ or $t$.
%The interpretation of the estimated effects $\hh_j := \hh_j^{[m_{\text{stop}}]} (\xf_i) \in \B,~j = 1, \ldots , J$, has to respect the special structure of Bayes Hilbert spaces.
%In particular, it should be independent of the selected representative of an equivalence class in $\B$.
%Naturally, interpretation in a Bayes Hilbert space is relative. %, as $\ominus$ corresponds to a division.
%Accordingly, the shape of clr transformed effects can be interpreted using differences, resulting in an interpretation analogous to the well-known odds ratios.
%For two effects $\hh_j$ and $\hh_k$ for $j \neq k \in \{ 1, \ldots , J \}$ and $s, t \in \Tcal$, we have
%\begin{align} 
%	\exp\left( \clr [ \hh_j ] (t) - \clr [ \hh_j ](s) - \left( \clr [ \hh_k ](t) - \clr [ \hh_k ](s)\right)\right)
%	%= \frac{\frac{\clr^{-1} \left[ \clr [ \hh_j ] \right] (s)}{\clr^{-1} \left[ \clr [ \hh_j ] \right] (t)}}{\frac{\clr^{-1} \left[ \clr [ \hh_k ] \right] (t)}{\clr^{-1} \left[ \clr [ \hh_k ] \right] (s)}}
%	%= \frac{\frac{\hh_j(t)}{\hh_j(s)}}{\frac{\hh_k(t)}{\hh_k(s)}}. \label{density_odds_ratio}
%	= \frac{\hh_j(t) \, / \, \hh_j(s)}{\hh_k(t) \, / \, \hh_k(s)}. \label{density_odds_ratio}
%\end{align}
%The compound fraction on the right is called \emph{odds ratio of $\hh_j$ and $\hh_k$ for $t$ compared to $s$}, its numerator is called \emph{odds of $\hh_j$ for $t$ compared to $s$}.
%Thus, the log odds ratio corresponds to the difference of the differences of the clr transformed effects evaluated at $t$ and $s$.
%
%In the more general case, covering discrete, continuous or mixed densities, 
In a general mixed Bayes Hilbert space (including discrete and continuous ones as special case),
$\OR(s,t)$ can be interpreted as limit of usual odds ratios in the vicinity of $s$ and $t$, and provides bounds for odds ratios for general events $A, B \in \Acal$, as summarized in the proposition below.

\begin{prop}
	\label{thm:oddsratio}
	Let $B^2(\mu)$ be a mixed Bayes Hilbert space (compare Section~\ref{chapter_bayes_hilbert_space}) % $\Tcal = I \cup \Dcal$ with $I \subset \Rbb$ a nontrivial interval and $\Dcal = \{ t_1, \ldots , t_D\}$, $\Acal$ the smallest $\sigma$-algebra containing all cosed subintervals of $I$ and all points of $\Dcal$, $\mu = \lambda + \sum_{d = 1}^D w_d \delta_{t_d}$ with $w_1, \dots, w_D > 0$, 
    and $\Acal^+ := \{ A\in\Acal ~|~\mu(A)>0 \}$. Then,
	
	\begin{enumerate}[(a)]
		\item\label{thm:odds_probabilities_inequality} for all $A, B \in \Acal^+$, $\inf_{s\in A, t \in B} \OR(s,t) \leq  \frac{\Pbb_1 (A) \, / \, \Pbb_1 (B)}{\Pbb_2 (A) \, / \, \Pbb_2 (B)}  \leq \sup_{s\in A, t \in B} \OR(s,t),$
		\item\label{thm:odds_ratio_continuous_mixed} for ($\mu$-almost) all $s,t \in \Tcal$ and for $A_n, B_n \in \Acal^+$ nested sequences of intervals centered at $s$ and $t$, respectively, with $\bigcap_{n \in \Nbb} A_n = \{s\}$ and $\bigcap_{n \in \Nbb} B_n = \{t\}$,
		\begin{align*}
			\OR(s,t)=
			\lim_{n\rightarrow\infty} \frac{\Pbb_1 (A_n) \, / \, \Pbb_1 (B_n)}{\Pbb_2 (A_n) \, / \, \Pbb_2 (B_n)} . %\qquad \text{for } \bigcap_{n \in \Nbb} A_n = \{s\} \text{ and } \bigcap_{n \in \Nbb} B_n = \{t\}
		\end{align*}
		% with $A_n, B_n \in \Acal^+$ nested sequences of intervals centered at $s$ and $t$ for $n \in \Nbb$, respectively.
	\end{enumerate}

\end{prop}

Point (\ref{thm:odds_probabilities_inequality}) in particular entails that if $\OR(s,t) > 1$ for all $s\in A, t \in B$, then ${\Pbb_1 (A) \, / \, \Pbb_1 (B)} > {\Pbb_2 (A) \, / \, \Pbb_2 (B)}$,
which analogously holds when conditioning on any event $\tilde{\Tcal} \supset A \cup B$, illustrating the subcompositional coherence of the odds ratio.
When considering, by contrast, $\Pbb_1 (A) > \Pbb_2 (A)$, we cannot infer that $\Pbb_1 (A \mid \tilde{\Tcal}) > \Pbb_2 (A\mid \tilde{\Tcal})$.
By conditioning on outcomes in $A$ or $B$, $\OR(s,t) > 1$ can, however, be translated to an inequality of probabilities $\Pbb_1 (A \mid A\cup B) > \Pbb_2 (A\mid A\cup B)$. 
%To still get a coherent qualitative interpretation on the level of probabilities, the odds inequality might be interpreted as $\Pbb_1 (A \mid A\cup B) > \Pbb_2 (A\mid A\cup B)$, however.
%Due to the double fraction in the odds ratio, 
Note that the limit in (\ref{thm:odds_ratio_continuous_mixed}) %does not depend on the ratio of the $\mu(A_n)$ and $\mu(B_n)$, and 
is even well-defined and meaningful for comparison between points with $\mu(\{s\}) = 0$ mass and positive mass $\mu(\{t_d\}) = w_d > 0$ in mixed densities, since $\mu(A_n)/\mu(B_n)$ cancels out. 

\subsection{Odds ratio interpretation of additive effects}\label{chapter_interpretation}

Such an odds ratio interpretation of differences is naturally employed for a subcompositionally coherent interpretation of an effect %$h_1$ 
in an additive model %$f = h \oplus \varepsilon = h_0 \oplus h_1 \oplus \varepsilon$, 
as introduced in Section \ref{chapter_regression_model}. %, with $f,\varepsilon$ random elements in $B^2(\Tcal)$ and restricting, for simplicity and without loss of generality, to two effects $h_j: \Rbb^K \rightarrow B^2(\Tcal)$, $j=0,1$, suppressing the dependence on the covariates $\mathbf{x} \in \Rbb^K$ in the notation. 
For simplicity and without loss of generality, consider a model 
$f_i = h \oplus \varepsilon_i 
= h_0 \oplus h_1 \oplus \varepsilon_i$ with two effects $h_j: \Rbb^K \rightarrow B^2(\Tcal)$, $j \in \{0,1\}$, suppressing the dependence on the covariates $\mathbf{x}_i \in \Rbb^K$ in the notation.
Here, $h_1 = h \ominus h_0$ makes up the difference between the full predictor $h$ and all other effects in the model $h_0$ and %, suppressing dependence on $\mathbf{x}$ for brevity, 
determines their odds ratios 
\begin{align*}
	\OR_1(s,t) := \frac{(h_0(s)\oplus h_1(s))/(h_0(t) \oplus h_1(t))}{h_0(s)/h_0(t)} = h_1(s)/h_1(t) && \text{where}~s,t\in\Tcal.
\end{align*}
Clearly, $\OR_1(s,t)$ is independent of $h_0$, and hence allows for ceteris paribus interpretation as in usual linear models. 
On $\clr$ level, it might be tempting to interpret $\clr[h_1](s)>0$ as increasing effect on the overall density $h(s)$ at $s$, which is however not valid.
Instead, an appropriate relative interpretation is again obtained via odds ratios by simply using that $\log \OR_1(s,t) = \clr[h_1](s) - \clr[h_1](t)$, such that vertical differences in plots translate into log odds and in particular their sign. 
Further ideas for interpreting effects % compared to their geometric mean 
are developed in appendix~E, including the interpretation of our model as a family of scalar-on-scalar logistic models. % \ref{appendix_interpretation}. % involving their geometric means
The interpretation via odds ratios is illustrated in our application in Section~\ref{chapter_application}.

% -----------------------------------------------------------------------------------------------------------------------------------------------------

\subsection{Conditioning as projection and subcompositional dominance}
\label{chapter_subcomp_projection}
For a coherent regression approach, it is necessary that linear problems may be restricted onto subsets of the domain consistently with the geometry of the underlying space.
In the following, we show that this applies to Bayes Hilbert spaces, since restriction corresponds to orthogonal projection.
This result will in particular be used in Section~\ref{chapter_mixed_model} to simplify estimation in the mixed case.

%From Equation \eqref{eq_distance_oddsratio}, % in Section \ref{chapter_oddsratio_differences}, 
From the definition of the norm in Section \ref{chapter_bayes_hilbert_space},
it is immediately evident that for two densities $f_1,f_2 \in B^2(\Tcal) := B^2(\Tcal, \Acal, \mu)$, the distance $\|f_1 \ominus f_2\|_{B^2({\Tcal})} \geq \| f_1|_{\tilde{\Tcal}} \ominus f_2|_{\tilde{\Tcal}} \|_{B^2({\tilde{\Tcal}})}$ is greater or equal to the distance between densities on a subdomain, $B^2(\tilde{\Tcal}) := B^2(\tilde{\Tcal}, \Acal \cap \tilde{\Tcal}, \mu)$. 
This property is referred to as \emph{subcompositional dominance} in compositional data analysis 
%While it does not necessarily make sense to speak of orthogonal projections in other metric spaces for distributional data, subcompositional dominance may still be discussed. 
and already indicates that restriction/conditioning of the densities behaves similar to a projection in Bayes Hilbert spaces. 
The following proposition shows how $f\mid_{\tilde{\Tcal}}$ can indeed be understood as orthogonal projection of $f \in B^2(\Tcal)$, by first introducing a canonical embedding that enables us to identify the Bayes Hilbert space $B^2(\tilde{\Tcal})% := B^2(\tilde{\Tcal}, \Acal \cap \tilde{\Tcal}, \mu)
$ with a closed subspace of $B^2(\Tcal)$. % for any $\tilde{\Tcal} \in \Acal$, $\mu(\tilde{\Tcal}) > 0$.
%This construction is new to the best of the authors' knowledge.
%For compositional data, the correspondence of subcompositions in $\tilde{\Tcal} \subset \Tcal$ to subspaces of the Bayes Hilbert space is referred to as \emph{subcompositional coherence}~\citep{pawlowsky2015}.
%This means, performing an analysis on the subspace yields the same result as performing the analysis on the whole space and restricting the resuls on $\tilde{\Tcal}$ afterwards.

\begin{prop}\label{thm_subcompositional_coherence}
	For any $\tilde{\Tcal} \in \Acal$ with $\mu(\tilde{\Tcal})>0$, the space $B^2(\tilde{\Tcal}) = B^2(\tilde{\Tcal}, \Acal \cap \tilde{\Tcal}, \mu)$ is a closed subspace of $B^2(\Tcal) = B^2(\Tcal, \Acal, \mu)$ with respect to the embedding
	\begin{align*}
		\iota: B^2(\tilde{\Tcal}) \hookrightarrow B^2(\Tcal), &&
		\tilde{f} \mapsto
		\begin{cases}
			\tilde{f} & \text{on}~\tilde{\Tcal} \\
			\exp \Scal_{\tilde{\Tcal}} (\tilde{f}) & \text{on}~\Tcal \setminus \tilde{\Tcal}
		\end{cases} \, ,
	\end{align*}
	where $\Scal_{\tilde{\Tcal}} (\tilde{f})$ is the mean logarithmic integral as defined in \eqref{def_clr}.\footnote{Note that $\exp \Scal_{\tilde{\Tcal}}(\ft)$ corresponds to the geometric mean of $\ft$ on $\tilde{\Tcal}$ using the natural generalization of the usual definition of the geometric mean over a discrete set: For $\Tcal = \{s_1, \ldots , s_L\}$ and $g \in B^2 (\Tcal, \Pcal(\Tcal), \sum_{l = 1}^L \delta_{s_l})$, the geometric mean of %$\{g(s_1), \ldots , g(s_L)\}$ 
    $g$ on $\Tcal$ is $(\prod_{l = 1}^L g(s_l))^{1/L}%\frac{1}{L} 
    = \exp \Scal_{B^2 (\Tcal, \Pcal(\Tcal), \sum_{l = 1}^L \delta_{s_l})}(g)$. %geometric mean of $\ft$ on $\tilde{\Tcal}$. %, see Section~\ref{chapter_mixed_model}.
}\label{footnote_geometric_mean}
	This means that $\iota$ is linear and preserves the norm.
	The orthogonal projection onto this closed subspace is given by
	\begin{align*}
		P: B^2(\Tcal) \ra B^2(\Tcal), && f \mapsto \iota(f|_{\tilde{\Tcal}}),
	\end{align*}
	where $f|_{\tilde{\Tcal}} \in B^2(\tilde{\Tcal})$ denotes the function $f$ restricted to $\tilde{\Tcal}$.
	In particular, this means,
	% \begin{enumerate}[(a)]
	% 	\item\label{proof_subcoh_selbstadjungiert}
		$P^2 = P$, $P^\ast = P$ (self-adjointness),
		% \item\label{subsomp_dominance}
  and
		$\Vert P \Vert := \sup_{f \neq 0} \frac{\Vert P(f) \Vert_{B^2(\Tcal)}}{\Vert f \Vert_{B^2(\Tcal)}} = 1$.
	% \end{enumerate}
\end{prop}

% An important consequence of the properties of the orthogonal projection is that we may restrict linear problems (e.g., in regression models) onto subsets of $\Tcal$ consistently with the geometry of the Bayes Hilbert spaces. % (subcompositional coherence).

% ---------------------------------------------------------------------------------------------------------------------------

\subsection{Estimation in the mixed case using projections}\label{chapter_mixed_model}
% Recall the mixed case, i.e., $\B = B^2\left( \Tcal, \Acal, \mu\right)$ with $\Tcal = I \cup \Dcal$, where $I \subset \Rbb$ is a nontrivial interval and $\Dcal = \{ t_1, \ldots , t_D\}$, $\Acal$ the smallest $\sigma$-algebra containing all cosed subintervals of $I$ and all points of $\Dcal$, and $\mu = \delta + \lambda$, where $\delta = \sum_{d = 1}^D w_d \, \delta_{t_d}$ %for $\{t_1, \ldots , t_D\} = \Dcal \subset I$ 
% and $w_d > 0$. % for $d = 1, \ldots , D$.
Prop.~\ref{thm_subcompositional_coherence} is particularly useful for a mixed Bayes Hilbert space $\B$ as introduced in Section~\ref{chapter_bayes_hilbert_space}.
Due to the mixed reference measure, the specification of suitable basis functions $\bfe_Y \in B^2(\mu)^{K_Y}$ as required in Section~\ref{chapter_estimation_bayes} is not straightforward.
We simplify this by tracing the estimation problem back to two separate estimation problems~--~one continuous and one discrete.
For the continuous one, consider the Bayes Hilbert space $\Bl = B^2\left( \Ccal, \mathfrak{B} \cap \Ccal, \lambda \right)$, where $\Ccal := I \setminus \Dcal \in \mathfrak{B}$.
Remarkably, its orthogonal complement in $\B$ is not the Bayes Hilbert space $B^2\left( \Dcal, \mathfrak{B} \cap \Dcal, \delta \right)$. % \Pcal \left( \Dcal \right)
Instead, an additional arbitrary discrete value $t_{D+1} \in \Rbb \setminus \Dcal$ is required, which can be considered the discrete summary of $\Ccal$.
Thus, an intuitive choice is some $t_{D+1}\in \Ccal$. 
Then, the orthogonal complement of $\Bl$ in $\B$ is the Bayes Hilbert space $\Bd = B^2\left( \Dcal^\bullet, \Pcal \left( \Dcal^\bullet \right), \delta^\bullet \right)$, where $\Dcal^\bullet := \Dcal \cup \{t_{D + 1}\}$ and $\delta^\bullet := \sum_{d = 1}^{D + 1} w_d \, \delta_{t_d}$ with $w_{D + 1} := \lambda(I)$. % = \lambda (\Ccal)$.
%
%In the special case of a mixed reference measure, the specification of suitable basis functions $\bfe_Y \in B^2(\mu)^{K_Y}$ is not straightforward.
%We simplify this by tracing the estimation problem back to two separate estimation problems~--~one continuous and one discrete.
%More precisely, we decompose the Bayes Hilbert space $\B = B^2\left( I, \mathfrak{B}, \mu\right)$ with $\mu = \delta + \lambda$, where $\delta = \sum_{d = 1}^D w_d \, \delta_{t_d}$ with $t_d \in I$ and $w_d > 0$ for $d = 1, \ldots , D$, into two orthogonal subspaces, which are Bayes Hilbert spaces in their own right. 
%The first is $\Bl = B^2\left( I, \mathfrak{B}, \lambda \right)$, corresponding to the continuous case, the second is $\Bd = B^2\left( \Dcal^\bullet, \Pcal \left( \Dcal^\bullet \right), \delta^\bullet \right)$, where $\Dcal^\bullet := \{t_1, \ldots , t_{D + 1}\}$ with $t_d \neq t_{D + 1} \in \Rbb$ for all $d = 1, \ldots , D$ and $\delta^\bullet := \sum_{d = 1}^{D + 1} w_d \, \delta_{t_d}, ~ w_{D + 1} := \lambda(I)$.
%The value $t_{D+1}$ can be considered the discrete equivalent of $I \setminus \{t_1, \ldots, t_D\}$, i.e., it might be intuitive to choose $t_{D+1}$ in this set although in principle it can be an arbitrary real number.
The embeddings to consider $\Bl$ and $\Bd$ as subspaces of $\B$ are
%\begin{align*}
%&\iota_{\mathrm{c}} : \Bl \hookrightarrow \B && \fc \mapsto 
%\begin{cases}
%\fc & \mathrm{on}~ \Ccal \\
%\exp \Scal_\lambda(\fc) & \mathrm{on}~ \Dcal
%\end{cases} \\
%&\iota_{\mathrm{d}}: \Bd \hookrightarrow \B && \fd \mapsto 
%\begin{cases} 
%\fd \left( t_{D+1} \right) & \mathrm{on}~ \Ccal\\
%\fd & \mathrm{on}~ \Dcal  ,
%\end{cases} 
%\end{align*}
%where
$\iota_{\mathrm{c}} : \Bl \hookrightarrow \B$, which is the embedding defined in Proposition~\ref{thm_subcompositional_coherence} for $\tilde{\Tcal} = \Ccal$, and $\iota_{\mathrm{d}}: \Bd \hookrightarrow \B$ with %which are defined as $\iota_{\mathrm{c}} (f_c) = f_c$ and 
$\iota_{\mathrm{d}} (\fd) = \fd \left( t_{D+1} \right)$ on $\Ccal$ %, respectively, and $\iota_{\mathrm{c}} (f_c) = \exp \Scal_\lambda(\fc)$ 
and $\iota_{\mathrm{d}} (\fd) = \fd$ on $\Dcal$.
% Here,
% $\Scal_\lambda(\fc) 
% %:= S_{I} (\fc) 
% %:= S_{\Ccal}(\fc)
% %:=\frac{1}{\lambda(\Ccal)}\int_{\Ccal} \log \fc \, \dlamb
% %:= \frac{1}{\lambda(I)}\int_{I} \log \fc \, \dlamb
% $ is the mean logarithmic integral as defined in~\eqref{def_clr}.
% Note that $\exp \Scal_\lambda(\fc)$ corresponds to the geometric mean of $\fc$ using the natural generalization of the usual definition of the geometric mean of a discrete set $\{g(s_1), \ldots , g(s_L)\}$, since $(\prod_{l = 1}^L g(s_l))^{1/L}%\frac{1}{L} 
% = \exp \Scal_{B^2 (\Tcal, \Pcal(\Tcal), \sum_{l = 1}^L \delta_{s_l})}(g)$ for $\Tcal = \{s_1, \ldots , s_L\}$. %, g \in B^2 (\Tcal, \Pcal(\Tcal), \sum_{l = 1}^L \delta_{s_l})$.
%It can be viewed as the special case for $\Tcal = \{s_1, \ldots , s_D\}$ and $\mu = \sum_{d = 1}^D \delta_{s_d}$.
%\begin{align*}
%\overline{\hh_j} 
%&= \sqrt[\vert \Tcal \vert]{\prod_{t \in \Tcal} \hh_j(t)}
%%= \prod_{t \in \Tcal} h_j(t)^{\nicefrac{1}{\mu (\Tcal )}}
%= \prod_{t \in \Tcal}  \exp \left( \log \hh_j(t)^{\nicefrac{1}{\mu (\Tcal )}} \right) %\\
%%&= \exp \Bigl( \frac{1}{\mu (\Tcal )} \sum_{t \in \Tcal} \log \hh_j(t) \Bigr) 
%= \exp \Bigl( \frac{1}{\mu (\Tcal )} \int_{\Tcal} \log \hh_j \, \dmu \Bigr).
%\end{align*}
%
For $f \in \B$, the unique functions $f_{\mathrm{c}} \in \Bl, f_{\mathrm{d}} \in \Bd$ such that $f = \Jc(f_{\mathrm{c}}) \oplus \Jd (f_{\mathrm{d}})$ are given by
\begin{align}
&\fc: \Ccal \ra \Rbb , \quad
t \mapsto f(t) ,%\label{decomposition_continuous}\\
&&\fd: \Dcal^\bullet \ra \Rbb, \quad
t \mapsto 
\begin{cases} 
1 , & t = t_{D + 1} \\
% \frac{f(t)}{\exp \Scal_\lambda(f)} , & t \in \Dcal .
\frac{f(t)}{\exp \Scal_\Ccal(f)} , & t \in \Dcal .
\end{cases} \label{decomposition} %\label{decomposition_discrete}
\end{align} 
%\begin{align}
%&\fc: I \ra \Rbb  && t \mapsto f(t) \, ,\label{decomposition_continuous}\\
%&\fd: \Dcal^\bullet \ra \Rbb && t \mapsto 
%\begin{cases} 
%1 , & t = t_{D + 1} \\
%\frac{f(t)}{\exp \Scal_\lambda(f)} , & t \in \Dcal .
%\end{cases}\label{decomposition_discrete}
%\end{align} 
See Proposition~B.1 % \ref{thm_decomposition_bayes} 
in appendix~B % \ref{appendix_proofs} 
for the proof that the orthogonal complement of $\Bl$ in $\B$ is $\Bd$, including \eqref{decomposition}.
Then, we obtain $\Vert f \Vertb^2 = \Vert f_{\mathrm{c}} \Vertbl^2 + \Vert f_{\mathrm{d}} \Vertbd^2$ implying that %(as long as parameters are disjoint between components) 
minimizing the sum of squared errors~(\ref{sse}) is equivalent to minimizing its discrete and continuous components separately, greatly simplifying the model fitting, and then combining the solutions $\hh_{\mathrm{c}}$ and $\hh_{\mathrm{d}}$ in the overall solution $\hh = \Jc(\hh_{\mathrm{c}}) \oplus \Jd(\hh_{\mathrm{d}})$. 
%Therefore, we can estimate two separate models based on the discrete and continuous components of the response densities, respectively. 
%Afterwards, we embed and add the so-obtained effects to obtain the effect estimates of the mixed case.

Equivalently, we can decompose the Hilbert space $L^2_0 \left(\mathcal{T}, \mathcal{A}, \mu \right)$ such that embeddings and clr transformations commute.
See Proposition~B.2 % \ref{thm_decomposition_L2} 
in appendix~B % \ref{appendix_proofs} 
for details and proof.
\section{Application}\label{chapter_application}

% Aufpassen:
% - Immer mixed case (NIE mixed model!)
% - continuous/discrete component (NICHT part!)
% - obtain, nicht receive für aus etwas erhalten
% - Variablennamen kursiv

We use our approach to analyze the distribution of the women's share in a couple's total labor income in Germany depending on covariates. 
% Other than some recent economic literature initiated by \citet{bertrand2015}, we do not focus on the question of whether there is a drop in the density at $0.5$, but on the dependence of the share density on covariates also implied to be of potential importance. % from the economic literature.
%Although we refer to \citet{bertrand2015}, we do not focus on the question of whether there is a drop in the density at $0.5$, but on the dependence of the share density on covariates implied to be of potential importance from the economic literature.
Note that for simplicity we use the terms East/West Germany %for the areas of these former states 
also after reunification.

\subsection{Background and hypotheses}\label{background}
While there is no consensus in the literature regarding a discontinuous drop of the female income share at $0.5$ (as in \citealp{bertrand2015} for the U.S.) for Germany, there is a larger share fraction below $0.5$ reflecting the gender pay gap \citep{sprengholz2020,kuehnle2021}.
% There is a larger share fraction in Germany below $0.5$ (as in \citealp{bertrand2015} for the U.S.) reflecting the gender pay gap, but there is no consensus in the literature regarding a discontinuous drop at $0.5$ \citep{sprengholz2020,kuehnle2021}. 
The employment and earnings of female partners show a strong childhood penalty \citep{kleven2019,fitzenberger2013}. The social norm in West Germany used to be that mothers should stay at home with their children. Institutionalized child care was scarce and there are strong financial incentives for part-time work for the second earner. Together, this results in part-time employment increasing strongly for women after having their first child. We thus expect that the female income share is lower in the presence of children, reflecting a childhood penalty.
%This means that the cumulative distribution function (cdf) of the income share for couples without children lies to the right of the cdf in the presence of children.

Due to changing social norms, female employment increases strongly over time. However, occupational segregation by gender is persistent \citep{cortes2018} with men being more likely to work in better paying occupations. Still, occupations with a higher share of women seem to benefit from technological change \citep{black2010}. Thus, the income share of female partners without children is predicted to grow over time.

Ex ante reasoning suggests an ambiguous effect on %the trend in 
the childhood penalty. On the one hand, the incentives for part-time work especially for female partners with young children may %work against 
prevent an increase in the income share. Thus, the childhood penalty in the income share may even grow over time. On the other hand, growing female employment may actually increase the female income share, especially among female partners with older children.

Turning to the comparison between East and West Germany, the literature emphasizes that social norms are likely to differ between the two parts of the country \citep{beblo2018}. Before reunification, it was basically mandatory for women to work in East Germany and comprehensive institutionalized child care was available. This suggests that the female income share in East Germany is higher than in West Germany. %, i.e.\ the cdf for the former lies to the right of the cdf of the latter.
After reunification, social norms have been converging between the East and the West. In East Germany, female employment may have fallen more strongly than for males due to the strong economic transformation and the lower mobility of female partners after job loss. Part-time employment is likely to become more prevalent in East Germany, and over time mothers more often drop out of the labor force. While we expect the childhood penalty to be lower in East Germany than in West Germany, it is ex ante ambiguous whether the East-West gap in the childhood penalty decreases over time, a question of interest.
\\
To investigate these hypotheses without restricting the attention a priori to a scalar summary statistic, we investigate the female share distributions as introduced by \cite{bertrand2015} as object of interest, using  comprehensive representative German data.   

\subsection{Data and descriptive evidence on response densities}\label{chapter_soep_data}

Our data set derived from the German Socio-Economic Panel (see appendix~F % \ref{appendix_soep} 
for details) contains $154,924$ observations of couples of opposite sex living together in a household, where at least one partner reports positive labor income. 
We include cohabitating couples in addition to married ones as there is a strong tax incentive to get married in case of unequal incomes, leading to a bias. % when analyzing only married couples. 
The women's %(female) 
\emph{share} in the couple's total gross labor income together with the household's sample \emph{weight} (to ensure representativeness %of the sample 
for the German population) yields the response densities. 
Four variables serve as covariates.
First, the binary covariate \emph{West\_East} specifies whether the couple lives in \emph{West} %Germany 
or in \emph{East} Germany (including Berlin). A second finer disaggregation distinguishes six \emph{regions} (two in \emph{East} %Germany 
and four in \emph{West} Germany, see appendix~F.1). % \ref{appendix_soep_regions}).
The third covariate \emph{c\_age} is a categorical variable for the age range (in years) of the couple's youngest child living in the household: \emph{0-6}, \emph{7-18}, and \emph{other} (i.e., couples without minor children). % living in their household.
%Here, ``child of a couple'' does not necessarily imply a biologic relation. Adoptive children or grandchildren are considered the couple's children, too, if their biological parents are not living in the same household.
Finally, \emph{year} ranges from 1984 (\emph{West} Germany)/1991 (\emph{East} Germany) to 2016.

We construct an empirical response density $f_{\text{\emph{region, c\_age, year}}}: [0, 1] \ra \Rbb^+ %, s \mapsto f_{\emph{\text{region, c\_age, year}}}(s)
$
of the woman's income share $s$
for each combination of covariate values (note that \emph{region} determines \emph{West\_East}). %, with $s$ denoting the woman's income share. 
In total, this yields 552 response densities. 
Often, we just write~$f$ and omit the indices.
Before elaborating on the estimation, we determine a suitable underlying Bayes Hilbert space $\B = B^2(\Tcal, \Acal, \mu)$.
Since $s$ denotes a share, we consider $\Tcal %= I
= [0, 1]$ with $\Acal = \mathfrak{B}$. %, where~$\mathfrak{B}_{[0, 1]}$ is the Borel $\sigma$-algebra restricted on~$[0, 1]$.
The Lebesgue measure is no appropriate reference, as the boundary values~$0$ and~$1$ correspond to single-earner households and thus have positive probability mass (see appendix~F.2 % \ref{appendix_soep_barplots} 
for exemplary barplots).
A suitable reference measure respecting this structure is $\mu := \delta_0 + \lambda + \delta_1$, i.e., the mixed case with $D = 2,~ t_1 = 0, ~t_2 = 1$, and $w_1 = 1 = w_2$, see Section~\ref{chapter_bayes_hilbert_space}.
%density
%values of $f \in B^2(\mu)$ at the boundary values are obtained straightforwardly as
The values $f(0)$ and $f(1)$ are
%$f (0) = p_0$ and $f(1) = p_1$, where $p_0$ and $p_1$ are the relative frequencies for a share of $0$ and $1$, respectively.
the (weighted) relative frequencies for shares of $0$ and $1$, denoted by $p_0$ and $p_1$, respectively.
To estimate $f$ on~$(0, 1)$, we compute continuous densities
%$f_{(0, 1)}: (0, 1) \ra \Rbb^+$ using weighted kernel density estimation with beta-kernels, which preserve the predetermined support $(0, 1)$,
based on dual-earner households, and multiply them by $p_{(0, 1)} = 1 - p_0- p_1$. %, which is the relative frequency for a share in $(0, 1)$. %, to obtain unconditional densities for all couples.
For this purpose, weighted kernel density estimation with beta-kernels \citep{chen1999} is used to preserve the support $(0, 1)$ and include sample weights, see appendix~F.3 % \ref{appendix_kernel_density_estmation} 
for details.
% A small bandwidth is used to account for sharp variations in the densities (see appendix~\ref{appendix_kernel_density_estmation}).
\\
The response densities are very similar in the different \emph{regions} within \emph{West} and \emph{East} Germany, respectively.
Thus, % we restrict visualization in 
Figure~\ref{original_densities_south_east} % to the exemplary 
exemplarily shows the 
\emph{regions} \emph{west} (North Rhine-Westphalia) for \emph{West} Germany and \emph{east} (Saxony-Anhalt, Thuringia, Saxony) for \emph{East} Germany. % (\emph{west} (\emph{east}) is one \emph{region} in \emph{West} (\emph{East}) Germany, see appendix~\ref{appendix_soep_regions}).
See Figure~F.7 % \ref{figure_densities} 
in appendix~F.4 % \ref{appendix_soep_sensitivity_check} 
for the full figure for all six \emph{regions}, with additional illustration of the relative frequencies $p_0,\, p_{(0, 1)},\, p_1$ %per household type (non-working women, dual-earner househods, single-earner women)
over time.
%The upper panel of the 
Figure~\ref{original_densities_south_east} depicts the response densities for all \emph{years} by \emph{c\_age} for the \emph{regions} \emph{west} and \emph{east}, with a color gradient and different line types distinguishing the \emph{year}.
The density values $f(0)$ and~$f(1)$ %(corresponding to $p_0$ and $p_1$)
are represented as dashes, shifted slightly outwards for better visibility.
%The bottom panel shows the evolution of the shares of the three household types over time, with single-earner households indicated by green (non-working woman) and red (non-working man) dashes and dual-earner households indicated by blue circles.
%
Consider the continuous parts ($s \in (0,1)$)%in the upper panel: 
: In \emph{west} (first row), the densities differ between couples with (\emph{0-6} and~\emph{7-18}) and without minor children (\emph{other}), with the latter having more probability mass to the right reflecting lower female shares in the presence of children. In \emph{east}, the shapes are more egalitarian and vary much less with the age of the youngest child. In all cases, the fraction of couples with a share less than $0.5$ exceeds the fraction with a share larger than $0.5$. Over time, the probability mass for a small share increases and that of non-working women declines, reflecting the increase in female part-time employment. This highlights the importance of considering both single- and double-earner couples and thus mixed densities to obtain a full picture.
%Generally, the densities in East Germany are similar for all values of \emph{c\_age}.
%In West Germany, only the densities in \emph{c\_age} 1 and 2 are equally shaped.
%The ones in c\_age 3 rather match the trend of the densities related to the new federal states.
%
%Regarding the relative frequencies shown in the lower part of Figure~\ref{original_densities_south_east}, the frequencies for $s=0$, i.e., non-working women, tend to decrease over the years for all values of \emph{c\_age} in \emph{west} and \emph{east}.
%Note that for \emph{c\_age} 1 in \emph{east} the decrease of $p_0$ only starts in the late 1990s with the frequencies increasing beforehand.
%Both, increase and decrease are steeper in \emph{east} than in \emph{west}.
%In both regions and all values of \emph{c\_age}, t
%The development of the relative frequencies $p_{(0, 1)}$ of dual-earner households is in the opposite direction to $p_0$, while $p_1$ (single-earner women) stays roughly at the same mostly low level over the years. % (with the exception of \emph{c\_age} 3 in \emph{west}, where they increase slightly).
%%Single-earner women are mostly the minority.
The shares of dual-earner households and non-working women evolve in opposite direction over time, while the share of single-earner women remains small. % (see lower panel).
%Single-earner women are clearly the minority in both regions in \emph{c\_age} 1, while in \emph{c\_age} 3 they are almost at the same level as non-working women, especially in the most recent years, with dual-earner households being the most frequent.
%That is also the case in \emph{east} and \emph{c\_age}~2.
%In \emph{west} and \emph{c\_age}~2, $p_0$ starts as the highest frequency in the 1980s and decreases until it is almost at the level of $p_1$ in 2016.
% Drin lassen?
%In line with \citet{kuehnle2021}, but not with \citet{%bertrand2015,
%sprengholz2020}, our descriptive findings do not indicate a discontinuous drop at $s=0.5$.
%The descriptive analysis shows no evidence for a discontinuous drop in the share density above 0.5.
\ifnum\value{jasa}=1
{
\vspace{-0.87cm}
} \else
{
} \fi
\begin{figure}[t] %[H]
%\begin{center}
\begin{minipage}{0.89\textwidth} % for figure with frequencies graphic: 0.93
\includegraphics[width=\textwidth]{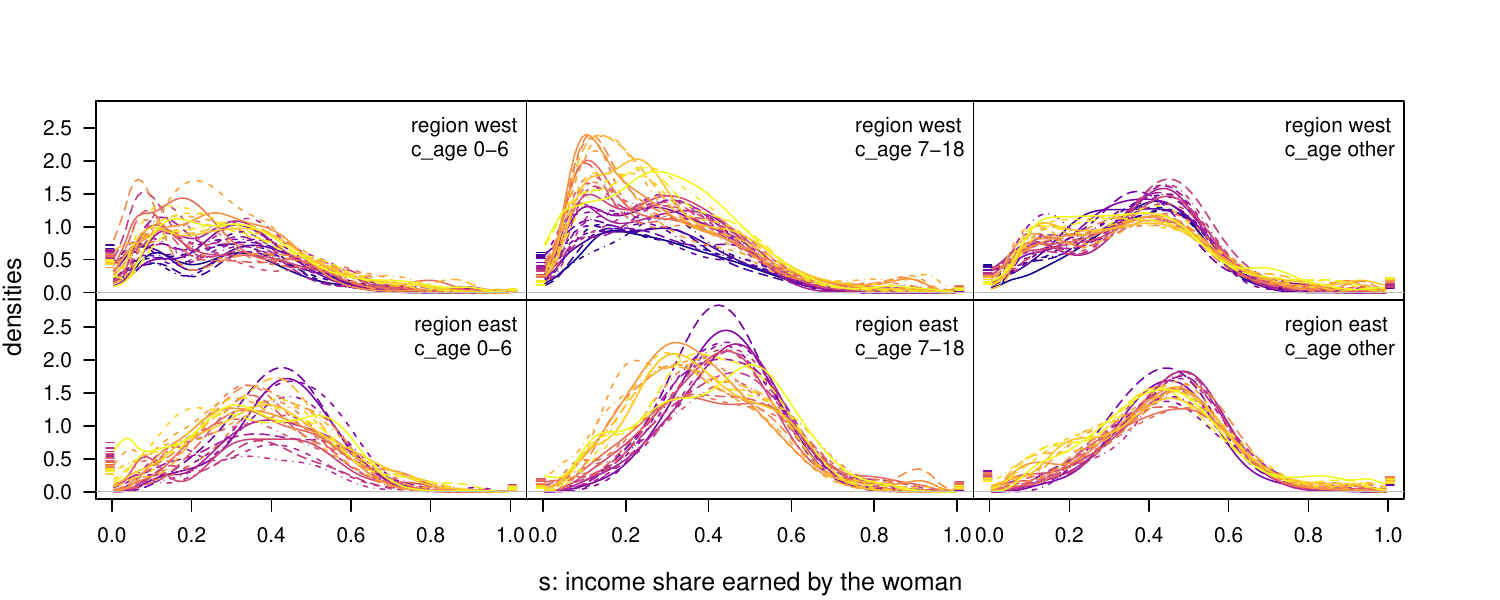} %\\[-1mm] % \\[2mm]
%\hspace*{0.8cm}
%\includegraphics[width=0.87\textwidth]{Images/legend_initial_hist_years_h_long.pdf} \\[-11mm] % \\[-9mm]
%\includegraphics[width=\textwidth]{Images/initial_west_east.pdf}
\end{minipage}
\hspace{-0.6cm}
\begin{minipage}{0.109\textwidth} % for figure with frequencies graphic (year_legend_v_1): 0.059
\includegraphics[width=1.3\textwidth]{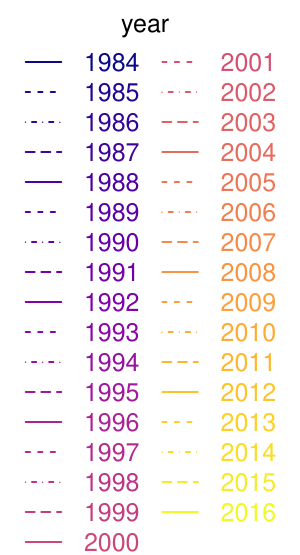}
\end{minipage}
\vspace{-0.2cm}
\caption{Response densities %[upper $2 \times 3$ panels] and corresponding relative frequencies [lower $2 \times 3$ panels]
for \emph{regions} \emph{west} and \emph{east} [rows] for all three values of \emph{c\_age} [columns]. %, 1 corresponding to couples whose youngest child is aged 0-6 years, 2 to 7-18, and 3 to the remaining couples.
\label{original_densities_south_east}}
\end{figure}
%\vspace{-0.3cm}

\subsection{Model specification}\label{chapter_soep_model}
% Reference: old, c\_age 3, 1991

% Based on the empirical response densities $f_{\text{\emph{region, c\_age, year}}}$, we 
We estimate the model
% \ifnum\value{jasa}=1 {\vspace{-0.35cm}} \else{} \fi
\begin{align}
f_{\text{\emph{region, c\_age, year}}} &= \beta_0 \oplus \beta_{\text{\emph{West\_East}}} \oplus \beta_{region} \oplus \beta_{c\_age} \oplus \beta_{\text{\emph{c\_age, West\_East}}} \notag \\
&\hspace{0.45 cm} \oplus g(year) \oplus g_{\text{\emph{West\_East}}} (year) \oplus g_{c\_age} (year) \notag \\
&\hspace{0.45 cm} \oplus g_{\text{\emph{c\_age, West\_East}}}(year)
\oplus \varepsilon_{\text{\emph{region, c\_age, year}}}, \label{soep_model}
\end{align}
% \ifnum\value{jasa}=1 {\vspace*{-1cm}} \else{} \fi
based on the empirical response densities $f_{\text{\emph{region, c\_age, year}}}$.
%with an intercept~$\beta_0$ corresponding to the reference category \emph{old}, \emph{c\_age} 3, \emph{year} 1991, group-specific intercepts $\beta_{\text{\emph{West\_East}}},$ $\beta_{region},$ and $\beta_{c\_age},$ an interaction effect $\beta_{\text{\emph{c\_age, West\_East}}},$ a flexible effect $g(year),$ group-specific flexible effects $g_{\text{\emph{West\_East}}} (year),$ $g_{c\_age} (year),$ and a flexible interaction effect $g_{\text{\emph{c\_age, West\_East}}}(year)$.
All summands are densities of the share~$s \in [0, 1]$ and elements of the Bayes Hilbert space $B^2 (\mu)$.
The model is reference coded with reference categories $West\_East = West, \, c\_age = other$, and $year = 1991$.
The corresponding effect for the reference is given by the intercept $\beta_0$.
The effect for the six regions $\beta_{region}$ is centered around the respective $\beta_{\text{\emph{West\_East}}}$.
% The other group-specific intercepts $\beta_{region}$ and $\beta_{c\_age}$ correspond to the covariates in the index, i.e., region in $\{\text{\emph{northwest}}$, \emph{west}, \emph{southwest}, \emph{south}, \emph{east}, $\text{\emph{northeast}}\}$ and c\_age in $\{ 1, 2, 3\}$.
%Furthermore, $\beta_{region}$ is centered around the \emph{West\_East} effects, i.e., we obtain the expected densities of regions corresponding to old or new federal states by calculating $\beta_0 \oplus \beta_{region}$ or $\beta_0 \oplus \beta_{new} \oplus \beta_{region}$, respectively.
%Furthermore, $\beta_{\text{\emph{c\_age, West\_East}}}$ specifies the interaction of c\_age and West\_East.
The smooth year effect $g(year)$ describes the deviation for each \emph{year} from the reference 1991 (for \emph{West} Germany and \emph{c\_age} \emph{other}).
Finally, several interaction terms are included with a group-specific intercept density $\beta_{\text{\emph{c\_age, West\_East}}}$ as well as group-specific flexible terms $g_{West\_East} (year)$, $g_{c\_age} (year)$, and $g_{\text{\emph{c\_age, West\_East}}}(year)$.
They are constrained to be orthogonal to the respective main effects using a similar constraint as~(\ref{identifiability_constraint}) to ensure identifiability.
Due to reference coding, all partial effects for the reference categories are zero. %, e.g., $\beta_{old} = 0 = g(1991)$.
%Note that due to the reference coding, we have $h_j(\xf_r) = 0$ for $\xf_r$ corresponding to the covariate values of the reference category for all partial effects $h_j \neq \beta_0$ specified in~(\ref{soep_model}).
%{\bf BF: How can the interaction effects be centered around the main effects, if there is reference coding? Does $\beta_{\text{\emph{c\_age, West\_East}}}$ take nonzero values for the reference categories $West\_East = old, \, c\_age = other$? Figure E.12, right graph, in appendix suggests that the interaction effect is zero for the reference category and that the effects do not sum to zero.}

As described in Section~\ref{chapter_mixed_model}, we decompose the Bayes Hilbert space $\B$ into two orthogonal subspaces $\Bl = B^2((0, 1), \mathfrak{B} \cap (0, 1), \lambda)$ and $\Bd = B^2(\Dcal^\bullet, \Pcal(\Dcal^\bullet), \delta^\bullet)$, where $\Dcal^\bullet = \{t_1, t_2, t_3\}$ and $\delta^\bullet = \sum_{d=1}^3 \delta_{t_d}$
%We choose $t_3 = 1/2 %\frac12
%$ to represent the continuous component in between the boundary values $t_1 = 0$ and $t_2 = 1$.
with $t_3 := 1/2$ chosen as additional discrete value.
%Note that mathematically, $t_3$ can be any real number.
For every $f$ we generate the unique functions $f_{\mathrm{c}} \in \Bl$ and $f_{\mathrm{d}} \in \Bd$ as in~\eqref{decomposition}. %, \eqref{decomposition_continuous} and~\eqref{decomposition_discrete}
%which are the response densities of two separate estimation problems, one continuous and one discrete.
As proposed in Section~\ref{chapter_estimation_bayes}, we choose transformed cubic B-splines as basis functions $\bfe_Y$ for the continuous component ($K_Y = 53$) and a transformed basis of indicator functions for the discrete component.
The remaining specification is identical in both model components.
We use an anisotropic penalty without penalizing in direction of the share, i.e., $\lambda_Y = 0$, to ensure the necessary flexibility towards the boundaries.
%To keep the models analogous, we also don't penalize in the discrete model.
%It also does not seem reasonable in this case, due to $t_3$ being rather a theoretical construct.
For the flexible nonlinear effects, the selected basis functions $\bfe_j$ are cubic B-splines with penalization of second order differences.
We set the degrees of freedom in covariate direction (per iteration) to 2 for all effects but %the intercept
$\beta_0$ and $\beta_{\text{\emph{West\_East}}}$, as these only allow for a maximum value of 1.
Regarding base-learner selection, $\beta_{\text{\emph{West\_East}}}$ thus is at a slight disadvantage compared to other main effects.
However, in a sensitivity check imposing equal degrees of freedom for all base-learners, %by adjusting $\lambda_Y$ to $1$ for all effects, 
we do not observe large deviations in the selection frequencies, while the fit to the data is better with unequal degrees of freedom, see appendix~F.4. % \ref{appendix_soep_sensitivity_check}.
Note that the intercept as well as the interaction effects are separated from the main effects due to the orthogonalizing constraints, ensuring a fair selection for the remaining base-learners.
The starting coefficients are set to zero in every component and we set the step-length $\kappa$ to $0.1$.
%We determine the stopping iteration using bootstrapping.
Based on $25$ bootstrap samples, we obtain a stopping iteration value of $262$ %Stand: 25.6.2020
for the continuous model and~$731$ %Stand: 25.6.2020
for the discrete model, respectively.

\subsection{Regression Results}\label{chapter_soep_results}

All effects in model (\ref{soep_model}) are selected (see appendix~F.5). %\ref{appendix_estimated_effects}) 
In total $R^2=47\%$ of the variance is explained by the covariate effects in the continuous model component, even $69\%$ in the discrete model component, using in-sample residuals from the model fit on the whole data.
As expected, we obtain slightly lower explained average variances of $40\%$ (ranging from $31\%$ to $50\%$) for the continuous and $64\%$ ($56\%$ to $70\%$) for the discrete model, considering out-of-sample errors from the 25 bootstrap samples instead.  %in-bag risk-reduction (i.e.\ to the . % after bootstrap-based early stopping. 
Due to early stopping, the in-sample $R^2$ is slightly over-optimistic, while the out-of-sample $R^2$ is somewhat pessimistic since it is based on effectively smaller training samples. 
The high explained variance is also reflected by predictions mostly showing a close fit (Fig.~F.8 % \ref{figure_predictions} 
in appendix~F.4). % \ref{appendix_soep_sensitivity_check}).
% Relative Risk Reduction without intercept
%                        continuous discrete
% beta_West_East               0.10     0.07
% beta_region                  0.06     0.04
% beta_c_age                   0.31     0.50
% beta_c_age_West_East         0.00     0.02
% g_year                       0.39     0.31
% g_West_East_year             0.05     0.01
% g_c_age_year                 0.07     0.03
% g_c_age_West_East_year       0.01     0.03
%Furthermore, Figure~\ref{relrisk} in appendix~\ref{appendix_soep_sensitivity_check} shows that 
Most of the explained variance is due to the main effects $\betah_{\mathrm{c\_age}}$ ($31\%$ in the continuous component of the density, $50\%$ in the discrete one; see also Fig.~ F.4 % \ref{relrisk} 
in appendix~F.4), % \ref{appendix_soep_sensitivity_check}), 
$\hat{g}(\mathrm{year})$ (continuous $39\%$, discrete $31\%$) and $\betah_{\mathrm{West\_East}}$ (continuous $10\%$, discrete $7\%$).
Percentages are computed based on the component-wise risk-reduction. 
%have the highest relative in-bag risk reduction among the estimated effects beside the intercept, i.e., together, they explain most of the total variability in the data.
%They show mostly good fit to the data.% compared to the respective response densities
%
In the following, we discuss the key findings, focusing on our hypotheses. All effects are illustrated in appendix~F.5 with quantitative example interpretations via (log) odds ratios provided for further main effects.
%Note that the algorithm selected every effect into the model.
%Furthermore, note that all explicit values serving as examples in this section are rounded to three digits.
% We illustrate them in both spaces: $L_0^2(\mu)$, which is always arranged left, and $B^2(\mu)$, which is ordered right.
% We refer to the former as the \emph{clr-level} and to the latter as the \emph{Bayes-level}.

\begin{figure}%[h] %[H]
\begin{center}
\includegraphics[width=0.49\textwidth]{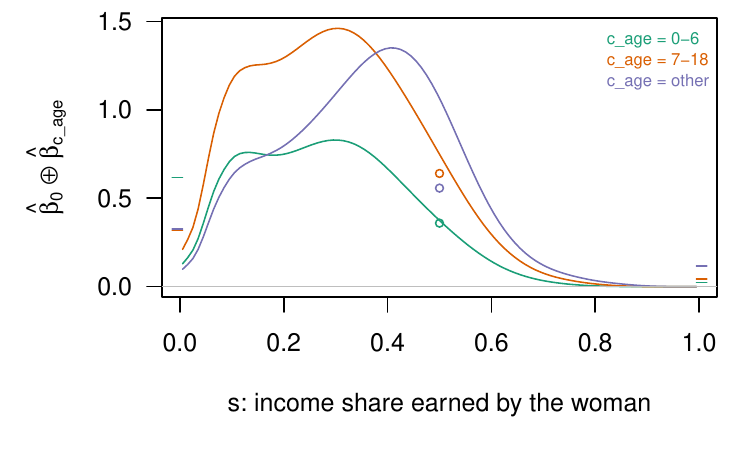}
\includegraphics[width=0.49\textwidth]{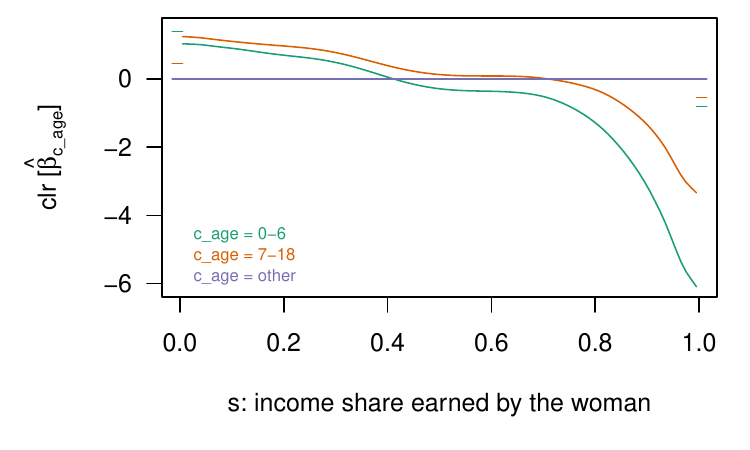}
\end{center}
\vspace{-0.5cm}
\caption{Expected densities for couples living in \emph{West} Germany in 1991 for all three values of \emph{c\_age} [left] and clr transformed estimated effects of \emph{c\_age} for ceteris paribus log odds ratio interpretations [right].
\label{estimated_child_group}}
\end{figure}

%Figure~\ref{estimated_child_group} [left] illustrates the perturbation of the intercept by the \emph{c\_age} effect's impact on the share distribution for West Germany.
%The circles at~$0.5$ represent the Lebesgue integral of the respective function, i.e., the expected relative frequency of dual-earner households.
%On the right, we show the clr transformed effect
%for interpretation via (log) odds ratios, see Section~\ref{chapter_interpretation}.
%Recall that the difference of the differences of two clr transformed effects at share values $t$ and $s$ is the log odds ratio of these effects for $t$ compared to $s$ (ceteris paribus). %We usually consider $s < t$.
%%Furthermore, recall that the relation of the odds of two effects for $t \in I_t \subset [0, 1]$ compared to $s \in I_s \subset [0, 1]$ is preserved by the odds of probabilities.
%As \emph{c\_age=other} is the reference category, we have $\betah_0 \oplus \betah_{other} = \betah_0$ and $\clr [\betah_{other}] = 0$. % (Figure~\ref{estimated_child_group} [right]).

The left part of Figure~\ref{estimated_child_group} shows % the perturbation of the intercept by the \emph{c\_age} effect, i.e., 
the expected densities for couples without minor children (\emph{c\_age other}), for couples with children aged \emph{0-6}, and for couples with children aged \emph{7-18} living in \emph{West} Germany in 1991.
The circles at~$0.5$ represent %the Lebesgue integral of the respective function, i.e., 
the expected relative frequency of dual-earner households.
Our main finding is that the expected density on $(0, 1)$ for \emph{c\_age other} is unimodal with a maximum above 0.4, while the densities for \emph{c\_age 0-6} and \emph{7-18} are bimodal with both maxima to the left of 0.4. The latter show a similar shape, but are scaled differently.
%
%In the left part of the figure, we observe this straightforwardly, as the shapes of the densities on $(0, 1)$ for \emph{0-6} and \emph{7-18} are similar, just scaled differently.
%Both show a larger probability mass for small share values %(about $s \in (0, 0.3)$)
%with a second maximum further left compared to the intercept (couples without minor children in 1991 living in the West.
%The absolute values for couples with children aged 7-18 years (\emph{c\_age}~2) are similar to couples without children (\emph{c\_age} 3) and larger than for couples with young children (\emph{c\_age}~1).
%This is also reflected in the expected relative frequency of dual-earner households (circles at $0.5$), which is lowest for \emph{c\_age} 1.
The relative frequencies of dual-earner households (circles at $0.5$) and the two types of single-earner households (dashes at $0$, $1$) are similar for couples with children aged \emph{7-18} years and couples without minor children, respectively. In contrast, the relative frequency of non-working women is much higher and the relative frequency of dual-earner households is much lower for couples with children aged \emph{0-6}.
The right part of the figure shows the clr transformed effect for interpretation via (log) odds ratios, see Section~\ref{chapter_interpretation}.
%Recall that the difference of the differences of two clr transformed effects at share values $t$ and $s$ is the log odds ratio of these effects for $t$ compared to $s$ (ceteris paribus). %We usually consider $s < t$.
%Furthermore, recall that the relation of the odds of two effects for $t \in I_t \subset [0, 1]$ compared to $s \in I_s \subset [0, 1]$ is preserved by the odds of probabilities.
As \emph{c\_age=other} is the reference category, we have %$\betah_0 \oplus \betah_{other} = \betah_0$ and 
$\clr [\betah_{other}] = 0$. % (Figure~\ref{estimated_child_group} [right]).
The clr transformed effects of \emph{c\_age 0-6} and \emph{7-18} %depicted to the right 
again show similar shapes on $(0, 1)$, but shifted vertically.
As the log odds ratio of $\betah_k$ and $\betah_{other}$ for $s$ compared to $t$ corresponds to vertical differences within $\clr [\betah_k]$ at $s$ and $t$, $k \in \{\emph{0-6}, \emph{7-18}\}$, the log odds ratio of $\betah_{\emph{0-6}}$ and $\betah_{other}$ is similar to the one of $\betah_{\emph{7-18}}$ and $\betah_{other}$,
%This matches our observation that the perturbations of these two effects with the intercept are shaped alike (approximately proportional).
implying similar impact on the shape of the density. %under perturbation leading to approximately proportional densities. %, which we already observed for the perturbations of the effects with the intercept.
%Proof of this statement:
%\begin{align*}
%\frac{\frac{\hh_j(t)}{\hh_j(s)}}{\frac{\hh_{\Jcal}(t)}{\hh_{\Jcal}(s)}}
%= \frac{\frac{\hh_k(t)}{\hh_k(s)}}{\frac{\hh_{\Jcal}(t)}{\hh_{\Jcal}(s)}}
%\quad \forall s, t \in \Tcal
%&& \Lera
%&& \frac{\hh_j(t)}{\hh_j(s)}
%&= \frac{\hh_k(t)}{\hh_k(s)}
%\quad \forall s, t \in \Tcal
%\\
%&& \Ra
%&& \int_{\Tcal} \frac{\hh_j(t)}{\hh_j(s)} \, \dmu(t)
%&= \int_{\Tcal} \frac{\hh_k(t)}{\hh_k(s)} \, \dmu(t)
%\quad \forall s \in \Tcal
%\\
%&& \Lera
%&& \frac{\hh_j(s)}{\int_{\Tcal} \hh_j(t) \, \dmu(t)}
%&= \frac{\hh_k(s)}{\int_{\Tcal} \hh_k(t) \, \dmu(t)}
%\quad \forall s \in \Tcal
%\\
%&& \Lera
%&& \hh_j
%&= \hh_k
%\\
%&& \Lera
%&& \hh_{\Jcal} \oplus \hh_j
%&= \hh_{\Jcal} \oplus \hh_k
%\end{align*}
%
Due to the monotonicity of both effect functions, both log odds ratios are always negative for $s > t \in (0, 1)$ (e.g., $-4.2$ for $\betah_{\emph{0-6}}$ and $-3.4$ for $\betah_{\emph{7-18}}$ for $s = 0.9, t = 0.1$), i.e., the odds for any larger % income share
versus any smaller income share are always smaller for couples with than for couples without minor children (by factor $\exp(-4.2) \approx 0.01$ for $\betah_{\emph{0-6}}$ and $\exp(-3.4) \approx 0.03$ for $\betah_{\emph{7-18}}$ for $s = 0.9, t = 0.1$), reflecting the strong childhood penalty in \emph{West} Germany in 1991.
% See appendix~ F.5 % \ref{appendix_estimated_effects} 
% for quantitative examples of concrete odds ratios. 
%As expected, the implied cdf for \emph{0-6} lies to the left of the cdf for \emph{7-18} and the latter lies to the left of the cdf for (\emph{other}).

%\ifnum\value{jasa}=1
%{
%\begin{figure}[H]
%\begin{center}
%\includegraphics[width=0.49\textwidth]{Images/estimated_child_group_intercept.pdf}
%\includegraphics[width=0.49\textwidth]{Images/estimated_child_group_clr.pdf}
%\end{center}
%\vspace{-0.5cm}
%\caption{Expected densities for couples living in \emph{old} federal states in 1991 for all three values of \emph{c\_age} [left] and clr transformed estimated effects of \emph{c\_age} [right].
%\label{estimated_child_group}}
%\end{figure}
%} \else
%{} \fi

%\vspace{-0.8cm}
\begin{figure}%[H] %[H]
\begin{center}
\includegraphics[width=1.05\textwidth]{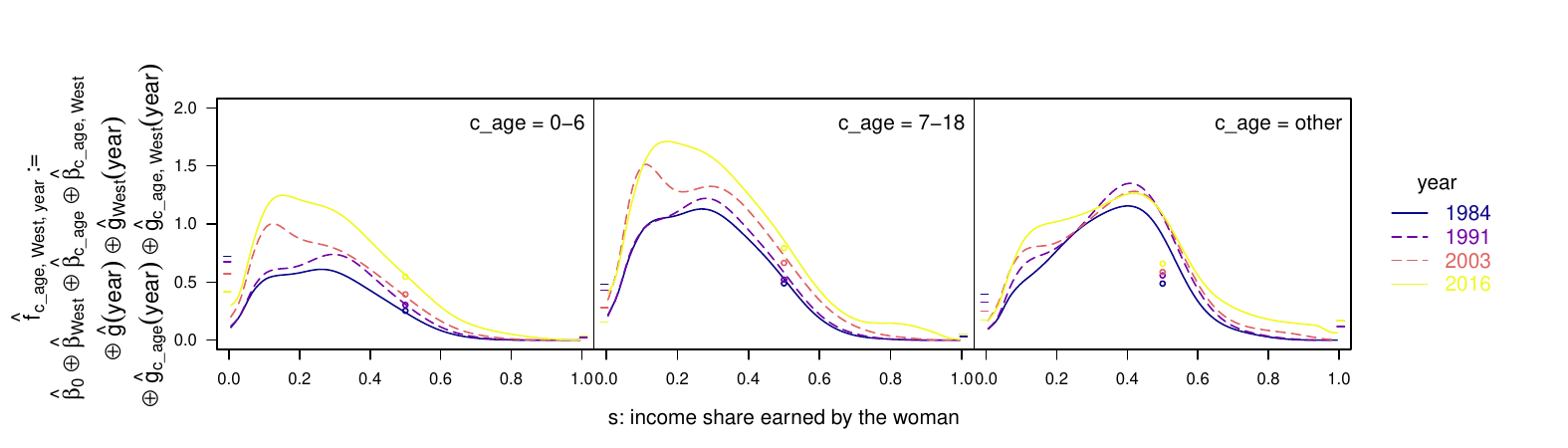}
\end{center}
\vspace{-0.5cm}
%\caption{Expected densities for couples without minor children living in West Germany over time [left] and clr transformed estimated effects of \emph{year} [right]. \label{estimated_year}
\caption{Expected densities in the \emph{years} 1984, 1991, 2003, and 2016 for \emph{West} Germany for couples whose youngest child is aged \emph{0-6} [left], \emph{7-18} [middle] and couples without minor children [\emph{c\_age} = \emph{other}, right]. \label{estimated_year_old_cgroup}}
%{\bf BF: This figure should only be shown in the appendix. Here we should show the year effects for four selected years 1984, 1991, 2003, 2016 and for the three \emph{c\_age} groups: other, 0-6, 7-18 in three separate figures.}
\end{figure}
% \vspace{-1.2cm}
\begin{figure}%[H]
\begin{center}
\includegraphics[width=1.05\textwidth]{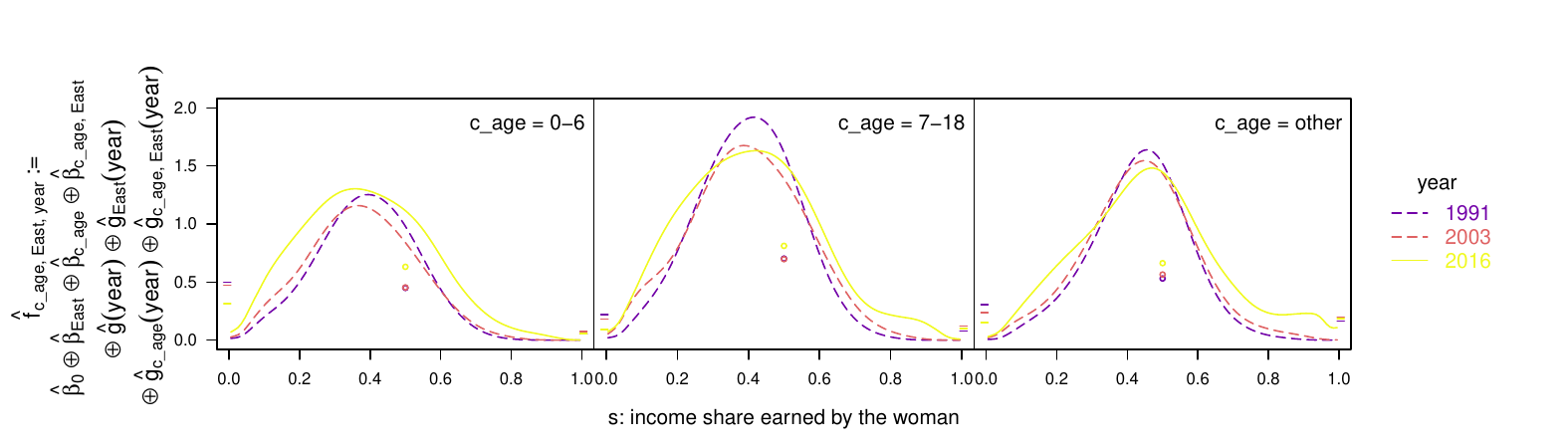}
\end{center}
\vspace{-0.5cm}
\caption{Expected densities in the \emph{years} 1991, 2003, and 2016 for \emph{East} Germany for couples whose youngest child is aged \emph{0-6} [left], \emph{7-18} [middle] and couples without minor children [\emph{c\_age} = \emph{other}, right]. \label{estimated_year_new_cgroup}}
%{\bf BF: This figure should again involve three separate figures, analogous to Figure \ref{estimated_year}.}
\end{figure}

Figure~\ref{estimated_year_old_cgroup} shows the expected densities for \emph{West} Germany for four selected \emph{years}, separately for couples with and without minor children (see Figure~ F.16 % \ref{appendix_estimated_year_old_new_cgroup} 
in appendix~F.5 % \ref{appendix_estimated_effects} 
for all \emph{years}). For \emph{other}, the frequency of non-working women ($s = 0$) falls % monotonously % continuously 
over time and the density becomes more dispersed with a lower maximum around 0.4 in 2016 than in 1993 and 2003 (which was even lower in 1984). In fact, by 2016 the expected density tends to have a second maximum further left, % and a heavier tail on the right, 
most likely due to the %strong 
growth of part-time employment even among women without minor children. Furthermore, the frequency of single-earner women ($s = 1$) increases to a level %which is 
similar to the frequency of non-working women and the continuous density has a heavier tail on the right. 
For \emph{0-6} and \emph{7-18}, we also observe a fall in the frequency of non-working women and a stronger concentration around the larger mode until 1991. However, up to 2016 the distributions show more probability mass for small shares, likely reflecting the even larger growth of part-time employment among women with minor children. %This implies a growing childhood penalty over time.

Figure~\ref{estimated_year_new_cgroup} shows the expected densities in \emph{East} Germany for selected \emph{years} (see Figure~F.16 % \ref{appendix_estimated_year_old_new_cgroup} 
in appendix~F.5 % \ref{appendix_estimated_effects} 
for all \emph{years}). In all three cases, the share distribution has a unique mode at or above 0.4. The distribution becomes more dispersed over time, with more probability mass moving to the left and a growing right tail. The frequency of non-working women is falling over time. 
While showing a similar trend as in \emph{West} Germany, %there remain persistent differences. I
in \emph{East} Germany, the frequency of non-working women for couples with minor children remains much lower and the shape of the distribution shows no trend towards a second maximum at a low share. Hence, there remains a considerable West-East gap in the childhood penalty, a main question of interest. %{\bf BF: My interpretation is based on Figure E.16 in the appendix. If this is correct, I think there is a conflict with Figure E.12 in the appendix. Please check.}

% %\vspace{-0.8cm}
% \begin{figure}%[h]
% \begin{center}
% \includegraphics[width=1.05\textwidth]{Images/estimated_sel_year_new_child_group_intercept_h.pdf}
% \end{center}
% \vspace{-0.5cm}
% \caption{Expected densities in the \emph{years} 1991, 2003, and 2016 for \emph{East} Germany for all three values of \emph{c\_age}: \emph{other} [left], \emph{0-6} [middle], \emph{7-18} [right]. \label{estimated_year_new_cgroup}}
% %{\bf BF: This figure should again involve three separate figures, analogous to Figure \ref{estimated_year}.}
% \end{figure}

%sact
To quantify this West-East gap in the childhood penalty for $year\in \{1991, 2016\}$, we make use of the additive model structure and calculate it by %the clr transformation of 
the difference-in-differences (DiD) effect: 
% \ifnum\value{jasa}=1
% {
$
DiD_{\text{\emph{c\_age, year}}}= %clr[
(\hat{f}_{\text{\emph{c\_age, West, year}}} \ominus \hat{f}_{\text{\emph{other, West, year}}} )  \ominus ( \hat{f}_{\text{\emph{c\_age, East, year}}} \ominus \hat{f}_{\text{\emph{other, East, year}}} )  %]  
$ for  $\text{\emph{c\_age}} \in \{\text{\emph{0-6}}, \text{\emph{7-18}}\}$.
% } \else
% {
% \[
% DiD_{\text{\emph{c\_age, year}}}= %clr[
% (\hat{f}_{\text{\emph{c\_age, West, year}}} \ominus \hat{f}_{\text{\emph{other, West, year}}} )  \ominus ( \hat{f}_{\text{\emph{c\_age, East, year}}} \ominus \hat{f}_{\text{\emph{other, East, year}}} ) . %]  
% \]
% } \fi
%These DiD effects have a ceteris paribus interpretation. %, as specified in eq.\ (\ref{soep_model}). 
Figure~\ref{estimated_penalty} shows the corresponding log odds %of $DiD_{c\_age, year}%(s,t)
%$ for $s$ compared to $t$ for pairs $(s,t) \in [0,1]^2$, i.e., 
% \ifnum\value{jasa}=1
% {
$LO_{\text{\emph{c\_age, year}}}(s, t) := \log \left([ DiD_{\text{\emph{c\_age, year}}}](s) / [ DiD_{\text{\emph{c\_age, year}}}](t)\right) % =\clr [ DiD_{\text{\emph{c\_age, year}}}](s) - \clr [ DiD_{\text{\emph{c\_age, year}}}](t)
$
% } \else
% {
% \[
% LO_{\text{\emph{c\_age, year}}}(s, t) := \log \left([ DiD_{\text{\emph{c\_age, year}}}](s) / [ DiD_{\text{\emph{c\_age, year}}}](t)\right) % =\clr [ DiD_{\text{\emph{c\_age, year}}}](t) - \clr [ DiD_{\text{\emph{c\_age, year}}}](s)
% \]
% }\fi
%of $DiD_{\text{\emph{c\_age, year}}}$ % for $s$ compared to $t$ 
for $s,t \in [0,1]$, see %eq.\ (\ref{density_odds_ratio})
Sec.~\ref{chapter_interpretation}, as heat maps. 
We omit the index $\text{\emph{c\_age, year}}$ in the following. % Originally after footnote;
%\footnote{
%%% footnote-start
The log odds for $s,t \in (0,1)$ are shown in the inner quadrant, those involving the two mass points $0$ and $1$ in the encircling bands, with inner bands comparing $0$, $1$ to shares in $(0, 1)$ and outer (constant) bands to the event dual-earner household $(0<s,t<1)$. Corners correspond to log odds comparing single-earner couples.
%The inner quadrant shows the respective heat map for $s,t \in (0,1)$. The log odds involving the two mass points $0$ and $1$ are given by the band around the inner quadrants. The top-left corner concerns the log odds for $s=0$ (non-working woman) compared to $t=1$ (single-earner woman). The inner bands around the inner quadrant correspond to the log odds between a mass point $0,1$ and a share in $(0, 1)$. The outer bands show the constant log odds between one of the mass points and the event dual-earner household $(0<s,t<1)$. %}. 
%%% footnote-end
%
%We omit the index $\text{\emph{c\_age, year}}$ in the following.
A positive [negative] value implies that the log odds for shares $s$ versus $t$ are higher [lower] in the \emph{West} than in the \emph{East}. Thus, $LO(s,t)>0$ for $s<t$ implies that the child penalty (lower share $s$ is more likely relative to $t$ in the presence of children) is more pronounced (stronger) in the \emph{West} than in the \emph{East}. %Note that the heatmaps in Figure~\ref{estimated_penalty} are symmetric around the 45-degree-line with $DiD(t,s)=-DiD(s,t)$ and the values are zero on the 45-degree-line.
%The log odds involving the mass points are depicted by the bands around the inner quadrants.
%
For 1991, the vertical band for $s=0$ to the left of the heatmap is quite red ($LO(0,t)>0$), implying that it is much more likely that women in the \emph{West} compared to the \emph{East} stop working in the presence of a child, relative to all other shares. This holds for both child ages \emph{0-6} (top panel) and \emph{7-18} (bottom panel). However, the entire heatmap shows positive [negative] values above [below] the 45-degree-line implying that the shift to lower shares compared to higher shares in the presence of children is stronger in the \emph{West} than in the \emph{East}, with an even larger  West-East gap in the child penalty for ages \emph{7-18}.

%% Old version: difference in differences for the childhood penalty
%\begin{figure}[H]
%\begin{center}
%\begin{minipage}{0.92\textwidth}
%\includegraphics[width=0.49\textwidth]{Images/estimated_childhood_penalty_diff_in_diff_pdf.pdf}
%\includegraphics[width=0.49\textwidth]{Images/estimated_childhood_penalty_diff_in_diff.pdf}
%\end{minipage}
%\begin{minipage}{0.069\textwidth}
%\includegraphics[width=1.2\textwidth]{Images/year_legend_v_diff_in_diff.pdf}
%\end{minipage}
%%\includegraphics[width=0.059\textwidth]{Images/year_legend_v_diff_in_diff.pdf}
%\end{center}
%\vspace{-0.5cm}
%\caption{Difference in differences for the childhood penalty in \emph{old} compared to \emph{new} federal states for \emph{c\_age} 0-6 and 7-18 for the years 1991, 2003 and 2016 on Bayes-level [left] and clr-level [right]. \label{estimated_penalty}}
%%{\bf BF: This figure should involve two separate figures, both showing clr-effects.}
%\end{figure}

% New: Odds for the difference in differences for the childhood penalty as heatmap

%% Version 1: 1991, 2003, 2016
%\begin{figure}[H]
%\begin{center}
%\includegraphics[width=\textwidth]{Images/estimated_cp_diff_in_diff_odds_with_2003.pdf}
%\\[0.1cm]
%\includegraphics[width=0.45\textwidth]{Images/cp_diff_in_diff_legend_h.pdf}
%\end{center}
%\vspace{-0.5cm}
%\caption{Odds of the Difference in differences for the childhood penalty in %\emph{old} compared to \emph{new} federal states for \emph{c\_age} 0-6 [top] and 7-18 [bottom] for the years 1991 [left], 2003 [middle], and 2016 [right]. %\label{estimated_penalty}}
%\end{figure}

% Version 2: 1991, 2016
%\vspace{-1cm}
\begin{figure}%[H]
\begin{center}
\begin{minipage}{0.92\textwidth}
\includegraphics[width=0.9\textwidth]{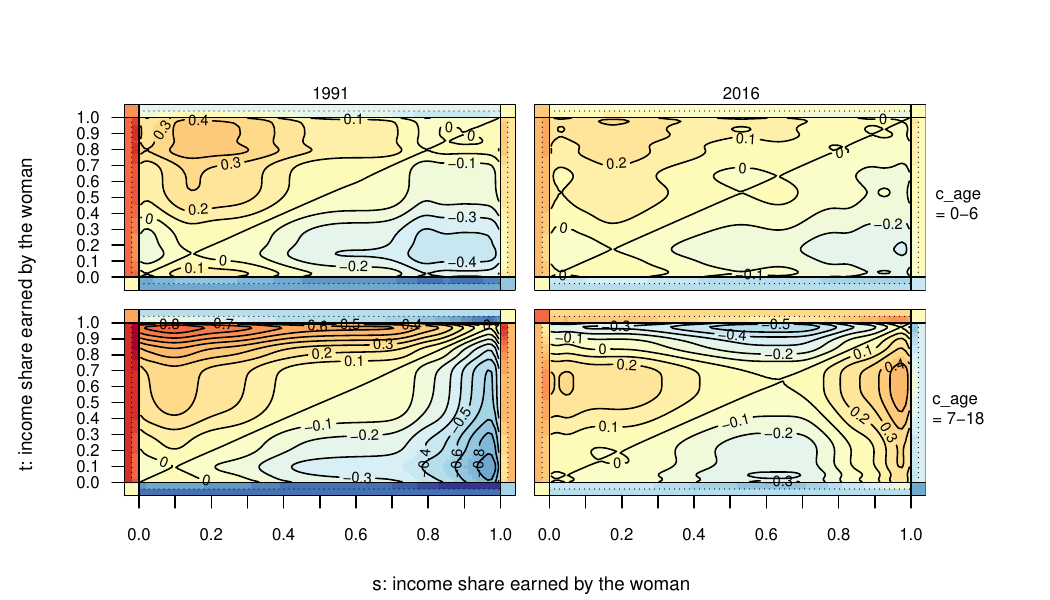}
\end{minipage}
\begin{minipage}{0.069\textwidth}
\includegraphics[width=\textwidth]{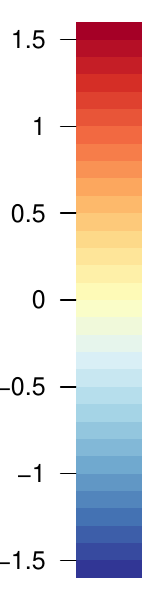}
\end{minipage}
\end{center}
\vspace{-0.5cm}
\caption{Log Odds $LO_{\text{\emph{c\_age, year}}}(s, t)$ of the West-East gap in the childhood penalty (DiD effects) for child age (\emph{c\_age}) \emph{0-6} [top] and \emph{7-18} [bottom] for the \emph{years} 1991 [left] and 2016 [right]. \label{estimated_penalty}}
\end{figure}

The comparison between the two years is informative about the change in the West-East gap in the childhood penalty over time. In 2016, the childhood penalty remains larger in the \emph{West} compared to the \emph{East} over almost the entire share distribution -- only for child ages \emph{7-18} is there  a reversal for very large shares compared to medium share levels. However, since the absolute log odds  have become much smaller, especially for non-working women, the West-East gap in the childhood penalty has decreased considerably over time.

Summarizing our main findings, the frequency of non-working women and women with a lower income share is higher in \emph{West} Germany than in \emph{East} Germany and these differences are larger for couples with children. Over time, the share of non-working women decreased. Among dual-earner households the dispersion of the share distribution increased over time with both a growing lower and higher tail. Despite persistent East-West differences in the share distributions and the child penalty until the end of the observation period, the West-East gap in the childhood penalty fell considerably over time.

%However, in general, the densities tend to show more probability mass at smaller shares and strongly decrease for share values larger than the mode at ca. $0.4-0.5$. The childhood penalty remains larger in West Germany than in East Germany and it diverging further over time for couples with a young child.
%Neither for West Germany nor for East Germany, there is evidence for a discontinuous drop in the share density above 0.5 as reported by \citet{bertrand2015} for the U.S. and by \citet{sprengholz2020} for West Germany.

%auto-ignore
\section{Simulation study}\label{chapter_simulation}
The gradient boosting approach has already been tested extensively in several simulation studies for scalar and functional data (e.g., \citet{%buehlmann2007, schmid2008, fenske2011, 
brh2015} and references therein). % Fenske: Quantile Regression 
For completeness and to validate our modified approach for density-on-scalar models, we present a small simulation study for this case.
It is based on the results of our analysis %of the woman's share in a couple's total labor income for couples in Germany, see
in Section~\ref{chapter_application}.
The predictions obtained there serve as true mean response densities for the simulation and are denoted by $F_i \in \B, ~ i = 1, \ldots , 552$, where each $i$ corresponds to one combination of values for the covariates \emph{region}, \emph{c\_age}, and \emph{year} and $\B$ is the Bayes Hilbert space from Section~\ref{chapter_application}.
To simulate data, we perform a functional principal component (PC) analysis \citep[e.g.][]{ramsay2005} on the clr transformed functional residuals $\clr [\hat{\varepsilon}_i] = \clr [f_i \ominus F_i] = \clr [f_i] - \clr [F_i]$, with $f_i \in \B$ the response densities from the application. %for $i = 1, \ldots , 552$.
Let $\psi_m$ denote the PC functions corresponding to the descending ordered eigenvalues $\xi_m$ and let $\rho_{im}$ denote the PC scores for $i = 1, \ldots , 552$ and $m \in \Nbb$.
Then, the truncated Karhunen-Lo\`eve expansion for $M \in \Nbb$ yields an approximation of the functional residuals: $\clr [\epsh_i] \approx 
%\epsh_{i,[M]}(s) = 
\sum_{m = 1}^M {\rho}_{im} \psi_m$.
The PC scores can be viewed as realizations of uncorrelated random variables $\rho_m$ with zero-mean and covariance $\Cov(\rho_{m} , \rho_{n}) = \xi_m \delta_{mn}$, where $\delta_{mn}$ denotes the Kronecker delta and $m, n = 1, \ldots , M$.
We simulate residuals $\epst_{i}$ by drawing uncorrelated random
%random numbers
$\tilde{\rho}_{im}$ from % a multivariate normal distribution with 
mean zero normal distributions with variance $\xi_m$  %covariance matrix $\left( \xi_m \delta_{mn} \right)_{m, n = 1, \ldots , M}$, calculating 
and applying the inverse clr transformation to the truncated Karhunen-Lo\`eve expansion, %i.e.,
%\begin{align*}
$
\epst_{i} 
%= \epst_{i,[M]} 
= \clr^{-1} %\biggl[
[\,\sum_{m = 1}^M \tilde{\rho}_{im} \psi_m \,] %\biggr]
= \bigoplus_{m = 1}^M \tilde{\rho}_{im} \odot \clr^{-1} \left[\psi_m\right]
.$
%\end{align*}
Adding these to the mean response densities yields the simulated data: $\ft_i = F_i \oplus \epst_{i},~ i = 1, \ldots , 552$.
Using these as observed response densities, we then estimate model~\eqref{soep_model} and denote the resulting predictions with $\fh_i \in \B, ~i = 1, \ldots , 552$.
We replicate this approach 200 times with $M = 102$, which is the maximal possible value due to the number of available grid points per density.
%There, we used $100$ equidistant points to evaluate the densities on $(0,1 )$. Adding the two boundary values this yields vectors of length $102$.

%\vspace{-0.27cm}
\begin{figure}[h] %[H]
\begin{center}
\includegraphics[width=0.49\textwidth]{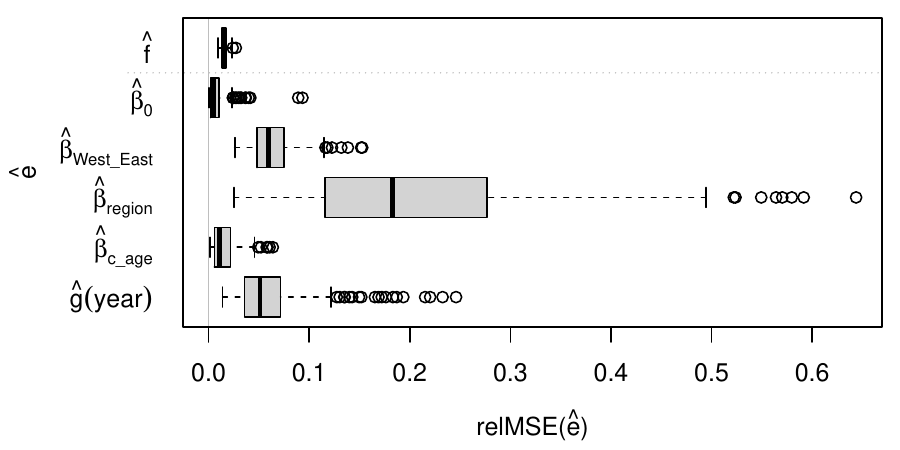}
\end{center}
\vspace{-0.5cm} %-0.8
\caption{RelMSE %[left] and mean squared errors [right] 
for prediction $\fh$ [top] and main effects [bottom]. \label{relMSE_main_effects}}
\end{figure}

To evaluate the goodness of the estimation results, we use the relative mean squared error (relMSE; defined in appendix~G.1) % \ref{appendix_simulation_relmse_def}) 
motivated by \citet{brh2015}, % where the mean squared error is standardized with respect to the global variability of the true density. 
standardizing the mean squared error with respect to the global variability of the true density. 
%See appendix~\ref{appendix_simulation_relmse_def} for the definition.
%
%%% Im Anhang:
%For predictions and estimated partial effects it is defined as
%\begin{align*}
%\rMSE (\eh) 
%:= \frac{\frac{1}{\upsilon (\Ycal)} \, \int_{\Ycal} \Vert E(y) \ominus \eh(y) \Vertb^2 \, \dups (y)}{\frac{1}{\upsilon (\Ycal)} \, \int_{\Ycal} \Vert E(y) \ominus \bar{E} \Vertb^2 \, \dups (y)}
%= \frac{\int_{\Ycal} \Vert E(y) \ominus \eh(y) \Vertb^2 \, \dups (y)}{\int_{\Ycal} \Vert E(y) \Vertb^2 \, \dups (y)}.
%\end{align*}
%Here, $\Ycal$ denotes the set $\{1, \ldots, 552\}$ for predictions, the set of possible values for categorical covariates (group-specific effects), e.g., $\{ \emph{old}, \emph{new} \}$ for the covariate \emph{old\_new}, or the observed range for scalar covariates (linear/flexible effects), e.g., $[1984, 2016]$ for \emph{year}.
%For effects depending on more than one covariate, $\Ycal$ is the Cartesian product of the appropriate sets.
%The measure $\upsilon$ is the counting measure, the Lebesgue measure, or a product measure thereof, respectively.
%The estimated densities are denoted by $\eh (y) \in \B$ for $y \in \Ycal$, corresponding to $\fh_i = \fh (i), i \in \Ycal$ for predictions or $\hh_j(\xf), \xf \in \Ycal$ for estimated effects.
%Analogously, the true densities are denoted by $E(y)$.
%Their overall mean, $\bar{E} := 1 / \upsilon (\Ycal) % \frac{1}{\upsilon (\Ycal)} 
%\int_{\Ycal} \int_\Tcal E(y) \, \dmu \, \dups (y)$, is $0 \in \B$ as a constant. 
%
Figure~\ref{relMSE_main_effects} shows the boxplots of the %obtained 
relMSEs (200 each) of the predictions and the main effects. %, i.e., all effects 
%(depending on only one covariate).
%The relMSEs of all 
All effects are illustrated in appendix~G.2. % \ref{appendix_simulation_relmse}.
%fivenum(relMSE_pred)
%[1] 0.00954609 0.01346093 0.01548987 0.01738398 0.02708743
%> boxplot_values[, main_effects]
%          intercept    old_new    region       c_age       year
%min     0.000313252 0.02601734 0.0252146 0.001595208 0.01405821
%1st Qtl 0.001952507 0.04805995 0.1159781 0.005733396 0.03590515
%median  0.004782292 0.05962976 0.1828359 0.011013931 0.05119193
%3rd Qtl 0.010591949 0.07524567 0.2765868 0.021791648 0.07158751
%max     0.093311618 0.15273295 0.6436599 0.063856470 0.24601251
The distribution of $\rMSE (\fh)$ over the 200 simulation runs shows good estimation quality, % ranging from $0.95\%$ to $2.71\%$ 
with a median of $1.55\%$. 
%The results are mainly good, apart from the region effects.
Regarding the main effects, the relMSEs are the smallest for $\betah_0$ and $\betah_{\text{\emph{c\_age}}}$ with medians of $0.48\%$ and $1.1\%$, respectively.
For $\betah_{\text{\emph{West\_East}}}$ and $\gh(year)$, the values tend to be slightly larger (medians: $5.96\%$ and $5.12\%$) while they are clearly larger for $\betah_{region}$ (median: $18.28\%$).
However, the larger relative values, especially for $\betah_{region}$, arise from the variability of the true effects being small, not from the %absolute 
mean squared errors being large. % (which is considerably larger for $\beta_0$ and $\beta_{\text{\emph{c\_age}}}$), 
This is also the case for the interaction effects,
see appendix~G.2. % \ref{appendix_simulation_relmse}.
%Recall that $\betah_{region}$ is centered around $\betah_{\text{\emph{old\_new}}}$, i.e., it is reasonable that the effect of \emph{region} is small as the majority of the effect is covered by $\betah_{\text{\emph{old\_new}}}$, already.
%This is also the case for the interaction effects, see appendix~\ref{appendix_simulation}.
%
% Effects not selected in continuous models:
%    8         8 9     9   7 8     4     7   4 7 8   4 8 
%  118    60     8     5     3     2     2       1     1
% Effects not selected in discrete models:
%       4   4 7   7   9 
% 163  33     2   1   1 
% Effect number correspondence:
% 4: old_new * year
% 7: c_age * year
% 8: c_age * old_new
% 9: c_age * old_new * year
Regarding model selection, the main effects are all selected in each simulation run, while the smaller interaction effects are not, see appendix~G.3 % \ref{appendix_simulation_selected_effects} 
for details.
%Most noticeable, $\beta_{\text{c\_age, old\_new}}$ is not selected in 131 of the 200 simulation runs in the continuous component, but in all simulation runs in the discrete component. 
%
Overall, the estimates capture the true means $F_i$ and all effects that are pronounced very well. 
Small effects in the model are estimated well in absolute, but badly in relative terms.

% Old: relMSE for each category:
%> boxplot_values[, main_effects]
%          intercept old        new region: northwest region: west region: southwest region: south
%min     0.000313252  NA 0.02601734       0.008562313   0.01575692       0.005297384   0.001391451
%1st Qtl 0.001952507  NA 0.04805995       0.044330783   0.07587011       0.044548384   0.027723336
%median  0.004782292  NA 0.05962976       0.129658812   0.20213072       0.139436135   0.095483314
%3rd Qtl 0.010591949  NA 0.07524567       0.362923384   0.49698807       0.278622465   0.209574508
%max     0.093311618  NA 0.15273295       1.366091589   2.33454017       0.990001865   0.754368883
%        region: east region: northeast child group: 1 child group: 2 child group: 3       year
%min       0.02403576        0.02403576    0.000488827    0.001928015             NA 0.01405821
%1st Qtl   0.14533377        0.14533377    0.003893915    0.008214378             NA 0.03590515
%median    0.28691535        0.28691535    0.008507411    0.015625027             NA 0.05119193
%3rd Qtl   0.62712164        0.62712164    0.016941328    0.032717686             NA 0.07158751
%max       3.18614948        3.18614948    0.052090454    0.153530614             NA 0.24601251

%> denominator$region_j
%[1] 0.10900272 0.05795085 0.16532016 0.24526999 0.01812420 0.01812420 
%auto-ignore
\section{Conclusion}\label{chapter_conclusion}
We presented a flexible framework for density-on-scalar regression models, formulating them in a Bayes Hilbert space $B^2(\mu)$, which %takes the relevant properties 
respects the nature of probability densities %into account 
and allows for a unified treatment of arbitrary finite measure spaces.
This covers in particular the common discrete, continuous, and mixed density cases.
To estimate the covariate effects in $\B$, we introduced a gradient boosting algorithm.
Furthermore, we developed several properties of Bayes Hilbert spaces related to subcompositional coherence, which are helpful for interpretation and highlight the consistency of (different possible sub-analyses within) our framework. % , emphasizing that formulating our regression models therein provides an overall consistent framework.
%We used our approach to analyze the distribution of the woman's share in a couple's total labor income, where we showed how to deal with the challenge of a mixed reference measure, using a decomposition into a continuous and a discrete estimation problem.
We used our approach to analyze the distribution of the woman's share in a couple's total labor income, an example of the challenging mixed case, for which we developed a decomposition into a continuous and a discrete estimation problem.
We observe strong differences between West and East Germany and between couples with and without children.
Among dual-earner households the dispersion of the share distribution increased over time. Despite persistent East-West differences in the share distributions and the child penalty until the end of the observation period, the West-East gap in the childhood penalty fell considerably over time.
Finally, we performed a small simulation study justifying our approach in a setting motivated by our application.

Density regression has particular advantages in terms of interpretation compared to approaches considering equivalent functions like quantile functions (e.g., \citealp{yang2018, koenker2005}) or distribution functions (CTMs, e.g., \citealp{hothorn2014}; distribution regression, e.g., \citealp{cherno2013}), as shifts in probability masses or bimodality are easily visible in densities.
Odds-ratio-type interpretations of effect functions further add to the interpretability of our model.
A crucial part in our approach is played by the clr transformation, which simplifies among other things estimation, as gradient boosting can be performed equivalently on the clr transformed densities in $L_0^2(\mu)$.
\ifnum\value{jasa}=1
{This allows taking advantage of and extending existing implementations for function-on-scalar regression like the \texttt{R} add-on package \texttt{FDboost} \citep{FDboost}, see also our vignette in the supplementary material.} 
\else
{This allows taking advantage of and extending existing implementations for function-on-scalar regression like the \texttt{R} add-on package \texttt{FDboost} \citep{FDboost}, see the github repository \href{https://github.com/boost-R/FDboost}{\emph{FDboost}} for our enhanced version of the package and in particular our vignette ``density-on-scalar\_birth''.} 
\fi
The idea to transform the densities to (a subspace of) the well-known $L^2$ space with its metric is also used by other approaches.
Besides the clr transformation, the square root velocity transformation \citep{srivastava2007} as well as the log hazard and log quantile density transformations (e.g., \citealp{%petersen2016, 
han2020}) are popular choices.
The approach of \citet{petersen2019} does not use a transformation, but also computes the applied Wasserstein metric via the $L^2$ metric.
What is special about the clr transformation based Bayes Hilbert space approach, is the embedding of the untransformed densities in a Hilbert space structure.
It is the extension of the well-established Aitchison geometry \citep{aitchison1986}, which provides an appropriate framework for compositional data -- the discrete equivalent of densities -- fulfilling appealing properties like subcompositional coherence.
The clr transformation helps to conveniently interpret covariate effects via ratios of density values (odds-ratios), which approximate or are equal to ratios of probabilities in three common cases (discrete, continuous, mixed).
Modeling those three cases in a unified framework is a novelty to the best of the authors' knowledge, and a contribution of our approach to the literature on density regression.

In this work, we only considered scalar covariates, motivated by our application, but extensions to further model terms e.g. for 
%However, extending our approach to 
functional covariates should be possible building %seems to be straightforward since we build 
on \citet{brh2015}. %, who allow for functional covariates as well.
Due to the gradient boosting algorithm used for estimation, our method includes variable selection and regularization, while it can deal with %a large number of 
numerous covariates.
%
%On the other hand,
However, like all gradient boosting approaches, it is limited by not naturally yielding inference -- unlike some existing approaches (e.g., \citealp{petersen2019}).
This might be developed using a bootstrap-based approach or selective inference \citep{rugamer2020} in the future.
Alternatively, other estimation methods for our proposed models allowing for formal inference could be derived.

The (current) definition of Bayes Hilbert spaces, which only allows finite reference measures, does not cover the interesting case of the measurable space $(\Rbb, \mathfrak{B}_{\Rbb})$ with Lebesgue measure~$\lambda$.
Though $(\Rbb, \mathfrak{B}_{\Rbb})$ can still be considered using, e.g., the probability measure corresponding to the standard normal distribution \citep{vdb2014} as reference, it would be desirable to extend Bayes Hilbert spaces to $\sigma$-finite reference measures, allowing for $B^2(\Rbb, \mathfrak{B}_{\Rbb}, \lambda)$.
Moreover, Bayes Hilbert spaces include only ($\mu$-a.e.) positive densities. 
While in the continuous case, values of zero can in many cases be avoided using a suitable density estimation method, they are often replaced with small values in the discrete case (see \citealp{pawlowsky2015}).
In contrast, the square root velocity transformation \citep{srivastava2007} allows density values of zero and may be an alternative in such cases, at the price of loosing the Hilbert space structure for the untransformed densities and subcompositional coherence. %, but has not been used for regression, yet, to the best of the authors' knowledge.

Finally, while in practice densities are sometimes directly reported%as in our example using administrative data on birth distributions over the 12 months in our software vignette
, one often does not observe the response densities directly, but has to first estimate them from individual data to enable the use of density-on-scalar regression.
%(This is related to the setting of count compositions in the discrete case, e.g., \citealp{vdb2013}.)
This can cause two problems.
First, when treating estimated densities as observed, like also in other approaches such as~\citet{petersen2019, han2020}, estimation uncertainty is not accounted for in the analysis.
Second, the number of individual observations for each covariate value combination which is available for density estimation can limit the number of covariates that can be included in the model. % (even though our algorithm is capable of working with large numbers of covariates).
In the future, we thus aim to extend our approach to also model conditional densities for individual observations, transferring our flexibility of covariate effects to allow flexible density regression without requiring  restrictive parametric  assumptions such as a particular distribution family in GAMLSS \citep{rigby2005}.
%Further extensions of interest of our Bayes Hilbert space approach include the handling of multivariate densities and the scalar-on-density regression case.

%% Supplementary Material
%\bigskip
%\begin{center}
%{\large\bf SUPPLEMENTARY MATERIAL}
%\end{center}
%
%\begin{description}
%
%\item[Appendix:] Appendix containing additional material, including all proofs. (.pdf file)
%
%%\item[R-package for  MYNEW routine:] R-package MYNEW containing code to perform the diagnostic methods described in the article. The package also contains all datasets used as examples in the article. (GNU zipped tar file)
%%
%%\item[HIV data set:] Data set used in the illustration of MYNEW method in Section~ 3.2. (.txt file)
%
%\end{description}

%\pagebreak
% Literaturverzeichnis
\addcontentsline{toc}{section}{References}
\printbibliography[heading=bibliography]
\end{refsection}

\pagebreak
\appendix
\begin{center}
    {\LARGE\bf APPENDIX}
\end{center}
\begin{refsection}
%auto-ignore
\section{Bayes Hilbert space fundamentals}\label{appendix_bayes_hilbert_space}

%Kürzpotenzial: Subsections raus

We briefly introduce Bayes spaces and summarize their basic vector space properties for a $\sigma$-finite reference measure as described in~\citet{vdb2010}.
Refining these to Bayes Hilbert spaces \citep{vdb2014}, we have to restrict ourselves to finite reference measures.
%The construction of Bayes spaces is carried out for a $\sigma$-finite reference measure.
%Progressing to Bayes Hilbert spaces this has to be restricted to a finite measure.
%%Progressing to Bayes Hilbert spaces, we have to restrict this to a finite measure.
%Note that proofs for all theorems are provided in appendix~\ref{appendix_proofs_BHS} as they go beyond those in \citet{vdb2010, vdb2014} in some instances.
We provide proofs for all theorems for completness, %in appendix~\ref{appendix_proofs_BHS}
taking a slightly different point of view compared to \citet{vdb2010, vdb2014}.

%\subsection{Introducing the Bayes space}\label{chapter_bayes_space_definition}

Let $(\Tcal, \Acal)$ be a measurable space and $\mu$ a $\sigma$-finite measure on it, the so-called \emph{reference measure}.
Consider the set $\Mcal(\Tcal, \Acal, \mu)$, or short $\Mcal(\mu)$, of $\sigma$-finite measures % equivalent to $\mu$, i.e., all measures 
with the same null sets as $\mu$.
Such measures are mutually absolutely continuous to each other, i.e., by Radon-Nikodyms' theorem, the $\mu$-density of $\nu$ or Radon-Nikodym derivative of $\nu$ with respect to $\mu$, $f_\nu := \dnu / \dmu %\frac{\dnu}{\dmu} 
: \Tcal \ra \Rbb$, exists for every $\nu \in \Mcal(\mu)$. 
It is $\mu$-almost every\-where ($\mu$-a.e.) positive and unique.
We write $\fnu \cong \nu$ for a measure $\nu \in \Mcal (\mu)$ and its corresponding $\mu$-density~$\fnu$.
For measures $\nu_1, \nu_2 \in \Mcal (\mu)$, let the equivalence relation $\isb$ be given by $\nu_1 =_{\Bcal} \nu_2$, iff there is a $c > 0$ such that $\nu_1 (A) = c \,\nu_2 (A)$ for every $A \in \Acal$, where $c \, (+ \infty) = + \infty$.
Respectively, we define $f_{\nu_1} \isb f_{\nu_2}$, iff $f_{\nu_1} = c \, f_{\nu_2}$ for some $c > 0$.
Here and in the following, pointwise identities have to be understood $\mu$-a.e. 
Both definitions of % the equivalence relation 
$\isb$ are compatible with the Radon-Nikodym identification $\fnu \cong \nu$. 
The set of $(=_{\Bcal})$-equivalence classes is called the \emph{Bayes space (with reference measure $\mu$)}, denoted by $\Bcal(\mu) = \Bcal(\Tcal, \Acal, \mu)$.
For equivalence classes containing finite measures, we choose the respective probability measure as representative in practice. 
Then, the corresponding $\mu$-density is a probability density. 
However, mathematically it is more convenient to use a non-normalized representative.
For better readability, we omit the index $\Bcal$ in $\isb$ and the square brackets denoting equivalence classes in the following. %, writing just $\nu \in \Bcal (\mu)$ instead of $[\nu] \in \Bcal (\mu)$.
%We also say $\fnu \in \Bcal$ for $\fnu \cong \nu \in \Bcal$.
For $f_{\nu_1} \cong \nu_1, f_{\nu_2} \cong \nu_2 \in \Bcal(\mu)$, the addition or \emph{perturbation} is given by the equivalent definitions
\begin{align*}
(\nu_1 \oplus \nu_2) (A) &:= \int_A \frac{\dnu_1}{\dmu} \, \frac{\dnu_2}{\dmu} \, \dmu ,  & & %\forall A \in \Acal.
%
%\intertext{Equivalently, it can be defined for the $\mu$-densities $f_{\nu_1}$ and $f_{\nu_2}$ by}
%
f_{\nu_1} \oplus f_{\nu_2} := f_{\nu_1} \, f_{\nu_2}. %& &~ \mu\text{-a.e.}
\end{align*}
For $f_{\nu} \cong \nu \in \Bcal (\mu)$ and $\alpha \in \Rbb$, the scalar multiplication or \emph{powering} is defined by
\begin{align*}
(\alpha \odot \nu ) (A) &:= \int_A \left( \frac{\dnu}{\dmu} \right)^\alpha \, \dmu , & & %\forall A \in \Acal.
%
%\intertext{The equivalent definition for the $\mu$-density $f_\nu$ is}
%
\alpha \odot f_\nu := (f_\nu)^\alpha. %& & \mu\text{-a.e.}
\end{align*}

\begin{thm}[{\citealp{vdb2010}}]\label{B_vectorspace} %[Theorem~5]
The Bayes space $\Bcal(\mu)$ with perturbation $\oplus$ and powering $\odot$ is a real vector space with additive neutral element $0 := \mu \cong 1$, additive inverse element $\ominus \nu := \int_A \dmu / \dnu %\frac{\dmu}{\dnu} 
\, \dmu \cong 1 / f_\nu % \frac{1}{f_\nu}
$ for $\nu \in \Bcal(\mu)$, and multiplicative neutral element $1 \in \Rbb$.
\end{thm}

%\textcolor{red}{Annotation (not meant to be in final paper): Theorem~\ref{B_vectorspace} is the same as \citet[Theorem~5]{vdb2010}. However, in their proof they don't show that $\nu_1 \oplus \nu_2, \alpha \odot \nu \in \Bcal(\mu)$, i.e., that the measures are  $\sigma$-finite and have the same null sets as $\mu$.
	%They prove in Theorem 4 that both measures are $\sigma$-finite.
	%% (however, they don't unite with $\Tcal_\infty$, so actually, their constructions do not increase to $\Omega$)
	%But it isn't shown that both measures have the same null sets as $\mu$.}
\begin{proof}%[Proof of Theorem~\ref{B_vectorspace}]
	This theorem equals~\citet[Theorem~5]{vdb2010}, where a brief proof is provided in the appendix.
	%However, this proof is rather short, skipping some of the details.
	%However, in their proof they don't show that $\nu_1 \oplus \nu_2, \alpha \odot \nu \in \Bcal(\mu)$ for $\nu_1, \nu_2, \nu \in \Bcal(\mu)$ and $\alpha \in \Rbb$.
	%For completeness we thus give here
	We give an alternative proof showing first that $\Mcal(\mu)$ is a vector space with perturbation and powering analogously defined.
	For this purpose, let $\nu_1, \nu_2 \in \Mcal(\mu)$ be measures
	%with $f_{\nu_1} \cong \nu_1, ~ f_{\nu_2} \cong \nu_2$
	and let $\alpha \in \Rbb$ be a scalar.
	The vector space axioms, i.e., $\Mcal(\mu)$ is an Abelian group with respect to $\oplus$, distributivity of $\oplus$ and $\odot$, associativity of $\odot$, and $1 \odot \nu = \nu$ for all $\nu \in \Mcal(\mu)$, are straightforward calculations.
	Thus, we content ourselves with showing that $\nu_1 \oplus \nu_2, \alpha \odot \nu \in \Mcal(\mu)$, which requires some more work.
	To see this, two properties have to be verified: the resulting measures have to be $\sigma$-finite and have the same null sets as $\mu$.
	The former is shown in the proof of Theorem 4 in appendix A of \citet{vdb2010}.
	To show that both $\nu_1 \oplus \nu_2$ and $\alpha \odot \nu$ have the same null sets as $\mu$, we first show  that for every $A \in \Acal$ and every $f: \Tcal \ra \Rbb_0^+$, the implication
	\begin{align}
		\bigl( f > 0 %\mu\text{-a.e.}
		~ \wedge ~ \int_A f \, \dmu = 0 \bigr) ~ \Ra ~ \mu (A) = 0 \label{positive_fuction_integral_zero_then_measure_zero}
	\end{align}
	is true.
	Let $f$ be a function that fulfills the properties on the left side of the implication and let $A \in \Acal$.
	For the sets $A_0 := \left\{f \geq 1 \right\} \cap A$ and $A_n := \{ \frac{1}{n+1} \leq f < \frac1{n} \} \cap A$, we get $A = \bigsqcup_{n \in \Nbb_0} A_n$.
	Moreover, for every $n \in \Nbb_0$, we have
	\begin{align}
		\int_{A_n} f \, \dmu \geq \int_{A_n} \frac1{n+1} \, \dmu = \frac1{n+1} \, \mu (A_n).
		\label{estimate_integral}
	\end{align}
	Now, assume that $\mu (A) \neq 0$, i.e., $\mu (A) > 0$. Then, there exists an $m \in \Nbb_0$ such that $\mu (A_{m}) >0$, because $\mu (A) = \sum_{n \in \Nbb_0} \mu (A_n)$. Thus, the inequality
	\begin{align*}
		\int_A f\, \dmu \geq \int_{A_{m}} f \, \dmu  \overset{(\ref{estimate_integral})}{\geq} \frac1{m+1} \, \mu (A_{m}) > 0
	\end{align*}
	holds.
	This is a contradiction to the hypothesis that $\int_A f \, = 0$, which proves implication~(\ref{positive_fuction_integral_zero_then_measure_zero}).
	
	Thereby, it is easy to show that $\nu_1 \oplus \nu_2$ and $\alpha \odot \nu$ have the same null sets as $\mu$:
	Let $A \in \Acal$ such that $0 = (\nu_1 \oplus \nu_2) (A) = \int_A f_{\nu_1} \, f_{\nu_2} \, \dmu$.
	We have $f_{\nu_1} \, f_{\nu_2} > 0$. %$\mu$-a.e.
	Using Equation~(\ref{positive_fuction_integral_zero_then_measure_zero}), we immediately get $\mu (A) = 0$.
	Analogously, we have $\left(f_{\nu}\right)^\alpha > 0$ for every $\alpha \in \Rbb$. %$\mu$-a.e.
	With Equation~(\ref{positive_fuction_integral_zero_then_measure_zero}) it follows $\mu (A) = 0$, if $(\alpha \odot \nu)(A) = 0$ for all $A \in \Acal$.
	The converse implications are trivial in both cases.
	This proves that $\nu_1 \oplus \nu_2, \alpha \odot \nu \in \Mcal(\mu)$ and thus, $\Mcal(\mu)$ is a real vector space with operations~$\oplus$ and~$\odot$.
	
	It remains to prove that also $\Bcal(\mu)$ is a real vector space. One easily shows that the set $[\mu]$ is a vector subspace of $\Mcal (\mu)$. Furthermore, the relation~$\isb$ defines an equivalence relation on $\Mcal(\mu)$ satisfying $\nu_1 \ominus \nu_2 \in [\mu]$ if and only if $\nu_1 \isb \nu_2$ for $\nu_1, \nu_2 \in \Mcal(\mu)$. By elementary linear algebra
	%, see, e.g.,  Theorem~2.2.7 in \citet{bosch},
	it follows that $\Bcal(\mu) = \nicefrac{\Mcal(\mu)}{[\mu]}$ is a vector space with respect to the evident quotient operations $\oplus$ and $\odot$.
\end{proof}

For subtraction, we write $\nu_1 \ominus \nu_2 := \nu_1 \oplus (\ominus \nu_2)$ and $f_{\nu_1} \ominus f_{\nu_2} := f_{\nu_1} \oplus (\ominus f_{\nu_2})$. % in the following. 

%\subsection{Introducing the Bayes Hilbert space}\label{chapter_bayes_hilbert_space_definition}

From now on, we restrict the reference measure $\mu$ to be finite, progressing to Bayes Hilbert spaces.
This is similar to \citet{vdb2014} with some details different.
In the style of the well-known $L^p$ spaces, $B^p$ spaces for $1 \leq p < \infty$ are defined as
\begin{align*}
B^p(\mu) 
= B^p(\Tcal, \Acal, \mu) 
:= \Bigl\{ \nu \in \Bcal (\mu) ~\Big|~ \int_{\Tcal} \bigl| \log \frac{\dnu}{\dmu}\bigr|^p \, \dmu < \infty \Bigr\}.
\end{align*}
We also say $f_\nu \in B^p(\mu)$ for $f_\nu \cong \nu \in \bpl$.
%The statement $\nu \in \bpl$ 
This is equivalent to $\log f_\nu \in L^p(\mu)$,
%Or, put differently, we have $\ft \in \lpl$, iff $\exp \ft \in \bpl$. Therefore, we have
which gives us $B^q(\mu) \subset B^p(\mu)$ for $p,q \in \Rbb$ with $1\leq p<q$.
%\begin{prop}\label{Bp_vectorspace}
Note that for every $p \in \Rbb$ with $1 \leq p < \infty$, the space $\bpl$ is a vector subspace of $\Bcal (\mu)$, see \citet{vdb2014}. %[Proposition 1]
%\end{prop}
For $f_\nu \cong \nu \in B^p(\mu)$, the \emph{centered log-ratio (clr) transformation of $\nu$} %or \emph{clr image of $\nu$ (with reference measure $\mu$}) 
is given by
\begin{align*}
\clr_{B^p(\Tcal, \Acal, \mu)} [\nu] 
%= \clr_\mu [\nu] 
= \clr_{B^p(\Tcal, \Acal, \mu)} [\fnu] 
%= \clr_{B^p(\mu)} [\fnu] 
%= \clr_\mu [\fnu] 
:= \log f_\nu - \Scal_{B^p(\Tcal, \Acal, \mu)} (\fnu), %\label{def_clr_appendix} % \frac1{\mu (\Tcal)} \, \int_{\Tcal} \log \fnu \, \dmu , 
\end{align*}
with $\Scal_{B^p(\Tcal, \Acal, \mu)} (\fnu) := 1 / \mu (\Tcal) % \frac1{\mu (\Tcal)} 
\, \int_{\Tcal} \log \fnu \, \dmu$ the mean logarithmic integral.
%Note that it is linear on $B^2(\Tcal, \Acal, \mu)$, which is straightforward to show.
We omit the indices $B^p(\Tcal, \Acal, \mu)$ 
or shorten them to $\mu$ or $\Tcal$, if the underlying space is clear from context. 
%Now, consider $\lpnl := \left\{ \ft \in \lpl ~|~ \int_{\Tcal} \ft \, \dmu = 0 \right\}$, which is a closed subspace of $\lpl$.

%%%%%%%%%%%%%%%%%%%%%%%%%%%%%%%%%%%%%%%%%%%%%%%%

\begin{prop}[For $p = 1$ shown in \citealp{vdb2014}]
% Propositions 2, 4 and 5
\label{clr_isomorphism}
For $1 \leq p < \infty$, %the clr trans\-formation 
$\clr : \bpl \ra \lpnl := \{ \ft \in \lpl ~|~ \int_{\Tcal} \ft \, \dmu = 0 \}$ is an isomorphism with inverse trans\-formation $\clr^{-1} [\ft] = \exp \ft$.
\end{prop}

\begin{proof}%[Proof of Proposition~\ref{clr_isomorphism}]
	%This proposition is similar to \citet[Propositions 2, 4 and 5]{vdb2014}, where only $p = 1$ is considered.
	This proposition is proven in \citet[Propositions 2, 4 and 5]{vdb2014} in the case $p = 1$.
	We show the statements for arbitrary $1 \leq p < \infty$, because we need them in particular for $p=2$.
	%We also add some details.
	
	% Well-defined
	% We first prove that the clr image is well-defined.
	Let $1 \leq p < \infty$ and let $\nu \in \bpl$ be a measure.
	% with $f_{\nu} \cong \nu$.
	The integral $\int_{\Tcal} \log f_{\nu} \, \dmu$ exists because of $\log f_{\nu} \in \lpl$.
	%Let $\nu_2 \in B^p(\mu)$ be another measure with $\nu_2 \isb \nu$ and $f_{\nu_2} \cong \nu_2$.
	%Then, there exists a constant $c >0$ such that $f_{\nu_2} = c \, f_{\nu_1}$. Thus, we get
	%\begin{align*}
	%\clr [\nu_2] &=
	%\log f_{\nu_2} - \frac1{\mu (\Tcal)} \, \int_{\Tcal} \log f_{\nu_2} \, \dmu \\
	%&= \log c + \log f_{\nu_1} - \frac1{\mu (\Tcal)} \, \left( \mu (\Tcal) \, \log c + \int_{\Tcal} \log f_{\nu_1} \, \dmu \right) \\
	%&= \log f_{\nu_1} - \frac1{\mu (\Tcal)} \, \int_{\Tcal} \log f_{\nu_1} \, \dmu
	%= \clr [\nu_1] .
	%\end{align*}
	Furthermore, it is straightforward to show that for every $\nu_2 \in \bpl$ with $\nu_2 \isb \nu$ the clr images are equal, $\clr [\nu] = \clr [\nu_2]$.
	Hence, the clr image of $[\nu]$ is well-defined on $B^p(\mu)$.
	% $\clr[mu] \in \lpnl$
	Next, we show that $\clr [\nu] \in \lpnl$, which is the case if $\clr [\nu] \in \lpl$ and $\int_{\Tcal} \clr [\nu] \, \dmu = 0$.
	The first statement corresponds to $\int_{\Tcal} | \clr [\nu] |^p \, \dmu < \infty$, which is equivalent to $\Vert \clr [\nu] \Vert_{\lpl}  < \infty$.
	Using the Minkowski inequality, we get
	\begin{align*}
		\Vert \clr [\nu] \Vert_{\lpl}
		&= \left\Vert \log f_{\nu} - \Scal (f_{\nu}) \right\Vert_{\lpl}
		\leq \left\Vert \log f_{\nu} \right\Vert_{\lpl} + \left\Vert \Scal (f_{\nu}) \right\Vert_{\lpl}.
	\end{align*}
	As $\nu \in \bpl$, we have $\log f_{\nu} \in \lpl$ and therefore the first term is finite.
	%, $\left\Vert \log f_{\nu} \right\Vert_{\lpl} < \infty$
	For the second term, the function in the norm is a constant, thus it is an element of $\lpl$ since $\mu$ is finite.
	Together, we get $\Vert \clr [\nu] \Vert_{\lpl} < \infty$. Moreover,
	\begin{align*}
		\int_{\Tcal} \clr [\nu] \, \dmu
		&= \int_{\Tcal} \log f_{\nu} - \Scal (f_{\nu}) \, \dmu %\\
		% &= \int_{\Tcal} \log f_{\nu} \, \dmu - \frac1{\mu (\Tcal)} \, \int_{\Tcal} \int_{\Tcal} \log f_{\nu}\, \dmu \, \dmu \\
		= \mu(\Tcal) \, \Scal (f_{\nu}) - \mu(\Tcal) \, \Scal (f_{\nu})
		= 0.
	\end{align*}
	Hence, it follows that $\clr [\nu] \in \lpnl$.
	% \end{proof} \ref{clr_image}
% Let $1 \leq p < \infty, \alpha \in \Rbb$ and $\nu_1, \nu_2 \in \bpl$ with $f_{\nu_1} \cong \nu_1$ and $f_{\nu_2} \cong \nu_2$. Then, we have
Furthermore, the clr transformation is linear:
\begin{align*}
	\clr \left[ \alpha \odot f_{\nu} \oplus f_{\nu_2} \right]
	&= \log \left( (f_{\nu})^\alpha \, f_{\nu_2} \right) - \Scal \left( (f_{\nu})^\alpha \, f_{\nu_2} \right) \\
	%&= \alpha \, \log f_{\nu} + \log f_{\nu_2} - \frac1{\mu (\Tcal)} \int_{\Tcal} \alpha \, \log f_{\nu} + \log f_{\nu_2} \, \dmu \\
	&= \alpha \left( \log f_{\nu} - \Scal (f_{\nu}) \right) + \log f_{\nu_2} - \Scal (f_{\nu_2})
	= \alpha \, \clr [f_{\nu} ]+ \clr [f_{\nu_2} ].
\end{align*}
% Therefore, the clr transformation is linear.
It remains to show that it is bijective. For $\ft \in \lpnl$, we have
\begin{align*}
	\clr \left[ \exp \ft\right]
	&= \log \left( \exp \ft \right) - \Scal \left( \exp \ft \right)
	= \ft - \frac1{\mu (\Tcal)} \int_{\Tcal} \ft \, \dmu = \ft ,
\end{align*}
using that the last integral is zero since $\ft \in \lpnl$. Conversely, for $ f \in \bpl$, we get
\begin{align*}
	\exp \left( \clr [f] \right)
	&= \exp \left( \log f - \Scal (f_{\nu}) \right)
	= \frac{f}{\exp\left( \Scal (f) \right)}
	= f
\end{align*}
and therefore, the clr transformation is bijective.
\end{proof}

%%%%%%%%%%%%%%%%%%%%%%%%%%%%%%%%%%%%%%%%%%%%%%%%

Note that $\lpnl$ is a closed subspace of $\lpl$.
The space $B^2(\mu)$ is called the \emph{Bayes Hilbert space (with reference measure $\mu$)}.
%which is an inner product on $\B$, see Proposition~\ref{inner_product_thm} in appendix~\ref{appendix_proofs_BHS}. 

\begin{prop}\label{inner_product_thm}
The transformation %we define an inner product on it by
\begin{align*}
&\langle \nu_1 , \nu_2 \rangleb := \langle f_{\nu_1} , f_{\nu_2} \rangleb := \int_{\Tcal} \clr [f_{\nu_1}] \, \clr [f_{\nu_2}] \, \dmu, && , f_{\nu_1} \cong \nu_1, f_{\nu_2} \cong \nu_2 \in \B ,
\end{align*}
is an inner product on $\B$.
\end{prop}

\begin{proof}%[Proof of Proposition~\ref{inner_product_thm}]
	%Note that the inner product on $\B$ defined in \citet{vdb2014} differs from ours by a factor $\frac{1}{\mu (\Tcal)}$.
	%For completeness, we show that it is indeed an inner product.
	Linearity of $\langle \, , \, \rangleb$ follows from the linearity of the clr transformation, see Proposition~\ref{clr_isomorphism}, and basic calculation rules. Symmetry is obvious by the commutativity of multiplication in $\Rbb$. It remains to show that $\langle \, , \, \rangleb$ is positive definite. For this purpose, let $f_{\nu} \in \B$ be a density.
	\begin{itemize}
		\item
		We have $\langle f_{\nu} , f_{\nu} \rangleb = \int_{\Tcal} (\clr [f_{\nu}] )^2 \, \dmu \geq 0$ because the integrand is nonnegative.
		\item
		We need to show that $\langle f_{\nu} , f_{\nu} \rangleb = 0 ~\Llra~ f_{\nu} = 0$.
		\begin{itemize}
			\item[``$\Ra$'']
			If $\langle f_{\nu} , f_{\nu} \rangleb = \int_{\Tcal} (\clr [f_{\nu}] )^2 \, \dmu = 0$, then $\clr [f_{\nu}] = 0$ must hold. This is equivalent to
			$
			\log f_{\nu} = \Scal (f_{\nu}) ~ \mu\text{-almost everywhere}
			$,
			which means $\log f_{\nu}$ is a constant function. Then, $f_{\nu}$ is constant as well and thus $f_{\nu} = 0$.
			\item[``$\La$'']
			If otherwise $f_{\nu} = 0 $, then $\clr [f_{\nu}] = 0$ by linearity of the clr transformation, see Proposition~\ref{clr_isomorphism}, and therefore
			$
			\langle f_{\nu} , f_{\nu} \rangleb
			%= \int_{\Tcal} (\clr [f_{\nu}] )^2 \, \dmu
			%= \int_{\Tcal} 0 \, \dmu
			= 0.
			$ \qedhere
		\end{itemize}
	\end{itemize}
\end{proof}

% This is indeed an inner product, see appendix~\ref{appendix_proofs_BHS}. Furthermore, 
The inner product induces a norm on $\B$ by $\Vert \nu \Vertb := \Vert \fnu \Vertb := \sqrt{\langle \fnu, \fnu \rangleb}$ for $f_\nu \cong \nu \in \B$.
% with $\mu$-density $\fnu$.
By definition, we have
$ % \begin{align*}
\langle f_{\nu_1}, f_{\nu_2} \rangleb = \langle \clr [f_{\nu_1}] , \clr [f_{\nu_2}] \rangle_{L^2(\mu)},
$ % \end{align*}
which immediately implies that $\clr: \B \ra \Ln$ is isometric.
We now formulate the main statement of this section:

\begin{thm}[{\citealp{vdb2014}}]\label{bayes_hilbert_space_thm} % [Theorem~1]
The Bayes Hilbert space $\B$ is a Hilbert space.
\end{thm}

\begin{proof}%[Proof of Theorem~\ref{bayes_hilbert_space_thm}]
	%This statement was already proven in~\citet{vdb2014}.
	%We are able to provide a shorter version of this proof here, due to Proposition~\ref{clr_isomorphism} being more general.
	
	We provide an alternative proof to~\citet{vdb2014}:
	It is a known fact from functional analysis that $L^2(\mu)$ is a Hilbert space. % (e.g. \citealp[Theorem~2.15]{elstrodt2011}). % Kapitel VI, § 2, Theorem~2.15; Deutsche Quelle!
	As a closed subspace, $\Ln$ is a Hilbert space as well. As the clr transformation $\clr : \B \ra \Ln$ is isometric, it follows that also $\B$ is a Hilbert space.
\end{proof}

Note that under very modest assumptions on the measure space $(\Tcal, \Acal, \mu)$, the Hilbert spaces $L^2(\mu)$ and $L_0^2(\mu)$ are separable, see~\citet[Korollar~2.29]{elstrodt2011}. % Deutsche Quelle!
This was used in the pioneering work of~\citet{egoz2006} to construct the Bayes Hilbert space and show its separability.

%auto-ignore

\section{Proofs}\label{appendix_proofs}

\begin{proof}[Proof of Equation~(6)] % (\ref{gradient})]
This proof requires knowledge about differential calculus for real functionals. A review can be found in~\citet[Section~1.3]{funcana}.

We want to show that the negative gradient of the loss functional
\begin{align*}
\rho_{y_i} : \B \ra \Rbb, & & f_1 \mapsto \Vert y_i \ominus f_1 \Vertb^2
\end{align*}
at $f_1 \in \B$ for $y_i \in \B$ exists and determine it.
First, we show that $\rho_{y_i}$ is Fr\'echet differentiable at $f_1 \in \B$, i.e., that there exists $A \in (\B)'$ such that
\begin{align}
\lim_{\Vert f_2\Vertb \ra 0} \frac{\rhoy (f_1 \oplus f_2) - \rhoy (f_1) - A(f_2)}{\Vert f_2\Vertb} = 0, \label{frechet}
\end{align}
%The Bayes Hilbert space $\B \,(= H = X)$ is a Hilbert space (and in particular a Banach space), the loss function is a functional.
where $(\B)' := \{ A: \B \ra \Rbb ~|~ A~\text{linear and continuous}\}$ is the topological dual of $\B$.
Consider $A = A_{y_i, f_1}: \B \ra \Rbb, f_2 \mapsto \langle \ominus 2 \odot (y_i \ominus f_1), f_2\rangleb$.
Then $A \in (\B)'$ and for $f_1, f_2 \in \B$, we have
\begin{align*}
\rhoy (f_1 \oplus f_2) - \rhoy (f_1) - A(f_2)
&= \Vert y_i \ominus (f_1 \oplus f_2)\Vertb^2 - \Vert y_i \ominus f_1 \Vertb^2 \\
&\hspace{0.5cm} - \langle \ominus 2 \odot (y_i \ominus f_1), f_2\rangleb \\
&= \Vert y_i \ominus f_1 \Vertb^2 - 2 \langle y_i \ominus f_1 , f_2 \rangleb + \Vert f_2 \Vertb^2 \\
&\hspace{0.5cm} - \Vert y_i \ominus f_1 \Vertb^2
+ 2 \langle y_i \ominus f_1 , f_2 \rangleb \\
&= \Vert f_2 \Vertb^2.
\end{align*}
This implies that the limit in~(\ref{frechet}) is zero.
Thus, $\rhoy$ is Fr\'echet differentiable at $f_1 \in \B$ with differential $\mathrm{d} \rhoy (f_1) = A = A_{y_i, f_1}$.
As $\B$ is a Hilbert space, Riesz' Representation Theorem holds and the gradient of $\rhoy$ at $f_1$ is $\nabla \rhoy (f_1) = \ominus 2 \odot (y_i \ominus f_1)$.
\end{proof}

%\subsection{Proofs related to Section~\ref{chapter_subcomp}}\label{appendix_proofs_subcomp}

\begin{proof}[Proof of Proposition~3.1] % \ref{thm:oddsratio}]
\begin{enumerate}[(a)]
\item 
Let $A, B \in \Acal^+$, $m := \inf_{s\in A, t \in B} \OR(s,t)$, and $ M := \sup_{s\in A, t \in B} \OR(s,t)$. 
Then, for all $s \in A, t \in B$, we have
$m \leq \OR(s,t) = \frac{f_1(s) / f_1(t)}{f_2(s) / f_2(t)} 
\leq M$ 
and thus,
$
m \, f_1(t) f_2(s) \leq f_1(s) f_2(t)
$
and
$
f_1(s) f_2(t) \leq M \, f_1(t) f_2(s)
$.
Integrating over $A \times B$ yields
\begin{align*}
m  \int_{A \times B} f_1(t) f_2(s) \, \dmumu (s, t)
\leq \int_{A \times B} f_1(s) f_2(t) \, \dmumu (s, t)
\intertext{and}
\int_{A \times B} f_1(s) f_2(t) \, \dmumu (s, t)
\leq M \int_{A \times B} f_1(t) f_2(s) \, \dmumu (s, t).
\end{align*}
By Tonelli's Theorem all integrals factorize and we get
$m \, \Pbb_1(B) \Pbb_2(A) \leq \Pbb_1(A) \Pbb_2(B)$
and
$\Pbb_1(A) \Pbb_2(B) \leq M \, \Pbb_1(B) \Pbb_2(A)$, i.e.,
$m \leq \frac{\Pbb_1 (A) \, / \, \Pbb_1 (B)}{\Pbb_2 (A) \, / \, \Pbb_2 (B)} \leq M$.
\item 
% Consider $\Tcal = I \subset \Rbb$ and %$\mu = \lambda$ or
% $\mu = \sum_{d = 1}^D w_d \, \delta_{t_d} + \lambda$ for $\Dcal = \{t_1, \ldots , t_D\} \subset I$. %, let
Let $s \in \Tcal$ and $A_n \in \Acal^+$ be intervals such that $A_n$ is centered at $s$ for all $n \in \Nbb$, $\bigcap_{n \in \Nbb} A_n = \{s\}$ and $A_{n+1} \subset A_n,$ % and $\mu (A_n) > 0$
for $n \in \Nbb$.
%Let $s \in I$ and $A_n \subseteq I$ be a nested sequence of intervals centered at $s$ for all $n \in \Nbb$, whose intersection is $\{s\}$.
%I.e., $A_{n+1} \subset A_n,$ for all $n \in \Nbb$ and $\bigcap_{n \in \Nbb} A_n = \{s\}$.
It is sufficient to show
\begin{align}
%\frac{\Pbb_j (\{t_{d_0}\})}{\mu (\{t_{d_0}\})} = \hh_j (t_{d_0})
%&& \text{and} &&
\lim_{n \ra \infty} \frac{\Pbb_j (A_n)}{\mu (A_n)} = f_j (s) && \text{for}~ j \in \{ 1, 2 \}. % && \mu\text{-a.e.} \label{proof_approximate_probability}
\end{align}
% In all three cases we have,
% \begin{align}
% \lim_{n \ra \infty} \frac{\Pbb_j (A_n)}{\mu (A_n)}
% &= \lim_{n \ra \infty} \frac{1}{\mu(A_n)} \int_{A_n} f_j \, \dmu. \label{odds_probability}
% %\lim_{n \ra \infty} \frac{\frac{\Pbb_j (A_n)}{\mu (A_n)}}{\frac{\Pbb_j (BA_n)}{\mu (B_n)}}
% %&= \lim_{n \ra \infty} \frac{\frac{1}{\mu(A_n)} \int_{A_n} f_j \, \dmu}{\frac{1}{\mu(B_n)} \int_{B_n} fj \, \dmu}. \label{odds_probability}
% \end{align}
\begin{enumerate}[i)]
%\item
%In the discrete case, we have
%%\begin{align*}
%$\frac{\Pbb_j (\{t_{d_0}\})}{\mu (\{t_{d_0}\})}
%= \frac{w_{d_0} \, f_j (t_{d_0})}{w_{d_0}}
%= f_j (t_{d_0})$.
%%\end{align*}
\item\label{proof_approximate_probability_continuous}
In the continuous case, i.e., $\Dcal = \emptyset$, we have
\begin{align*}
\lim_{n \ra \infty} \frac{\Pbb_j (A_n)}{\lambda (A_n)}
= \lim_{n \ra \infty} \frac{1}{\lambda(A_n)} \int_{A_n} f_j \, \dlamb
= f_j (s) %&& \mu\text{-a.e.},
%\lim_{n \ra \infty} \frac{\frac{\Pbb_j (A_n)}{\mu (A_n)}}{\frac{\Pbb_j (BA_n)}{\mu (B_n)}}
%= \lim_{n \ra \infty} \frac{\frac{1}{\lambda(A_n)} \int_{A_n} \hh_j \, \dlamb}{\frac{1}{\lambda(B_n)} \int_{B_n} \hh_j \, \dlamb}
%= \frac{\hh_j (s)}{\hh_j (t)} %&& \mu\text{-a.e.},
\end{align*}
using Lebesgue's Differentiation Theorem \citep[Theorem~7.2]{wheeden2015} in the last equation.
Note that the equation holds for all $s$ in the interior of $I$ (not only $\mu$-a.e.), if $f_j$ is continuous.\footnote{In practice, this is the case, when choosing continuous basis functions $\bfe_Y$ like B-splines (for the continuous component).}
Extending $f_j$ outside of $I$ by~$0$ also yields the statement for the boundary values of $I$.
\item
In the mixed case, we have
\begin{align*}
\lim_{n \ra \infty} \frac{\Pbb_j (A_n)}{\mu (A_n)}
&= \lim_{n \ra \infty} \frac{\sum_{d = 1}^D w_d \, \delta_{t_d}(A_n) f_j (t_d) + \int_{A_n} f_j \, \dlamb}{\sum_{d = 1}^D w_d \, \delta_{t_d}(A_n) + \lambda(A_n)}.
\end{align*}
If $s \in \Dcal = \{t_1, \ldots , t_D\}$, % there exists a $d_0 \in \{1, \ldots , D\}$ with $s = t_{d_0}$, 
the term simplifies to the discrete case.
Otherwise, the term simplifies to the continuous case.
In both cases, the limit is $f_j(s)$. % ($\mu$-almost everywhere in the continuous case, see \eqref{proof_approximate_probability_continuous}).
%\qedhere
\end{enumerate}
%Statements (b) and (c) show~(\ref{odds_ratio_continuous_mixed}).
%
%If $\hh_j$ and $\hh_k$ are continuous on the interior of $I \setminus \Dcal$, then~\eqref{odds_ratio_continuous_mixed} and~\eqref{approximate_probability} hold for all values in the interior of $I$ and for the boundary values if they are in~$\Dcal$.
%%If they are not in $\Dcal$, the equations~(\ref{odds_ratio_continuous_mixed}) can still be extended to the boundary values in appendix~\ref{appendix_proofs}, where we also prove~(\ref{odds_ratio_continuous_mixed}).
%Otherwise, the densities need to be extended outside of $I$ by~$0$ to prove~\eqref{odds_ratio_continuous_mixed} and~\eqref{approximate_probability} for the boundary values of $I$.
%%Regarding $I$ without its boundaries, assuming continuity on $I$ in the continuous case or continuity on every interval $I \setminus \{t_1, \ldots, t_D\}$ is sufficient.
%In practice, the continuity assumption for the estimated partial effects is fulfilled automatically when choosing continuous basis functions $\bfe_Y$ like B-splines.
%%Then, the equations hold for all $s, t \in I$. %in the continuous and mixed cases.
\end{enumerate}
\end{proof}

\begin{proof}[Proof of Proposition~3.2] % \ref{thm_subcompositional_coherence}]
%In this proof (and also in the formulation of the proposition), we denote functions in $B^2(\tilde{\Tcal}$ with a tilde to distinguish them from functions in $B^2(\Tcal)$.
%
It is straightforward to show that $\iota$ is well-defined and linear.
Let $\tilde{f} \in B^2(\tilde{\Tcal})$ and $g \in B^2(\Tcal)$.
Preservation of the norm is implied by the more general preservation of the inner product,
$\langle \iota(\tilde{f}), g \rangle_{B^2(\Tcal)} = \langle \tilde{f}, g|_{\tilde{\Tcal}} \rangle_{B^2(\tilde{\Tcal})}$,
considering the special case $g = \iota(\tilde{f})$.
As we need the preservation of the inner product later, we show this more general property instead of just preservation of the norm.
We have
\begin{align*}
\langle \iota(\tilde{f}), g \rangle_{B^2(\Tcal)}
&= \int_{\Tcal} \clr \left[ \iota(\tilde{f}) \right]
\left( \left(\log g - \Scal_{\tilde{\Tcal}} (g|_{\tilde{\Tcal}}) \right) + \left( \Scal_{\tilde{\Tcal}} (g|_{\tilde{\Tcal}}) - \Scal_{\Tcal} (g)\right) \right) \, \dmu, % \label{proof_subcoh_orthogonality_1}
\end{align*}
where the last term $\Scal_{\tilde{\Tcal}} (g|_{\tilde{\Tcal}}) - \Scal_{\Tcal} (g)$ is constant.
Thus, it does not contribute to the integral as $\clr \left[ \iota(\smash{\tilde{f}}) \right] \in L^2_0(\Tcal)$.
%clr transformed functions integrate to zero.
By the additivity of $\mu$ we get
\begin{align}
\Scal_{\Tcal} \left(\iota(\tilde{f})\right)
&= \frac{1}{\mu(\Tcal)} \left( \int_{\tilde{\Tcal}} \log \tilde{f} \, \dmu + \int_{\Tcal \setminus \tilde{\Tcal}} \Scal_{\tilde{\Tcal}} (\tilde{f}) \, \dmu \right) %\notag \\
%%
%&= \frac{1}{\mu(\Tcal)} \left( \mu (\tilde{\Tcal}) \Scal_{\tilde{\Tcal}} (\tilde{f}) + \left(\mu(\Tcal) - \mu (\tilde{\Tcal})\right) \Scal_{\tilde{\Tcal}} (\tilde{f}) \, \dmu \right)
%
= \Scal_{\tilde{\Tcal}} (\tilde{f}) \label{proof_subcoh_S_embedding}
\end{align}
and thus
\begin{align}
\langle \iota(\tilde{f}), g \rangle_{B^2(\Tcal)}
&= \int_{\Tcal} \left( \log \iota(\tilde{f}) - \Scal_{\tilde{\Tcal}} (\tilde{f}) \right)  \left( \log g - \Scal_{\tilde{\Tcal}} (g|_{\tilde{\Tcal}}) \right) \, \dmu. \notag
\end{align}
Note that the first factor of the integrand is zero on $\Tcal \setminus \tilde{\Tcal}$ as $\iota(\ft) = \exp \Scal_{\tilde{\Tcal}} (\tilde{f})$ on this set.
This leaves us with
\begin{align}
\langle \iota(\tilde{f}), g \rangle_{B^2(\Tcal)}
%%
%= \int_{\tilde{\Tcal}} \left( \log \tilde{f} - \Scal_{\tilde{\Tcal}} (\tilde{f}) \right)
%\left( \log g|_{\tilde{\Tcal}} - \Scal_{\tilde{\Tcal}} (g|_{\tilde{\Tcal}}) \right) \, \dmu %\notag \\
%
= \int_{\tilde{\Tcal}} \clr_{\tilde{\Tcal}} \left[ \tilde{f} \right] \, \clr_{\tilde{\Tcal}} \left[ g|_{\tilde{\Tcal}} \right] \, \dmu %\notag \\
&= \langle \tilde{f}, g|_{\tilde{\Tcal}} \rangle_{B^2(\tilde{\Tcal})}, \label{proof_subcoh_orthogonality}
\end{align}
i.e., $\iota$ preserves the inner product.
In particular, $\iota$ preserves the norm and is an embedding.
Being a Hilbert space, $B^2(\tilde{\Tcal})$ is complete and thus is a closed subspace of $B^2(\Tcal)$.
For
$ %$\begin{align*}
P: B^2(\Tcal) \ra B^2(\Tcal), %&& 
~
f \mapsto \iota(f|_{\tilde{\Tcal}}),
$ %\end{align*}
%is an orthogonal projection, i.e., we have to show
we show
\begin{enumerate}[(a)]
\item%\label{proof_subcoh_selbstadjungiert}
$P^2 = P$,
\item%\label{subsomp_dominance}
$\Vert P \Vert := \sup_{f \neq 0} \frac{\Vert P(f) \Vert_{B^2(\Tcal)}}{\Vert f \Vert_{B^2(\Tcal)}} = 1$,
 \item 
 $P^\ast = P$.
\end{enumerate}
Proofs of (a)-(c):
\begin{enumerate}[(a)]
\item\label{proof_subcoh_squared_projection}
On $\tilde{\Tcal}$, the embedding $\iota$ is the identity and thus $P\left( P(f) \right) = P(f)$ holds. % $\iota(f|_{\tilde{\Tcal}})|_{\tilde{\Tcal}} = f|_{\tilde{\Tcal}}$.
%Let $f \in B^2(\Tcal)$. Then,
%\begin{align*}
%p\left( P(f) \right)
%= \iota \left( \iota(f|_{\tilde{\Tcal}})|_{\tilde{\Tcal}} \right)
%= \iota(f|_{\tilde{\Tcal}})
%= P(f)
%\end{align*}
\item
Let $f \in B^2(\Tcal)$.
First, we show %$\Vert f|_{\tilde{\Tcal}} \Vert_{B^2(\tilde{\Tcal})}^2 \leq \Vert f \Vert_{B^2(\Tcal)}^2$, which is equivalent to
$\Vert P(f) \Vert_{B^2(\Tcal)}^2 \leq \Vert f \Vert_{B^2(\Tcal)}^2$.
We have
%as $\iota$ preserves the norm.
\begin{align*}
\Vert f \Vert_{B^2(\Tcal)}^2
&= \int_{\tilde{\Tcal}} \left(\clr_{\tilde{\Tcal}} [f] + \left( \SfTn - \SfT \right) \right)^2 \, \dmu
+ \int_{\Tcal \setminus \tilde{\Tcal}} \left( \clr_{\Tcal} [f] \right)^2 \, \dmu .
\end{align*}
The first term is bounded from below by $\Vert f|_{\tilde{\Tcal}} \Vert_{B^2(\tilde{\Tcal})}^2$ since $\clr_{\tilde{\Tcal}} [f]$ is orthogonal to the constant $\SfTn - \SfT$ and the square integral of the latter is nonnegative.
Furthermore, the last term is nonnegative, i.e., $\Vert f \Vert_{B^2(\Tcal)}^2 \geq \Vert f|_{\tilde{\Tcal}} \Vert_{B^2(\tilde{\Tcal})}^2$.
As $\iota$ preserves the norm, this implies the claim.
Since $P(f) \in B^2(\Tcal)$ saturates the inequality because of~\eqref{proof_subcoh_squared_projection}
%\begin{align*}
%\Vert p(f_p) \Vert_{B^2(\Tcal)}^2
%= \Vert p\left( p (f) \right) \Vert_{B^2(\Tcal)}^2
%\overset{\eqref{proof_subcoh_squared_projection}}{=} \Vert p (f) \Vert_{B^2(\Tcal)}^2
%= \Vert f_p \Vert_{B^2(\Tcal)}^2,
%\end{align*}
we get
$\Vert P \Vert %= \sup_{f \neq 0} \frac{\Vert P(f) \Vert_{B^2(\Tcal)}}{\Vert f \Vert_{B^2(\Tcal)}}
= 1$.
\item
Let $f, g \in B^2(\Tcal)$.
%Recall that $\langle I(\tilde{f}), g \rangle_{B^2(\Tcal)} \overset{\eqref{proof_subcoh_orthogonality}}{=} \langle \tilde{f}, g|_{\tilde{\Tcal}} \rangle_{B^2(\tilde{\Tcal})}$.
Then, using the symmetry of the inner product, we have
\begin{align*}
\left\langle P(f), g \right\rangle_{B^2(\Tcal)}
&\overset{\eqref{proof_subcoh_orthogonality}}{=} \left\langle f|_{\tilde{\Tcal}}, g|_{\tilde{\Tcal}} \right\rangle_{B^2(\tilde{\Tcal})}
\overset{\eqref{proof_subcoh_orthogonality}}{=}
\left\langle f, P(g) \right\rangle_{B^2(\Tcal)}.
\end{align*}
\end{enumerate}
In particular, $P$ is an orthogonal projection.
\end{proof}

%%%%%%%%%%%%%%%%

\begin{prop}\label{thm_decomposition_bayes}
Consider a mixed Bayes Hilbert space $\B = B^2\left( \Tcal, \Acal, \mu\right)$, i.e., $\Tcal = I \cup \Dcal$, where $I \subset \Rbb$
is a nontrivial interval and $\Dcal = \{t_1, \ldots , t_D\} \subset \Rbb$, $\Acal$ is the smallest $\sigma$-algebra containing all closed subintervals of $I$ and all points of $\Dcal$, and $\mu = \delta + \lambda$, where $\delta = \sum_{d = 1}^D w_d \, \delta_{t_d}$ with %$\{t_1, \ldots , t_D\} = \Dcal \subset I$ and 
$w_d > 0$.
For $\Ccal := I \setminus \Dcal$, the orthogonal complement of the Bayes Hilbert space $\Bl = B^2\left( \Ccal, \mathfrak{B} \cap \Ccal, \lambda \right)$ in $\B$ is $\Bd = B^2\left( \Dcal^\bullet, \Pcal \left( \Dcal^\bullet \right), \delta^\bullet \right)$, where $\Dcal^\bullet := \Dcal \cup \{t_{D + 1}\}$ with $t_{D + 1} \in \Rbb \setminus \Dcal$ and $\delta^\bullet := \sum_{d = 1}^{D + 1} w_d \, \delta_{t_d}, ~ w_{D + 1} := \lambda(I)$.
The embeddings to consider $\Bl$ and $\Bd$ as subspaces of $\B$ are
\begin{align*}
&\Jc : \Bl \hookrightarrow \B && \fc \mapsto
\begin{cases}
\fc & \mathrm{on}~ \Ccal \\
% \exp \Scal_\lambda(\fc) & \mathrm{on}~ \Dcal
\exp \Scal_\Ccal(\fc) & \mathrm{on}~ \Dcal
\end{cases} \\
&\Jd: \Bd \hookrightarrow \B && \fd \mapsto
\begin{cases}
\fd \left( t_{D+1} \right) & \mathrm{on}~ \Ccal\\
\fd & \mathrm{on}~ \Dcal
\end{cases} \, ,
\end{align*}
where $\exp \Scal_\Ccal(\fc) % \Scal_\lambda(\fc)
$ is the geometric mean of $\fc$, see Proposition~3.2. % \ref{thm_subcompositional_coherence}.
This means, for every $\alpha \in \Rbb, ~\fc, g_{\mathrm{c}} \in \Bl, ~\fd, g_{\mathrm{d}} \in \Bd$:
\begin{enumerate}[(a)]
\item\label{proof_decomposition_linearity}
$\Jc(\alpha \odot \fc \oplus g_{\mathrm{c}}) = \alpha \odot \Jc (\fc) \oplus \Jc (g_{\mathrm{c}})$ and $\Jd(\alpha \odot \fd \oplus g_{\mathrm{d}}) = \alpha \odot \Jd (\fd) \oplus \Jd (g_{\mathrm{d}})$ (Linearity),
\item\label{proof_decomposition_preservation_norm}
$\Vert \Jc(\fc) \Vertb = \Vert \fc\Vertbl$ and $\Vert \Jd(\fd) \Vertb = \Vert \fd\Vertbd$ (Preservation of the norm),
\item
$\langle \Jc(\fc), \Jd(\fd) \rangleb = 0$ (Orthogonality).
%\end{enumerate}
%Moreover:
%\begin{enumerate}[(d)]
\item
For every $f \in \B$, there exist unique functions $\fc \in \Bl, \fd \in \Bd$ such that $f = \Jc(\fc) \oplus \Jd (\fd)$, given by
\begin{align}
&\fc: \Ccal \ra \Rbb , \quad
t \mapsto f(t) ,%\label{decomposition_continuous}\\
&&\fd: \Dcal^\bullet \ra \Rbb, \quad
t \mapsto
\begin{cases}
1 , & t = t_{D + 1} \\
\frac{f(t)}{\exp \Scal_\lambda(f)} , & t \in \Dcal .
\end{cases} \label{proof_decomposition_decomposition} %\label{decomposition_discrete}
\end{align}
\end{enumerate}
\end{prop}

\begin{proof}
%As $\Dcal$ is a discrete set, the Lebesgue measure on $I$ yields the same results as on $\Ccal$ and thus
%$\Bl
%= B^2\left( I, \mathfrak{B}, \lambda \right)
%= B^2\left( \Ccal, \mathfrak{B} \cap \Ccal, \lambda \right) $.
%Furthermore, the Lebesgue measure $\lambda$ on $\Ccal$ yields the same results as the mixed measure $\mu$ on $\Ccal$, i.e.,
We have
$\Bl = B^2\left( \Ccal, \mathfrak{B} \cap \Ccal, \lambda \right)
= B^2\left( \Ccal, \mathfrak{B} \cap \Ccal, \mu \right)$, per definition of $\mu$.
%In particular, we have
%$\Scal_\lambda(\fc)
%%= \frac{1}{\lambda(I)}\int_{I} \log \fc \, \dlamb
%%= \frac{1}{\lambda(\Ccal)}\int_{\Ccal} \log \fc \, \dlamb
%%= \Scal_{B^2\left( I, \mathfrak{B}, \lambda \right)}(\fc)
%%= \Scal_{B^2\left( \Ccal, \mathfrak{B} \cap \Ccal, \lambda \right)}(\fc)
%= \Scal_{B^2\left( \Ccal, \mathfrak{B} \cap \Ccal, \mu \right)}(\fc)
%$.
Since $\Ccal \in \mathfrak{B}$, we obtain from Proposition~3.2 % \ref{thm_subcompositional_coherence} 
that $\Jc$ is well-defined and fulfills~\eqref{proof_decomposition_linearity} and~\eqref{proof_decomposition_preservation_norm}.
For $\Jd$, well-definedness is obvious.
%Now, consider $f_\mathrm{d} \in \Bd$ and $c > 0$, i.e., $\fd \isb c \, \fd$.
%From the definition of $\Jd$, we obviously  have $\Jd (c \, f_\mathrm{d}) = c\, \Jd (\fd) \isb \Jd (f_\mathrm{d})$.
%Thus, $\Jd$ is well-defined, as well.

%Before we prove the statements of the theorem, we show two equations, which are helpful for the following. For every $\fc \in \Bl, \fd \in \Bd$ and $n \in \{1, 2\}$, we have
%\begin{align}
%\int_{I} \left( \log \Jc(\fc) \right)^n \, \dmu
%&= \sum_{d=1}^D w_d \, \left( \log K_{\fc} \right)^n + \int_{I} \left( \log \fc \right)^n \, \dlamb \notag \\
%&= \delta(\Dcal) \, \left( \frac{1}{\lambda(I)} \, \Scal_{\fc} \right)^n + \int_{I} \left( \log \fc \right)^n \, \dlamb \notag \\
%&= \begin{cases}
%\frac{\delta(\Dcal) + \lambda (I)}{\lambda (I)} \, \Scal_{\fc}
%= \frac{\mu(I)}{\lambda (I)} \, \Scal_{\fc} & , n = 1 \\
%\frac{\delta(\Dcal)}{\lambda(I)^2} \Scal_{\fc}^2 + \int_{I} \left( \log \fc \right)^2 \, \dlamb &, n = 2
%\end{cases} \label{S1}
%\\
%\int_{I} \left( \log \Jd(\fd) \right)^n \, \dmu
%&= \sum_{d=1}^D w_d \, \left( \log \fd(t_d) \right)^n + \int_{I} \left( \log \fd \left( t_{D + 1} \right) \right)^n \, \dlamb \notag \\
%&= \sum_{d=1}^D w_d \,\left( \log \fd(t_d) \right)^n + \lambda (I) \, \left( \log \fd \left( t_{D + 1} \right) \right)^n \notag \\
%&= \sum_{d=1}^{D + 1} w_d \, \left( \log \fd(t_d) \right)^n
%= \int_{\Dcal^\bullet} \left( \log \fd \right)^n \, \ddel \label{S2}
%\end{align}

\begin{enumerate}[(a)]
\item
%For $\Jc$ this was shown already.
For $\Jd$, this is straightforward by definition.
%Let $\alpha \in \Rbb, \fc, g_{\mathrm{c}} \in \Bl$ and $\fd, g_{\mathrm{d}} \in \Bd$.
%\begin{align*}
%% Discrete:
%\Jd (\alpha \odot \fd \oplus g_{\mathrm{d}})
%%&= \Jd \left( \fd^\alpha \, g_{\mathrm{d}} \right)
%&= \left\{\begin{array}{ll}
%\fd^\alpha \, g_{\mathrm{d}} & \text{on}~\Dcal \\
%\left[\fd^\alpha \, g_{\mathrm{d}} \right] \left(t_{D+1} \right)  & \text{on}~\Ccal
%\end{array}\right\} \\
%&= \left\{\begin{array}{ll}
%\fd^\alpha \, g_{\mathrm{d}} & \text{on}~\Dcal \\
%\left[\fd \left(t_{D+1} \right) \right]^\alpha \, g_{\mathrm{d}}\left(t_{D+1} \right) & \text{on}~\Ccal
%\end{array}\right\}
%=\alpha \odot \Jd (\fd) \oplus \Jd (g_{\mathrm{d}}).
%\end{align*}
%
\item
%Consider an arbitrary Bayes Hilbert space $B^2(\nu) = B^2\left(\Tcal_\nu, \Acal_\nu , \nu\right)$. Then, for every $f \in B^2(\nu)$
%\begin{align}
%\Vert f \Vert_{B^2(\nu)}^2
%%&= \int_{\Tcal_\nu} \left( \clr_\nu [f] \right)^2 \, \dnu
%&= \int_{\Tcal_\nu} \left( \log f - \Scal_{\nu} (f) \right)^2 \, \dnu \notag \\
%&= \int_{\Tcal_\nu} \left( \log f \right)^2 \, \dnu
%- 2\nu (\Tcal_\nu) \Scal_{\nu} (f) ^2
%+ \nu (\Tcal_\nu) \Scal_{\nu} (f) ^2 \notag \\
%&= \int_{\Tcal_\nu} \left( \log f \right)^2 \, \dnu - \nu (\Tcal_\nu) \Scal_{\nu} (f) ^2. \label{calculation_norm}
%\end{align}
%Furthermore, for every $n \in \Nbb$ and $\fd \in \Bd$,
Let $\fd \in \Bd$.
With $\mu (\Tcal) = \delta (\Dcal) + \lambda (I) = \delta^\bullet (\Dcal^\bullet)$ we have
\begin{align}
\Scal_\mu\left( \Jd(\fd) \right)
&= \frac1{\mu(\Tcal)} \left( \int_{\Dcal} \log \fd(t_d) \, \ddel + \lambda (I) \log \fd \left( t_{D + 1} \right) \right)
%%
%= \frac1{\delta^\bullet(\Dcal^\bullet)} \int_{\Dcal^\bullet} \log \fd \, \ddel
%
= \Scal_{\delta^\bullet}(\fd). \label{S2}
\end{align}
This yields
\begin{align*}
\Vert \Jd(\fd) \Vertb^2
%%
%&= \int_{I} \left( \log \Jd(\fd) - \Scal_\mu\left( \Jd(\fd) \right) \right)^2 \, \dmu \\
%
&= %&\hspace{-0.19cm}\overset{\eqref{S2}}{=}
\int_{\Dcal} \left( \log \fd - \Scal_{\delta^\bullet}(\fd) \right)^2 \, \ddel
+ \lambda(I) \left( \log \fd (t_{D+1}) - \Scal_{\delta^\bullet}(\fd) \right)^2 \\
&=\int_{\Dcal^\bullet} \left( \log \fd - \Scal_{\delta^\bullet}(\fd) \right)^2 \, \ddelb
= \Vert \fd\Vertbd^2.
\end{align*}
\item
%Let $\fc \in \Bl, \fd \in \Bd$ and recall
%$\Bl
%\cong B^2\left( \Ccal, \mathfrak{B} \cap \Ccal, \lambda \right)$.
%Then, we have
For $\fc \in \Bl, \fd \in \Bd$, we have
\begin{align*}
\langle \Jc(\fc), \Jd(\fd) \rangleb
&\hspace{-0.19cm}\overset{\eqref{proof_subcoh_orthogonality}}{=}
\langle \fc, \Jd(\fd)|_{\Ccal} \rangle_{\Bl} %\\
%
%&= \int_{\Ccal} \clr_{\Bl} [\fc] \, \clr_{\Bl} [\fd (t_{D+1})] \, \dlamb
%%
= 0,
\end{align*}
as $\Jd(\fd)|_{\Ccal}$ is a constant and thus $0 \in \Bl$.
\item
For $f \in \B$ consider $\fc$ and $\fd$ as in~\eqref{proof_decomposition_decomposition}.
%\begin{align*}
%&\fc: \Ccal \ra \Rbb && t \mapsto f(t) \, , \\
%&\fd: \Dcal^\bullet \ra \Rbb && t \mapsto
%\begin{cases}
%1 & , t = t_{D + 1} \\
%\frac{f(t)}{\exp \left( \Scal_\lambda(f) \right)} & , t \in \Dcal
%\end{cases} \, .
%\end{align*}
With $f \in \B$, we have $\int_{\Dcal} \left( \log f \right)^2 \ddel + \int_{I} \left( \log f\right)^2 \dlamb = \int_{\Tcal} \left( \log f\right)^2 \, \dmu < \infty$, thus all terms on the left side have to be finite, as well.
Looking at the second term, we get $\fc \in \Bl$ since the Lebesgue integral yields the same results on $I$ and $\Ccal$.
Moreover,
$\fc \in \Bl \subset B^1(\lambda)$ implies $\Scal_\lambda(f) = \Scal_\lambda (\fc) < \infty$.
Similarly, from $\int_{\Dcal} \left( \log f \right)^2 \ddel < \infty$ it follows $\Scal_\delta(f) < \infty$.
Then, we get
\begin{align*}
\int_{\Dcal^\bullet} \left( \log \fd \right) ^2 \ddelb
&= \int_{\Dcal} \left( \log f - \Scal_\lambda(f) \right)^2 \, \ddel + \lambda (I) \\
&= \int_{\Dcal} \left( \log f \right)^2  \, \ddel
- 2 \delta (\Dcal) \Scal_\delta(f) \Scal_\lambda(f) + \delta (\Dcal) \Scal_\lambda(f)^2 + \lambda (I)
< \infty,
\end{align*}
%$f(t_d) = \infty$ geht nicht, weil $f~\mu$-fast überall endlich ist und $\mu{\{t_d}\} > 0$
i.e., $\fd \in \Bd$. Finally,
\begin{align*}
\Jc(\fc) \oplus \Jd(\fd)
=
\left\{\begin{array}{ll}
f & \text{on}~\Ccal \\
\exp \left( \Scal_\lambda(f) \right) \, \frac{f}{\exp \left( \Scal_\lambda(f) \right)}  & \text{on}~\Dcal
\end{array}
\right\}
= f.
\end{align*}
As we already showed that $\Bl$ and $\Bd$ form an orthogonal decomposition of $\B$ in (a) -- (c), the representation of $f$ by $\fc$ and $\fd$ is unique and thus $\Bd$ is the orthogonal complement of $\Bl$ in $\B$.
\qedhere
\end{enumerate}
\end{proof}

\begin{prop}\label{thm_decomposition_L2}
% Let $\Tcal, \Acal, \mu, I, \Dcal, \Ccal, \delta, \Dcal^\bullet$, and $\delta^\bullet$ be defined as in Proposition~\ref{thm_decomposition_bayes}.
%The Hilbert space $\Ln = L^2_0 \left( I, \mathfrak{B}, \mu\right)$
%%$\mu = \delta + \lambda$, where $\delta = \sum_{d = 1}^D w_d \, \delta_{t_d}$ with $t_d \in I$ and $w_d > 0$ for $d = 1, \ldots , D$,
%can be decomposed into two orthogonal subspaces $\Lnl = L_0^2\left( I, \mathfrak{B}, \lambda \right)$ and $\Lnd = L_0^2\left( \Dcal^\bullet, \Pcal \left( \Dcal^\bullet \right), \delta^\bullet \right)$
%%, where $\Dcal^\bullet := \{t_1, \ldots , t_{D + 1}\}$ with $t_d \neq t_{D + 1} \in \Rbb$ for all $d = 1, \ldots , D$ and $\delta^\bullet := \sum_{d = 1}^{D + 1} w_d \, \delta_{t_d}, ~ w_{D + 1} := \lambda(I)$,
%with respect to the embeddings
Defining all measures and sets as in Proposition~\ref{thm_decomposition_bayes}, the orthogonal complement of $\Lnl = L_0^2\left( \Ccal, \mathfrak{B}, \lambda \right)$ in $\Ln = L^2_0 \left( I, \mathfrak{B}, \mu\right)$ is $\Lnd = L_0^2\left( \Dcal^\bullet, \Pcal \left( \Dcal^\bullet \right), \delta^\bullet \right)$
%, where $\Dcal^\bullet := \{t_1, \ldots , t_{D + 1}\}$ with $t_d \neq t_{D + 1} \in \Rbb$ for all $d = 1, \ldots , D$ and $\delta^\bullet := \sum_{d = 1}^{D + 1} w_d \, \delta_{t_d}, ~ w_{D + 1} := \lambda(I)$,
with respect to the embeddings
\begin{align*}
&\Jtc : \Lnl \hookrightarrow \Ln && \ftc \mapsto
\begin{cases}
\ftc & \mathrm{on}~ \Ccal \\
0 & \mathrm{on}~ \Dcal
\end{cases} \\
&\Jtd: \Lnd \hookrightarrow \Ln && \ftd \mapsto
\begin{cases}
\ftd \left( t_{D+1} \right) & \mathrm{on}~ \Ccal \\
\ftd & \mathrm{on}~ \Dcal
\end{cases} \, .
\end{align*}
%with $\Dcal := \{ t_1, \ldots , t_D\}, ~ \Ccal := I \setminus \Dcal, ~ K_{\fc}(k) := \exp \left( k \, \Scal_{\fc} \right)$ for $k \in \Rbb$, where $\Scal_{\fc} := \int_{I} \log \fc \, \dlamb, ~ k_1 := -\frac{1}{\delta(\Dcal) + \lambda (I)}$ and $k_2 := \frac{\delta(\Dcal)}{\lambda (I) \, \left( \delta(\Dcal) + \lambda (I) \right)}$.
% with all measures and sets defined as in Proposition~\ref{thm_decomposition_bayes}.
The decomposition is equivalent to the one in Proposition~\ref{thm_decomposition_bayes}, i.e., for all $\fc \in \Bl$ and all $\fd \in \Bd$ we have $\Jtc \left(\clr_{\lambda} \left[ \fc\right] \right) = \clr_\mu \left[ \Jc \left( \fc \right) \right]$ and $\Jtd \left(\clr_{\delta^\bullet} \left[ \fd\right] \right) = \clr_\mu \left[ \Jd \left( \fd \right) \right]$.
%This means, for every $\alpha \in \Rbb, \ftc, \gtc \in \Lnl, \ftd, \gtd \in \Lnd$:
%\begin{enumerate}[(a)]
%\item
%$\Jtc(\alpha \, \ftc + \gtc) = \alpha \, \Jtc (\ftc) + \Jtc (\gtc)$ and $\Jtd(\alpha \, \ftd + \gtd) = \alpha \, \Jtd (\ftd) + \Jtd (\gtd)$ (Linearity),
%\item
%$\Vert \Jtc(\ftc) \Vert_{\Ln} = \Vert \ftc \Vert_{\Lnl}$ and $\Vert \Jtd(\ftd) \Vert_{\Ln} = \Vert \ftd \Vert_{\Lnd}$ (Preservation of the norm),
%\item
%$\langle \Jtc(\ftc), \Jtd(\ftd) \rangle_{\Ln} = 0$ (Orthogonality).
%\end{enumerate}
Moreover, the representation of $\ft \in \Ln$ as $\ft = \Jtc(\ftc) + \Jtd (\ftd)$ with unique functions $\ftc \in \Lnl, \ftd \in \Lnd$ given by
\begin{align}
&\ftc: \Ccal \ra \Rbb && t \mapsto \ft (t) - \frac{1}{\lambda(\Ccal)} \int_\Ccal \ft \, \dlamb \, , \notag \\
&\ftd: \Dcal^\bullet \ra \Rbb && t \mapsto
\begin{cases}
\frac{1}{\lambda(\Ccal)} \int_\Ccal \ft \, \dlamb & , t = t_{D + 1} \\
\ft (t) & , t \in \Dcal
\end{cases} \, , \label{proof_decomposition_L2_decomposition}
\end{align}
%\begin{enumerate}[(d)]
%\item
%for every $\ft \in \Ln$, there exist unique functions $\ftc \in \Lnl, \ftd \in \Lnd$ such that $\ft = \Jtc(\ftc) + \Jtd (\ftd)$. This
is equivalent to the unique representation of $f \in \B$ as $f = \Jc(\fc) \oplus \Jd (\fd)$, see~\eqref{proof_decomposition_decomposition}, % with $\fc \in \Bl$ and $\fd \in \Bd$
via clr transformations.
%$\fc \in \Bl$ and $\fd \in \Bd$ as in~(\ref{decomposition_continuous}) and~(\ref{decomposition_discrete}) and
This means, for $\ft = \clr_\mu [f] \in \Ln$
%with $\ftc \in \Lnl$ and $\ftd \in \Lnd$ as in~(\ref{decomposition_continuous_clr}) and~(\ref{decomposition_discrete_clr})
we have $\ftc = \clr_\lambda [\fc] \in \Lnl$ and $\ftd = \clr_{\delta^\bullet} [\fd] \in \Lnd$.
%\end{enumerate}
\end{prop}

\begin{proof}
Linearity, preservation of the norm, and orthogonality are straightforward calculations.
Thus, $\Lnl$ and $\Lnd$ form an orthogonal decomposition of $\Ln$.
To show the equivalence to the decomposition in Proposition~\ref{thm_decomposition_bayes}, consider $\fc \in \Bl$ and $\fd \in \Bd$.
%Recall that $\Scal_\lambda (\fc) = \Scal_{B^2(\Ccal, \mathfrak{B} \cap \Ccal, \mu)} (\fc)$, see the proof of Proposition~\ref{thm_decomposition_bayes}.
%Using definition~\eqref{mean_log_integral},
Then, we have
\begin{align*}
% Continuous:
\clr_\mu \left[ \Jc \left( \fc \right) \right]
&= \log \Jc \left( \fc \right) - \Scal_{B^2(\Tcal, \Acal, \mu)} \left( \Jc( \fc ) \right)
\overset{\eqref{proof_subcoh_S_embedding}}{=} \log \Jc \left( \fc \right) - \Scal_{B^2(\Ccal, \mathfrak{B} \cap \Ccal, \mu)} (\fc) \\
&= \left\{\begin{array}{ll}
\log \fc - \Sfc & \mathrm{on}~ \Ccal \\
\Sfc - \Sfc  & \mathrm{on}~ \Dcal
\end{array}\right\} %\\
%%
%&= \left\{\begin{array}{ll}
%\log \fc - \frac{1}{\lambda (I)} \int_{I} \log \fc \, \dlamb & \mathrm{on}~ \Ccal \\
%0 & \mathrm{on}~ \Dcal
%\end{array}\right\}
%%
%= \left\{\begin{array}{ll}
%\clr_{\lambda} \left[ \fc\right] & \mathrm{on}~ \Ccal \\
%0 & \mathrm{on}~ \Dcal
%\end{array}\right\}
%
= \Jtc \left(\clr_{\lambda} \left[ \fc\right] \right) ,
\\
%
% Discrete:
\clr_\mu \left[ \Jd \left( \fd \right) \right]
&= \log \Jd \left( \fd \right) - \Scal_\mu \left( \Jd ( \fd ) \right)
\overset{(\ref{S2})}{=} \log \Jd \left( \fd \right) - \Scal_{\delta^\bullet} ( \fd ) \\
&= \left\{\begin{array}{ll}
\log \fd (t_{D+1}) - \Scal_{\delta^\bullet} ( \fd ) & \mathrm{on}~ \Ccal \\
\log \fd - \Scal_{\delta^\bullet} ( \fd ) & \mathrm{on}~ \Dcal
\end{array}\right\}
= \Jtd \left(\clr_{\delta^\bullet} \left[ \fd \right] \right).
\end{align*}
%Now we show that for every $\ft \in \Ln$ there exist unique functions $\ftc \in \Lnl, ~\ftd \in \Lnd$ with $\ft = \Jtc(\ftc) + \Jtd(\ftd)$.
For $\ft \in \Ln$ consider $\ftc$ and $\ftd$ as in~\eqref{proof_decomposition_L2_decomposition}.
As $\ft \in \Ln$, we have $\int_{\Dcal} \ft ^2 \, \ddel + \int_I \ft^2 \, \dlamb = \int_\Tcal \ft^2 \, \dmu < \infty$.
Thus, both terms on the left side are finite and in particular, $\ft \in L^2(\lambda) \subset L^1(\lambda)$.
Then,
%With a calculation analogical to~(\ref{calculation_norm}), we get
\begin{align*}
\int_\Ccal \ftc ^2 \, \dlamb
&= \int_\Ccal \left( \ft - \frac{1}{\lambda(\Ccal)} \int_\Ccal \ft \, \dlamb \right)^2 \, \dlamb %\\
%&= \int_\Ccal \ft^2 \, \dlamb - \frac{2}{\lambda(\Ccal)} \Bigl( \int_\Ccal \ft \, \dlamb \Bigr)^2 + \frac{\lambda(\Ccal)}{\left( \lambda(\Ccal) \right)^2} \left( \int_\Ccal \ft \, \dlamb \right)^2 \, \dlamb \\
= \int_\Ccal \ft^2 \, \dlamb - \frac{1}{\lambda(\Ccal)} \Bigl( \int_\Ccal \ft \, \dlamb \Bigr)^2
< \infty. %\\
%\intertext{Furthermore,}
%\int_\Ccal \ftc \, \dlamb
%&= \int_\Ccal \Bigl( \ft - \frac{1}{\lambda(\Ccal)} \int_\Ccal \ft \, \dlamb \Bigr) \, \dlamb
%= \int_I\Ccal \ft \, \dlamb - \frac{\lambda(\Ccal)}{\lambda(\Ccal)} \int_\Ccal \ft \, \dlamb
%= 0,
\end{align*}
It is straightforward to show $\int_\Ccal \ftc \, \dlamb = 0$, i.e., $\ftc \in \Lnl$.
Moreover, we have
\begin{align*}
\int_{\Dcal^\bullet} \ftd^2 \, \ddelb
&= \int_{\Dcal} \ft^2 \, \ddel + \frac{\lambda(I)}{\lambda(\Ccal)^2} \left( \int_\Ccal \ft \, \dlamb \right)^2
< \infty. %\\
%\intertext{and with $\ft \in \Ln$}
%\int_{\Dcal^\bullet} \ftd \, \ddelb
%&= \int_{\Dcal} \ft \, \ddel + \frac{\lambda(I)}{\lambda(I)} \int_I \ft \, \dlamb
%= \int_I \ft \, \dmu = 0,
\end{align*}
The same calculation without squares shows
$\int_{\Dcal^\bullet} \ftd \, \ddelb
= \int_{\Dcal} \ft \, \ddel + \int_\Ccal \ft \, \dlamb
= \int_\Tcal \ft \, \dmu$,
which is zero as $\ft \in \Ln$ and thus $\ftd \in \Lnd$.
Furthermore,
\begin{align*}
\Jtc(\ftc) + \Jtd (\ftd)
&= \left\{\begin{array}{ll}
\ft - \frac{1}{\lambda(\Ccal)} \int_\Ccal \ft \, \dlamb + \frac{1}{\lambda(\Ccal)} \int_\Ccal \ft \, \dlamb & \mathrm{on}~ \Ccal \\
0 + \ft & \mathrm{on}~ \Dcal
\end{array}\right\}
= \ft
\end{align*}
and the uniqueness of the representation follows from $\Lnl$ and $\Lnd$ being an orthogonal decomposition of $\Ln$.
This implies that $\Lnd$ is the orthogonal complement of $\Lnl$ in $\Ln$.
Finally, we show the equivalence to the representation $f = \Jc(\fc) \oplus \Jd (\fd)$ of $f \in \B$ with unique functions $\fc \in \Bl$ and $\fd \in \Bd$.
Consider $\ft = \clr_\mu [f] \in \Ln$.
With the equivalence of the decompositions and linearity of $\clr_\mu$ we get
\begin{align*}
\Jtc(\ftc) + \Jtd (\ftd)
= \ft
=\clr_\mu [f]
= \clr_\mu \left[ \Jc(\fc) \oplus \Jd (\fd) \right]
= \Jtc \left( \clr_\lambda[\fc] \right) + \Jtd \left( \clr_{\delta^\bullet} [\fd] \right)
\end{align*}
and uniqueness of the representation yields $\ftc = \clr_\lambda[\fc]$ and $\ftd = \clr_{\delta^\bullet} [\fd]$.
\end{proof}

\section{Transforming a vector from \texorpdfstring{$\Ltwo^{K_Y + 1}$ to $\Ln^{K_Y}$}{L2 to L2 with integrate-to-zero constraint}}\label{appendix_constraint}
The approach described in this section is motivated by the inclusion of the sum-to-zero constraint in  functional linear array models, compare~(3), % (\ref{identifiability_constraint}), 
described in the online appendix~A of~\citet{brh2015}.
It is based on the implementation of linear constraints \citep[Section~1.8.1]{wood2017}.
For a vector $\bar{\bfe}_Y = (\bar{b}_{Y, 1} , \ldots , \bar{b}_{Y, K_Y + 1}) \in \Ltwo^{K_Y + 1}$, consider
\begin{align*}
\Cf := \Bigl( \int_\Tcal \bar{b}_{Y, 1} \, \dmu , \ldots , \int_\Tcal \bar{b}_{Y, K_Y + 1} \, \dmu \Bigr) \in \Rbb^{1 \times K_Y + 1}.
\end{align*}
Determining the QR decomposition of $\Cf^\top$ yields
\begin{align*}
\Cf^\top = [ \Qf : \Zf ] \left[ \begin{array}{c}
R \\
\mathbf{0}_{K_Y}
\end{array} \right] ,
\end{align*}
where $[ \Qf : \Zf ]$ is a $(K_Y + 1) \times (K_Y + 1)$ orthogonal matrix, $R$ is a $1 \times 1$ (upper triangular) matrix and $\mathbf{0}_{K_Y}$ is the vector of length $K_Y$ containing zeros in every component.
The matrix $\Zf = (z_{ij})_{i = 1, \ldots , K_Y + 1, j = 1, \ldots , K_Y}$ is the desired transformation matrix.
We obtain the transformed vector $\tilde{\bfe}_Y = (\tilde{b}_{Y, 1} , \ldots , \tilde{b}_{Y, K_Y})$ by the linear combinations of each column of $\Zf$ with the vector $\bar{\bfe}_Y$:
\begin{align*}
\tilde{b}_{Y, m} := \sum_{i=1}^{K_Y + 1} \bar{b}_{Y, i} \, z_{im} && m = 1, \ldots , K_Y
\end{align*}
Then, we have
\begin{align*}
\Bigl(\int_\Tcal \tilde{b}_{Y, 1} \, \dmu , \ldots , \int_\Tcal \tilde{b}_{Y, K_Y} \, \dmu \Bigr)
&= \Bigl(\int_\Tcal \sum_{i=1}^{K_Y + 1} \bar{b}_{Y, i} \, z_{i1} \, \dmu , \ldots , \int_\Tcal \sum_{i=1}^{K_Y + 1} \bar{b}_{Y, i} \, z_{iK_Y} \, \dmu \Bigr) \\
&= \Bigl(\sum_{i=1}^{K_Y + 1} \int_\Tcal  \bar{b}_{Y, i} \, \dmu \, z_{i1}, \ldots , \sum_{i=1}^{K_Y + 1} \int_\Tcal \bar{b}_{Y, i} \, \dmu \, z_{iK_Y} \Bigr) \\
&= \Cf \Zf
= [ R : \mathbf{0}_{K_Y}^\top ] \left[ \begin{array}{c}
\Qf^\top \\
\Zf^\top
\end{array} \right]
\Zf  \\
&= [ R : \mathbf{0}_{K_Y}^\top ] \left[ \begin{array}{c}
\mathbf{0}_{K_Y}^\top \\
\Id_{K_Y}
\end{array} \right]
= \mathbf{0}_{K_Y}^\top,
\end{align*}
i.e., $\tilde{\bfe}_Y \in \Ln^{K_Y}$.
Now let $\bar{\bfe}_Y \in \Ltwo^{K_Y + 1}$ be a vector of basis functions with penalty matrix $\bar{\Pf}_Y \in \Rbb^{(K_Y + 1) \times (K_Y + 1)}$.
Then, the penalty matrix $\tilde{\Pf}_Y \in \Rbb^{K_Y \times K_Y}$ for the transformed basis $\tilde{\bfe}_Y \in \Ln^{K_Y}$ is obtained by transforming the original penalty matrix: $\tilde{\Pf}_Y = \Zf^\top \bar{\Pf}_Y \Zf$.

\section{Equivalence of Boosting in \texorpdfstring{$\B$ and $\Ln$}{B2 and L2 with integrate-to-zero constraint}}\label{appendix_boosting_equivalence}
To explain the equivalence of boosting in $\B$ and boosting in $\Ln$, we briefly summarize how the gradient boosting algorithm in $\B$ as described in Section~2.3 % \ref{chapter_estimation_bayes} 
is adapted for boosting in $\Ln$.
Obviously, all functions that are elements of $\B$ in the original model and algorithm are considered elements of $\Ln$ for this purpose.
In the following, we denote the latter functions with a tilde to distinguish them from the former ones.
Furthermore, the Bayes Hilbert space operations $\oplus, \odot$ and $\ootimes$ are replaced by their $\Ln$-counterparts $+, \cdot$ and $\otimes$.
%For data pairs~$(\tilde{y}_i, \xf_i) \in \Ln \times \Rbb^K, ~i = 1, \ldots , N$ consider the model
%\begin{align}
%\tilde{y}_i = \tilde{h}(\xf_i) \oplus \varepsilon_i = \bigoplus_{j=1}^J \tilde{h}_j (\xf_i) \oplus \tilde{\varepsilon}_i, \label{modelgleichung_Ln} %\label{modelgleichung_alternativ}
%\end{align}
%where $\tilde{\varepsilon}_i \in \Ln$ are functional error terms with $\Ebb(\tilde{\varepsilon}_i) = 0$ and  $\tilde{h}_j(\xf_i) \in \Ln$ are partial effects, $J \in \Nbb$.
%

We take a closer look at the second and third steps of the algorithm, which are crucial for the equivalence of the two algorithms.
In $\Ln$, they translate to:
\begin{enumerate}
\setcounter{enumi}{1}
\item
Calculate the negative gradient (with respect to the Fr\'{e}chet differential) of the empirical risk
\begin{align}
\tilde{U}_i :=
- \nabla \rho_{\tilde{y}_i} (\tilde{f}) \big|_{{\tilde{f} = \widetilde{\hh^{[m]}(\xf_i)}}}
= 2 \, \left(\tilde{y}_i - \widetilde{\hh^{[m]}(\xf_i)} \right) \in \Ln,\label{gradient_clr}
\end{align}
where $\widetilde{\hh^{[m]}(\xf_i)}
= \sum_{j=1}^J \left( \bfe_j(\xf_i)^\top \otimes \tilde{\bfe}_Y^\top \right) \thetaf_j^{[m]} \in \Ln$ and $\rho_{\tilde{y}_i} : \Ln \ra \Rbb, \tilde{f} \mapsto \Vert \tilde{y}_i - \tilde{f} \Vert_{\Ltwo}^2$ is the quadratic loss functional on $\Ln$.
%, compare~(\ref{gradient}).
For $j = 1, \ldots , J$, fit the base-learners
\begin{align}
\zetafh_j
&= \argmin_{\zetaf \in \Rbb^{K_j \, K_Y}} \sum_{i=1}^N \left\Vert \tilde{U}_i - \left( \bfe_j(\xf_i)^\top \otimes \tilde{\bfe}_Y ^\top \right) \zetaf \right\Vert_{\Ltwo}^2 + \zetaf^\top \Pf_{jY} \zetaf \label{fit_baselearner_clr}
\end{align}
and select the best base-learner
\begin{align}
j^\bigstar
= \argmin_{j = 1, \ldots , J} \sum_{i=1}^N \left\Vert \tilde{U}_i - \left( \bfe_j(\xf_i)^\top \otimes \tilde{\bfe}_Y ^\top \right) \zetafh_j \right\Vert_{\Ltwo}^2.  \label{select_baselearner_clr}
\end{align}
\item
The coefficient vector corresponding to the best base-learner is updated, the others stay the same: $\thetaf_{j^\bigstar}^{[m+1]} := \thetaf_{j^\bigstar}^{[m]} + \kappa \, \gammafh_{j^\bigstar}, ~\thetaf_j^{[m+1]} := \thetaf_j^{[m]} \quad \text{for}~j \neq j^\bigstar$.
\end{enumerate}
The proof of the existence of the gradient and the equality in Equation~(\ref{gradient_clr}) is analogous to the respective proof for the original algorithm, which is provided in appendix~\ref{appendix_proofs}.

Now we compare the estimation of the original model~(2) % (\ref{modelgleichung}) 
applying the algorithm described in Section~2.3 % \ref{chapter_estimation_bayes} 
with estimation of the clr transformed model
\begin{align}
\clr [y_i]
= \clr [h (\xf_i)]  + \clr [\varepsilon_i]
= \sum_{j=1}^J \clr[h_j (\xf_i)] + \clr [\varepsilon_i]. \label{modelgleichung_clr}
\end{align}
applying the adapted algorithm.
Let $\bfe_Y = (b_{Y, 1}, \ldots , b_{Y, K_Y})\in \B^{K_Y}$ be the vector of basis functions over $\Tcal$ in the original estimation problem.
On clr transformed level, we choose $\tilde{\bfe}_Y = \clr[\bfe_Y] = (\clr[b_{Y, 1}], \ldots , \clr[b_{Y, K_Y}]) \in \Ln^{K_Y}$ as the corresponding vector of basis functions over $\Tcal$.
Then, the negative gradient of the empirical risk in $\Ln$ equals the clr transformed negative gradient of the empirical risk in $\B$:
Using the linearity of the clr transformation, we get
\begin{align*}
\clr[\hh^{[m]}(\xf_i)]
&= \clr\left[\bigoplus_{j=1}^J \left( \bfe_j(\xf_i)^\top \ootimes \bfe_Y^\top \right) \thetaf_j^{[m]}\right] \\
&= \sum_{j=1}^J \left( \bfe_j(\xf_i)^\top \otimes \clr[\bfe_Y]^\top \right) \thetaf_j^{[m]}
= \widetilde{\hh^{[m]}(\xf_i)},
%\intertext{and thus}
%\clr[U_i]
%&= \clr \left[ 2 \odot \left( y_i \ominus \hh^{[m]}(\xf_i)\right) \right]
%= 2 \, \left(\clr[y_i] - \clr[\hh^{[m]}(\xf_i)] \right)
%= \tilde{U}_i.
\end{align*}
and thus
$
\clr[U_i]
= \clr \left[ 2 \odot \left( y_i \ominus \smash{\hh^{[m]}}(\xf_i)\right) \right]
= 2 \, \left(\clr[y_i] - \clr[\hh^{[m]}(\xf_i)] \right)
= \tilde{U}_i$.
Furthermore, for all $i = 1, \ldots , N, \, j = 1, \ldots , J$ and $\gammaf \in \Rbb^{K_j \, K_Y}$, we have
\begin{align*}
\left\Vert U_i \ominus\left( \bfe_j(\xf_i)^\top \ootimes \bfe_Y ^\top \right) \gammaf \right\Vertb^2
&= \left\Vert \clr \left[U_i \ominus\left( \bfe_j(\xf_i)^\top \ootimes \bfe_Y ^\top \right) \gammaf \right] \right\Vert_{\Ltwo}^2 \\
&= \left\Vert \tilde{U}_i - \left( \bfe_j(\xf_i)^\top \otimes \tilde{\bfe}_Y ^\top \right) \gammaf \right\Vert_{\Ltwo}^2.
\end{align*}
Here, we used the isometry of the clr transformation in the first equation and its linearity in the second one.
This implies that the pairs of equations~(7) % (\ref{fit_baselearner}) 
and~(\ref{fit_baselearner_clr}) and~(8) % (\ref{select_baselearner}) 
and~(\ref{select_baselearner_clr}) yield the same result, i.e., $\gammafh_j = \zetafh_j$ for all $j = 1, \ldots, J$ and $j^\ast = j^\bigstar$, in each iteration of the two algorithms.
This means that the update in the third step of both algorithms is identical as well.
Thus, the resulting estimator of model~(\ref{modelgleichung_clr}) is the clr transformed estimator of~(2): % (\ref{modelgleichung}):
\begin{align*}
\widehat{\clr[y_i]}
%&= \hat{\Ebb}(\clr[y_i] ~|~ \Xf = \xf_i)
= \sum_{j=1}^J \widetilde{\hh_j^{[m_{\text{stop}}]} (\xf_i)}
= \sum_{j=1}^J \clr \left[\hh_j^{[m_{\text{stop}}]} (\xf_i)\right]
= \clr \left[\bigoplus_{j=1}^J \hh_j^{[m_{\text{stop}}]} (\xf_i)\right]
= \clr[\yh_i].
\end{align*}

This proves that the algorithms provide equivalent results: We obtain the same estimates by applying the adapted algorithm to the clr transformed model~(\ref{modelgleichung_clr}) in $\Ln$ and retransforming the estimates with $\clr^{-1}$ as by estimating model~(2) % (\ref{modelgleichung}) 
directly in $\B$.
An advantage of transforming the model is that we can then use and extend implementations for function-on-scalar regression in practice, in particular the \texttt{R} add-on package \texttt{FDboost} \citep{FDboost}, which is based on the package \texttt{mboost} \citep{mboost}.
%In particular, the package \texttt{FDboost} has the advantage that the integrate-to-zero constraint of $L_0^2(I)$ can be integrated straightforwardly.
%For more information on \texttt{FDboost} in general see the hands-on tutorial of \citet{brh2020}. Furthermore,
\ifnum\value{jasa}=1
{Our vignette, which we provide in the supplementary material, illustrates how to use \texttt{FDboost} for the density-on-scalar case.}
\else
{Our enhanced version of the package can be found in the github repository \href{https://github.com/boost-R/FDboost}{\emph{FDboost}}.
The vignette ``density-on-scalar\_birth'' illustrates how to use it for the density-on-scalar case.}
\fi

\section{Further notes and ideas regarding interpretation}\label{appendix_interpretation}
In this section, we first briefly explain the connection of our interpretation presented in Section 3.2 to logistic models (Section~\ref{appendix_sec_logit_model}), before discussing further possibilities of interpreting effects in Sections \ref{appendix_sec_odds_geom_mean} to \ref{appendix_sec_interpretation_constant}.
% In addition to the approach of interpreting estimated effects via (log) odds ratios as presented in Section~3.2 % \ref{chapter_interpretation} 
% and used in the application in Section~4, % \ref{chapter_application}, 
% we developed the following ideas for interpretation.
More precisely, Sections~\ref{appendix_sec_odds_geom_mean} and~\ref{appendix_sec_odds_mixed} extend the ideas of odds (ratios). 
Section~\ref{appendix_sec_interpretation_constant} presents a completely different approach, decomposing the domain $\Tcal$ into two areas where the probability mass of another density increases/decreases under perturbation with this effect.

\subsection{Log odds ratios as family of logistic models}\label{appendix_sec_logit_model}
Due to the connection of our interpretation presented in Section 3.2 to odds ratios (compare Section 3.1), an estimated model can in fact be interpreted along the lines of a scalar-on-scalar logit model for comparing two parts of the female share distribution. Assume for simplicity and illustration, we have obtained a model predictor of the form $\hat{h}(x)(s) = \hat{\beta}_0(s) \oplus \hat{g}(x)(s)$ for a density of $s \in [0,1]$ and some covariate $x$ with an estimate $\hat{\beta}_0$ of the intercept and $\hat{g}$ of a covariate effect. 
Then, for two values $s, t \in [0,1]$, 
$$
\operatorname{logit}(\check{\pi}) = 
\underbrace{
\log\frac{\hat{h}(x)(s)}{\hat{h}(x)(t)}
}_{=:\check{h}(x)} =  
        \underbrace{
            \log \frac{\hat{\beta}_0(s)}{\hat{\beta}_0(t)}}_{=: \check{\beta}_0} + 
        \underbrace{
            \log\frac{\hat{g}(x)(s)}{\hat{g}(x)(t)}
        }_{=: \check{g}(x)} 
$$
yields the predictor of a scalar additive logit model for the (infinitesimal) probability $\check{\pi}$ for $s$ out of $s$ and $t$ (even though estimation is different of course). 
Here, we can also express $\check{\beta}_0 = \operatorname{clr} \hat{\beta}_0(s) - \operatorname{clr} \hat{\beta}_0(t)$ in terms of clr-transforms, and analogously for $\check{g}(x)$.
Hence, the estimated Bayes Hilbert space models can be interpreted as a family of scalar logit models, simultaneously fitted across all values of $s, t$ (in a mixed Bayes Hilbert space including values corresponding to the discrete component with point masses) and thus allowing  borrowing of strength across the domain and simultaneous interpretations for all such pairs.
While these are interesting theoretical considerations, evaluating a density at concrete single values $s, t \in \Tcal$, is however not reasonable from a probabilistic perspective, unless the values correspond to point masses (discrete component).

\subsection{Odds compared to geometric mean}\label{appendix_sec_odds_geom_mean}
The odds ratio defined in Section 3.1, % \ref{chapter_oddsratio_differences}
% While the odds of an effect $\hh_j = \hh_j^{[m_{\text{stop}}]} (\xf_i) \in \B,~j = 1, \ldots , J$ for $t$ compared to $s$, with $t, s \in \Tcal$, are given by  the exponential of the difference of the clr transformed effect evaluated at $t$ and $s$, see~\eqref{density_odds_ratio}, the exponential of the clr transformed effect at $t$ can also be interpreted directly %, i.e., without differences,
can be written as the exponential of the difference of the clr transformed densities $f_1, f_2 \in \B$ evaluated at $s$ and $t$:
\begin{align}
\OR(s,t)
= \frac{f_1(s) \, / \, f_1(t)}{f_2(s) \, / \, f_2(t)}
=\exp\left( \clr [ f_1 ] (s) - \clr [ f_1 ](t) - \left( \clr [ f_2 ](s) - \clr [ f_2 ](t)\right)\right)
%= \frac{\frac{\clr^{-1} \left[ \clr [ \hh_j ] \right] (s)}{\clr^{-1} \left[ \clr [ \hh_j ] \right] (t)}}{\frac{\clr^{-1} \left[ \clr [ \hh_k ] \right] (t)}{\clr^{-1} \left[ \clr [ \hh_k ] \right] (s)}}
%= \frac{\frac{\hh_j(t)}{\hh_j(s)}}{\frac{\hh_k(t)}{\hh_k(s)}}. \label{density_odds_ratio}
. \label{density_odds_ratio}
\end{align}
Similarly, the exponential of a clr transformed density $f \in \B$ at $s$ % the clr transformed effect at $t$ 
can also be interpreted directly %, i.e., without differences,
via the relation % even though it might be more abstract
\begin{align*}
\exp ( \clr [ f ] (s) )
= \frac{f (s)}{\exp \Scal_\mu (f )}, %\label{odds_ratio_geom}
% \exp ( \clr [ \hh_j ] (t) )
% %= \frac{\exp \left( \log \hh_j (t) \right)}{\exp \Scal (\hh_j )}
% %= \hh_j (t) / [\exp \Scal (\hh_j )]
% = \frac{\hh_j (t)}{\exp \Scal (\hh_j )}, %\label{odds_ratio_geom}
\end{align*}
where $\exp \Scal_\mu (\hh_j )$ is the geometric mean of $\hh_j$, see Footnote 2 (Proposition 3.2). % footnote~\ref{footnote_geometric_mean} in Proposition~\ref{thm_subcompositional_coherence}
Accordingly, the difference of two clr transformed densities $f_1, f_2 \in \B$ evaluated at $s$ corresponds to the log odds ratio of $f_1$ and $f_2$ compared to the geometric mean.
Again, this allows for a ceteris paribus interpretation.

\subsection{Odds for mixed case}\label{appendix_sec_odds_mixed}
For a mixed Bayes Hilbert space $\B$ as defined in Section 2.1, % \ref{chapter_bayes_hilbert_space}
% In the mixed case, i.e., $\B = B^2\left( I, \mathfrak{B}, \mu\right)$ with $\mu = \delta + \lambda$, where $\delta = \sum_{d = 1}^D w_d \, \delta_{t_d}$ for $\{t_1, \ldots , t_D\} = \Dcal \subset I$ and $w_d > 0$, 
we get a special interpretation for the odds (as defined in Section 3.1 % \ref{chapter_oddsratio_differences}
or~\eqref{density_odds_ratio}) of the discrete component $f_{\mathrm{d}} \in \Bd$ obtained from a density $f \in \B$ via (9): % \eqref{decomposition} % an effect:
% Let $\hh_{j, \, \mathrm{d}} \in \Bd$ denote the function obtained from \eqref{decomposition} %\eqref{decomposition_discrete}
% given the effect $\hh_j \in \B$.
%Then, the odds of $\hh_{j, \, \mathrm{d}}$ for a discrete value $t \in \Dcal$ compared to $t_{D+1}$ (representing the continuous component) correspond to the ratio of $\hh_j(t)$ %, i.e.,  the relative frequency of $t$,
%and the geometric mean of the continuous component.
%
% Then, for
For the odds of a discrete value $t\in \Dcal$ compared to the value $t_{D+1}$ representing the continuous component, we get
\begin{align*}
\frac{f_{\mathrm{d}}(t)}{f_{\mathrm{d}}(t_{D+1})}
%&\overset{(\ref{decomposition})}{=} \frac{f (t)}{\exp\left(\frac{1}{\lambda(I)} \int_{I} \log f \, \dlamb\right)}
&\overset{(9)}{=} \frac{f (t)}{\Scal_\lambda(f)} . % \overset{\ref{decomposition}}{=} %\label{odds_discrete}
% %\exp\left( \clr_{\delta^\bullet} [ \hh_{j, \, \mathrm{d}} ](t) - \clr_{\delta^\bullet} [ \hh_{j, \, \mathrm{d}} ](t)\right)
% %&=
% \frac{\hh_{j, \, \mathrm{d}}(t)}{\hh_{j, \, \mathrm{d}}(t_{D+1})}
% %&\overset{(\ref{decomposition})}{=} \frac{\hh_{j} (t)}{\exp\left(\frac{1}{\lambda(I)} \int_{I} \log \hh_{j} \, \dlamb\right)}
% &\overset{(\ref{decomposition})}{=} \frac{\hh_{j} (t)}{\Scal_\lambda(\hh_{j})} . %\label{odds_discrete}
\end{align*}
Thus, for the discrete component $f_{\mathrm{d}}$ the odds of $t \in \Dcal$ compared to $t_{D+1}$ correspond to the odds of the relative frequency of $t \in \Dcal$ compared to the geometric mean of the continuous component.
It is given by the exponential of the differences of the $\clr_{\delta^\bullet}$ transformed density~$\fd$ % effect $\hh_{j, \, \mathrm{d}}$ 
evaluated at $t$ and $t_{D+1}$.
%\begin{align*}
%\exp\left( \clr_{\delta^\bullet} [ \hh_{j, \, \mathrm{d}} ](s) - \clr_{\delta^\bullet} [ \hh_{j, \, \mathrm{d}} ](t_{D+1})\right) = \exp\left((\clr_{\mu} [\hh_j])_{\mathrm{d}} (s) - (\clr_{\mu} [\hh_j])_{\mathrm{d}} (t_{D+1}) \right),
%\end{align*}
%where $(\clr_{\mu} [\hh_j])_{\mathrm{d}} \in \Lnd$ is the discrete component of $\clr_{\mu} [\hh_j]$ as in~(\ref{decomposition_discrete_clr}).

\subsection{Decomposition of $\Tcal$ depending on constant}\label{appendix_sec_interpretation_constant}
The following statement applies to all Bayes Hilbert spaces $B^2(\Tcal, \Acal, \mu) =\B$, in particular to discrete, continuous and mixed ones. 
It implies that any positive constant $\alpha$ decomposes a density $f_1 \in \B$ % an effect $g \in \B$ 
into an area $I = \{f_1 \geq \alpha\}$, % $I = \{g \geq \alpha\}$, 
where the probability mass of an arbitrary density $f_2 \in \B$ % $f \in \B$ 
increases under perturbation with $f_1$ % g$ 
and an area $I^c = \{f_1 < \alpha\}$ % $I^c = \{g < \alpha\}$ 
where the probability mass decreases.
Note that this statement requires $I$ to be the maximal subset with $f_1 \geq \alpha$. % $g \geq \alpha$. 
If we don't presume $f_1 < \alpha$ % $g < \alpha$ 
on $I^c$, this is not true in general. \\
%Furthermore, equality in \eqref{eq_threshold_smaller} is only reached for $I^c = \Tcal$ or $I^c = \emptyset$ since both sides are then $0$ or $1$, respectively. Otherwise, \eqref{eq_threshold_smaller} even holds with a strict ``$<$''.
Since we are interested in probability masses, we consider probability density functions in the following.

\begin{thm}\label{thm_interpretation_threshold}
Let $f_1, f_2 \in \B$ with $\int_{\Tcal} f_1 \, \dmu = 1 = \int_{\Tcal} f_2 \, \dmu$ % $f, g \in \B$ with $\int_{\Tcal} f \, \dmu = 1 = \int_{\Tcal} g \, \dmu$
and $f_1 \geq \alpha$ on $I \in \Acal$ and $f_1 < \alpha$ on $I^c = \Tcal \setminus I$ for $\alpha \in \Rbb^+$. % $g \geq \alpha$ on $I \subseteq \Tcal$ and $g < \alpha$ on $I^c = \Tcal \setminus I$ for $\alpha \in \Rbb^+$.
Then,
\begin{align}
\int_I f_1 \oplus f_2 \, \dmu \geq \int_I f_2 \, \dmu % \int_I f \oplus g \, \dmu \geq \int_I f \, \dmu 
\label{eq_threshold_larger}
\intertext{and}
\int_{I^c} f_1 \oplus f_2 \, \dmu \leq \int_{I^c} f_2 \, \dmu. % \int_{I^c} f \oplus g \, \dmu \leq \int_{I^c} f \, \dmu. 
\label{eq_threshold_smaller}
\end{align}
\end{thm}
%\end{lem}

\begin{proof}
We have
\begin{align*}
\int_I f_1 \oplus f_2 \, \dmu 
= \frac{\int_I f_1 \cdot f_2 \, \dmu}{\int_{\Tcal} f_1 \cdot f_2 \, \dmu}
= \frac{\int_I f_1 \cdot f_2 \, \dmu}{\int_{I} f_1 \cdot f_2 \, \dmu + \int_{I^c} f_1 \cdot f_2 \, \dmu}
\intertext{and analogously}
\int_{I^c} f_1 \oplus f_2 \, \dmu 
= \frac{\int_{I^c} f_1 \cdot f_2 \, \dmu}{\int_{I} f_1 \cdot f_2 \, \dmu + \int_{I^c} f_1 \cdot f_2 \, \dmu}
%= \frac{1 - \int_{I} f_1 \cdot f_2 \, \dmu}{\int_{I} f_1 \cdot f_2 \, \dmu + \int_{I^c} f_1 \cdot f_2 \, \dmu}
.
\end{align*}
Consider
\begin{align*}
a := \int_I f_2 \, \dmu, &&
b := \int_I f_1 \cdot f_2 \, \dmu, &&
c := \int_{I^c} f_1 \cdot f_2 \, \dmu %&&
%d := \int_{I^c} f_2 \, \dmu
.
\end{align*}
Since $f_1 \geq \alpha$ on $I$ and $f_1 < \alpha$ on $I^c$, we have 
\begin{itemize}
\item[$(I)$]
$b \geq \alpha \cdot a$ %$ ~\Lera~ a \leq \frac{b}{\alpha}$
\item[$(II)$]
$c < \alpha \cdot (1-a) = \alpha - \alpha \cdot a$ %$~\Lera~ \frac{c}{\alpha} < 1-a$
%\item[$(III)$]
%$c < \alpha \cdot d$
%\item[$(IV)$]
%$b \geq \alpha \cdot (1-d)$
\end{itemize} 
Note that $a \in [0, 1]$ and $b, c \geq 0$ with $b + c > 0$.
If $a = 1$, we have $I = \Tcal$ and $I^c = \emptyset$.
Then, equality is reached in both \eqref{eq_threshold_larger} and \eqref{eq_threshold_smaller}, since both sides are $1$ and $0$, respectively.
If $a = 0$, we have $I = \emptyset$ and $I^c = \Tcal$ and again \eqref{eq_threshold_larger} and \eqref{eq_threshold_smaller} hold with equality reached.
Now, consider $a \in (0, 1)$.
Assume \eqref{eq_threshold_larger} is not true, i.e., $\frac{b}{b+c} < a$.
%\frac{\int_I f_1 \cdot f_2 \, \dmu}{\int_{\Tcal} f_1 \cdot f_2 \, \dmu} < \int_I f_2 \, \dmu.
%which is equivalent to $\frac{a - c \cdot a}{c} < b$.
Then, we have
\begin{align*}
\frac{b}{b+c} < a &&
\Lera && b < a \cdot (b + c) &&
\overset{(II)}{\Lra} && b < a \cdot (b + \alpha - \alpha \cdot a) \\
&& \Lera && b \cdot (1 - a) < \alpha \cdot a \cdot (1 - a) &&
\overset{a < 1}{\Llra} && b < \alpha \cdot a,
\end{align*}
which is a contradiction to $(I)$. Thus,
$\frac{b}{b+c} \geq a$, which shows \eqref{eq_threshold_larger}. %\\
%If $a = 0$, we have $I^c = \Tcal$ and \eqref{eq_threshold_smaller} is true, since both sides are $1$ in this case.
%
This also implies $\frac{c}{b+c} = 1 - \frac{b}{b+c} \leq 1- a$, which shows~\eqref{eq_threshold_smaller}.\footnote{
Note that using a similar approach as above, starting with the assumption $\frac{c}{b+c} \geq 1 - a$ and using (I) to obtain a contradiction to $(II)$, one can even show the strict inequality $\frac{c}{b+c} < 1- a$ for $a \in (0, 1)$.}
%Analogously, assume $\frac{c}{b+c} \geq 1 - a$. 
%(For $a \in (0, 1)$, we even show a strict inequality in \eqref{eq_threshold_smaller}.)
%%\frac{\int_I f_1 \cdot f_2 \, \dmu}{\int_{\Tcal} f_1 \cdot f_2 \, \dmu} < \int_I f_2 \, \dmu.
%%which is equivalent to $\frac{a - c \cdot a}{c} < b$.
%Then, we have
%\begin{align*}
%\frac{c}{b+c} \geq 1 - a && 
%\Lera && c \geq (1-a) \cdot (b+c) && 
%\overset{(I)}{\Lra} &&  c \geq (1 - a) \cdot (\alpha \cdot a + c) \\
%&& \Lera && a \cdot c \geq \alpha \cdot a \cdot (1 - a) &&
%\overset{a > 0}{\Llra} && c \geq \alpha \cdot (1 - a),
%\end{align*}
%which is a contradiction to $(II)$. Thus,
%$
%\frac{c}{b+c} < 1- a
%$, which shows \eqref{eq_threshold_smaller}.
\end{proof}

\pagebreak

\section{Application: Women's income share}\label{appendix_soep}

We use data from the German Socio-Economic Panel (SOEP) from 1984 to 2016 (version 33, doi:10.5684/soep.v33, see \citealp{soep2019}), with data for \emph{East} Germany being available only from 1991 onward.

\subsection{%ISO 3166-2 codes of the German federal states
Overview of regions}\label{appendix_soep_regions}
%% english names
%\emph{northwest} (Schleswig-Holstein, Bremen, Hamburg, Lower Saxony), \emph{west} (North Rhine-Westphalia), \emph{southwest} (Hesse, Rhineland-Palatinate, Saarland), \emph{south} (Bavaria, Baden-W\"{u}rttemberg), \emph{east} (Saxony-Anhalt, Thuringia, Saxony) and \emph{northeast} (Berlin, Brandenburg, Mecklenburg-West Pomerania).
%% german names
%\emph{northwest} (Schleswig-Holstein, Bremen, Hamburg, Niedersachsen), \emph{west} (Nordrhein-Westfalen), \emph{southwest} (Hessen, Rheinland-Pfalz, Saarland), \emph{south} (Bayern, Baden-W\"{u}rttemberg), \emph{east} (Sachsen-Anhalt, Th\"uringen, Sachsen) and \emph{northeast} (Berlin, Brandenburg, Mecklenburg-Vorpommern).

\ifnum\value{aoas}=1
{
\begin{table}[H]
\caption{German federal states with their ISO 3166-2 codes and the variables \emph{region} and \emph{West\_East} assigned in our application.}
\begin{center}
\begin{tabular}{l|l|l|l}
\hline
\textbf{Federal state} & \textbf{ISO 3166-2 code} & \textbf{\emph{region}} & \textbf{\emph{West\_East}} \\
\hline
Schleswig-Holstein & SH & \multirow{4}{*}{\emph{northwest}} & \multirow{10}{*}{\emph{West} (Germany)} \\
Bremen & HB & & \\
Hamburg & HH & & \\
Lower Saxony & NI & & \\
\cline{1-3}
North Rhine-Westphalia & NW & \emph{west} & \\
\cline{1-3}
Hesse & HE & \multirow{3}{*}{\emph{southwest}} & \\
Rhineland-Palatinate & RP & & \\
Saarland & SL & & \\
\cline{1-3}
Bavaria & BY & \multirow{2}{*}{\emph{south}} & \\
Baden-W\"{u}rttemberg & BW & & \\
\hline
Saxony-Anhalt & ST & \multirow{3}{*}{\emph{east}} & \multirow{6}{*}{\emph{East} (Germany)} \\
Thuringia & TH & & \\
Saxony & SN & & \\
\cline{1-3}
Berlin & BE & \multirow{3}{*}{\emph{northeast}} & \\
Brandenburg & BB & & \\
Mecklenburg-West Pomerania & MV & & \\ \hline
\end{tabular}
%\vspace{-0.5cm}
\end{center}
\end{table}
} \else
{
\begin{table}[H]
\begin{center}
\begin{tabular}{l|l|l|l}
\textbf{Federal state} & \textbf{ISO 3166-2 code} & \textbf{\emph{region}} & \textbf{\emph{West\_East}} \\
\hline
Schleswig-Holstein & SH & \multirow{4}{*}{\emph{northwest}} & \multirow{10}{*}{\emph{West} (Germany)} \\
Bremen & HB & & \\
Hamburg & HH & & \\
Lower Saxony & NI & & \\
\cline{1-3}
North Rhine-Westphalia & NW & \emph{west} & \\
\cline{1-3}
Hesse & HE & \multirow{3}{*}{\emph{southwest}} & \\
Rhineland-Palatinate & RP & & \\
Saarland & SL & & \\
\cline{1-3}
Bavaria & BY & \multirow{2}{*}{\emph{south}} & \\
Baden-W\"{u}rttemberg & BW & & \\
\hline
Saxony-Anhalt & ST & \multirow{3}{*}{\emph{east}} & \multirow{6}{*}{\emph{East} (Germany)} \\
Thuringia & TH & & \\
Saxony & SN & & \\
\cline{1-3}
Berlin & BE & \multirow{3}{*}{\emph{northeast}} & \\
Brandenburg & BB & & \\
Mecklenburg-West Pomerania & MV & & \\ \hline
\end{tabular}
\caption{German federal states with their ISO 3166-2 codes and the variables \emph{region} and \emph{West\_East} assigned in our application.}
\vspace{-0.5cm}
\end{center}
\end{table}
} \fi

\subsection{Barplots of share frequencies}\label{appendix_soep_barplots}
\begin{figure}[H]
\includegraphics[width=0.32\textwidth]{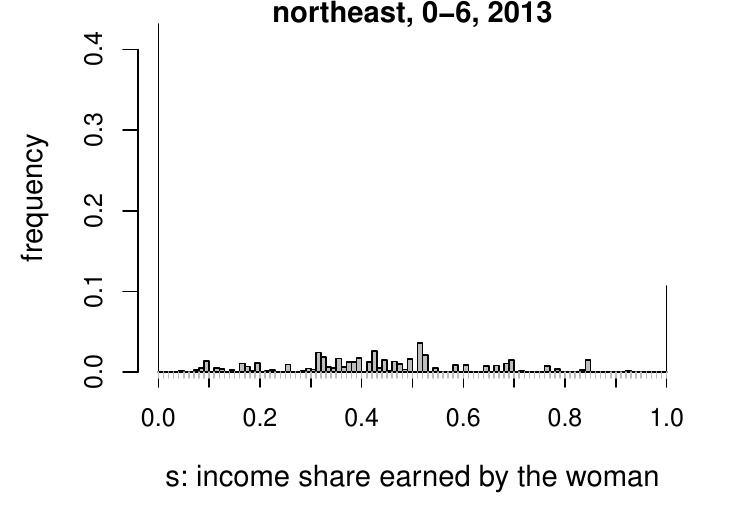}
\includegraphics[width=0.32\textwidth]{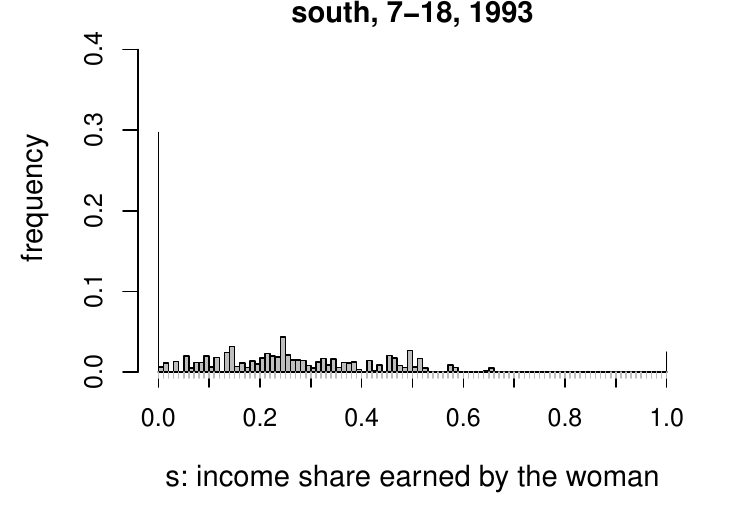}
\includegraphics[width=0.32\textwidth]{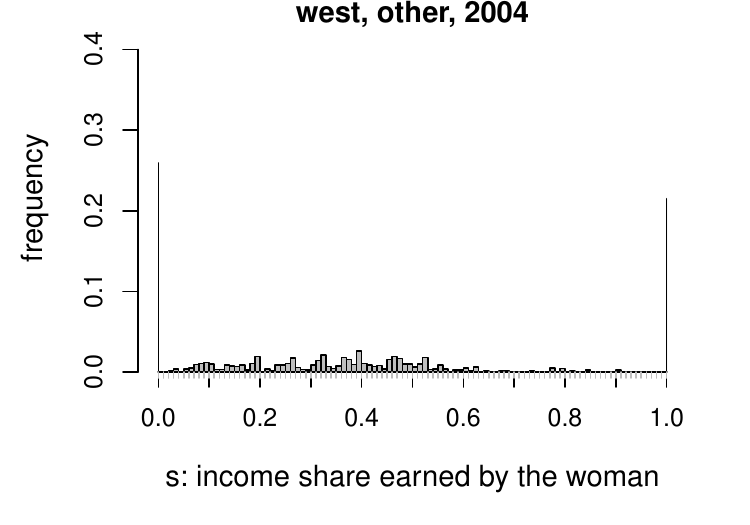}
\caption{Three barplots of share frequencies for different combinations of \emph{region}, \emph{c\_age}, and \emph{year}. The outmost bars have width zero, the ones in between width~$0.01$. \label{histograms}}
\end{figure}

\subsection{Estimation of the response densities}\label{appendix_kernel_density_estmation}

%\subsubsection{Kernel density estimation on bounded support}
In practice, density functions often have to be estimated from individual observations.
We focus on densities with bounded support $\Tcal$, which is predetermined by the application framework.
Without loss of generality, we assume  $\Tcal = [0, 1]$ as support of the unknown density~$f$, which has to be estimated.

A common approach to estimate densities is kernel density estimation.
The usual kernel density estimator for weighted observations is %\citep{elgammal2002, gisbert2003, wang2007}
\begin{align}
\hat{f}_b(t) := \sum_{l=1}^N w_l \, K_b (t - t_l ) , \label{kernel_standard}
\end{align}
where $t_1 \leq \ldots \leq t_N$ is a random sample of a random variable $T$ with (unknown) density~$f$, $w_1, \ldots, w_N$ with $\sum_{l = 1}^N w_l = 1$ are corresponding nonnegative weighting coefficients (sampling weights in our application to ensure representativeness of the survey) and $K_{b}$ is a kernel function depending on a bandwidth $b \in \Rbb$.
%In the case of equally weighted observations, all weighting coefficients are set to $\nicefrac1{N}$.
Usually, kernel functions fulfill~$K_b(t) = K \left(\frac{t}{b}\right)$, where $K$ is chosen as a density function that is symmetric around zero.
However, this is not suitable, when the bounded support $\Tcal$ of the estimator is predetermined:
If the support of $K$ is unbounded, which is the case for, e.g., the Gaussian kernel, the support of the estimator is unbounded as well.
%In the Bayes Hilbert space setting this is also problematic, because the reference measure has to be finite.
If the support of $K$ is bounded, i.e., $[-a, a]$ for an $a > 0$, the support of the estimator is $\left[\frac{t_1 - a}{b}, \frac{t_N + a}{b}\right]$ (assuming $t_l - t_{l-1} < 2a$ for all $l = 1, \ldots , N$).
Thus, it is not fixed, but depends on the sample $t_1, \ldots , t_N$ and doesn't necessarily yield the predefined $\Tcal = [0, 1]$.
%For $t_1= 0$ and $t_N = 1$ this obviously doesn't yield the predefined $\Tcal = [0, 1]$.
%
%Then, the symmetry of $K_b(t-t_l)$ around~$t_l$ for every $l = 1, \ldots , N$ leads to a fattened support compared to $[t_1, t_N]$. However, the exact range of the support is not fixed. It depends on the observations $t_1, \ldots , t_N$ and the bandwidth $b$ and can get arbitrary large.
%This is problematic as regression involves the application of the clr transformation which is dependent on the support $\Tcal$.

%Note that the usual kernel density estimator with symmetric kernels is not appropriate in our case as the support $(0,1)$ is predetermined.
To accommodate this, there are several possibilities.
\citet{petersen2016} propose a new kernel density estimator based on symmetric kernels. Outside of the predetermined interval, the value is set to $0$.
Normalization ensures that the estimator integrates to $1$ and a so-called weight function, which depends on $t$, the bandwidth, and the kernel and is unequal to 1 only in $[0, b)$ and $(1-b, 1]$, is multiplied with the kernel to remove boundary bias.
Another possibility is to use the usual kernel density estimator, but with asymmetric kernels, which are defined on the predetermined interval.
Two appropriate choices are beta-kernels introduced by \citet{chen1999} and Gaussian copula kernels presented by \citet{jones2007}.
The former are also recommended by \citet{petersen2016} as alternative to their own estimation approach.
%See appendix~\ref{appendix_kernel_density_estmation} for a summary of weighted kernel density estimation using asymmetric kernels with the focus on beta-kernels.
%To solve this problem, we use asymmetric kernel functions with support $\Tcal$.
%Two appropriate choices are beta-kernels introduced by~\citet{chen1999} and Gaussian copula kernels described by~\citet{jones2007}.
Both kernels are illustrated in Figure~\ref{kernel-functions} for bandwidths $0.02$ and $0.1$.
Besides obviously different scaling of the bandwidth parameter, the two kernels show very different behavior near the boundaries of the interval $[0,1]$.
\vspace{-1cm}
\begin{figure}[H]
%\includegraphics[width=0.5\textwidth]{Images/beta-kernel_002t_l.pdf}
%\includegraphics[width=0.5\textwidth]{Images/beta-kernel_01t_l.pdf}
%\\
%\includegraphics[width=0.5\textwidth]{Images/gcopula-kernel_002t_l.pdf}
%\includegraphics[width=0.5\textwidth]{Images/gcopula-kernel_01t_l.pdf}
\begin{center}
\includegraphics[width=0.8\textwidth]{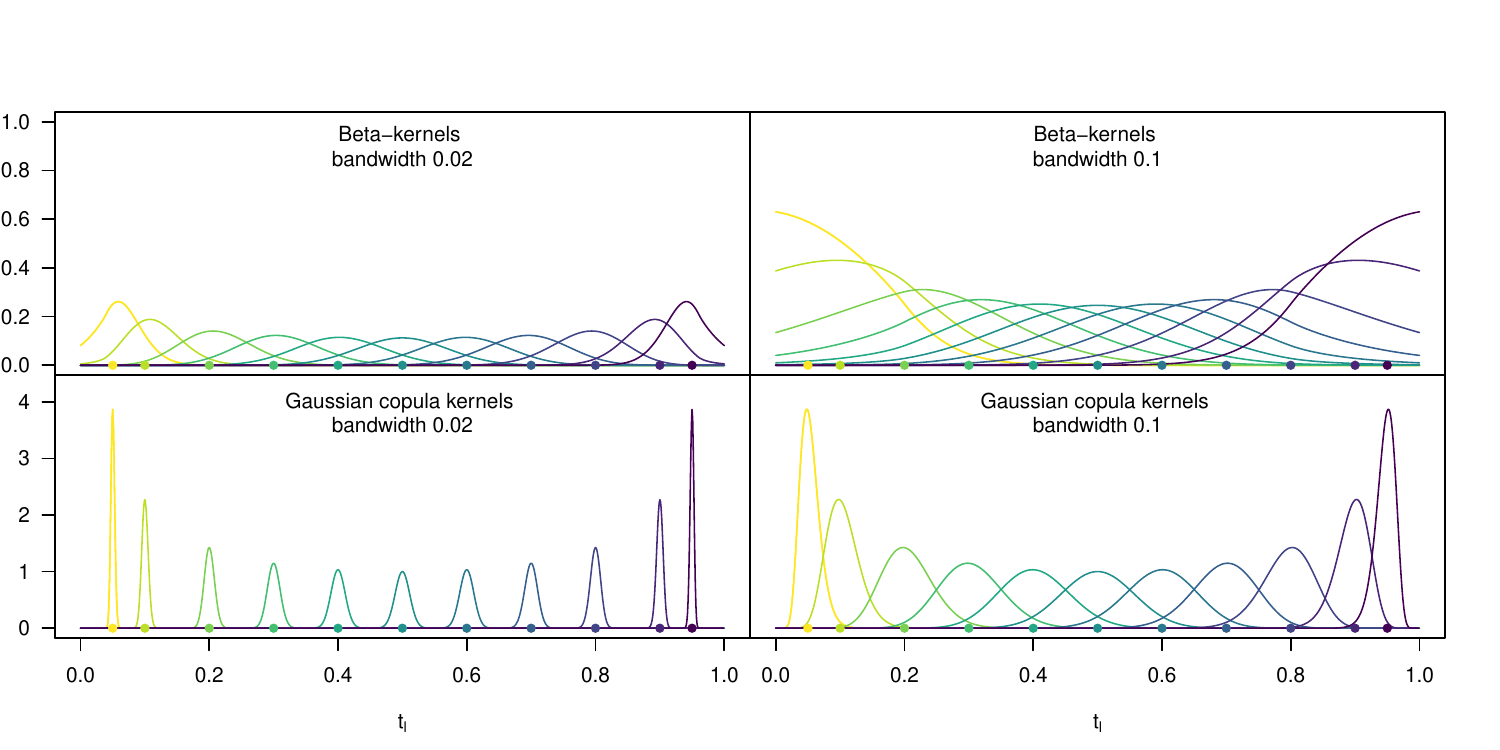}
\end{center}
\vspace{-0.7cm}
\caption{Beta-kernels [top] and Gaussian copula kernels [bottom] for the bandwidths $0.02$ [left] and $0.1$ [right] for different values of $t_l$.\label{kernel-functions}}
\end{figure}

% For the kernel density estimation, we use beta-kernels due to better performances in the forecasts of our model compared to the forecasts resulting from estimating the response densities with Gaussian copula kernels.
% The densities~$fc$ are the resulting estimators.
In our application we use beta-kernels due to better results. % and here give the definition for completeness.
\citet{chen1999} actually presents two versions of beta-kernels, of whom we use the second one, which is also the one depicted in Figure~\ref{kernel-functions}.
%We only give the definition of the second one which has reduced bias compared to the first. This is also the one depicted in Figure~\ref{kernel-functions}.
%It is defined as %\citep{chen1999}
It has reduced bias compared to the first and is defined as %\citep{chen1999}
\begin{align}
\hat{f}^\ast_b(t) := \sum_{l=1}^N w_l \, K^\ast_{t, b} (t_l) \label{kernel_chen_integral_not_one}
\end{align}
for $t \in [0, 1]$ with kernel functions
\begin{align*}
K^\ast_{t, b} (x) &:= \begin{cases}
K_{\rho(t, \, b)\, ,\, \nicefrac{(1-t)}{b}}(x), & t \in [0, 2 b) \\
K_{\nicefrac{t}{b}\, ,\, \nicefrac{(1-t)}{b}}(x), & t \in [2 b, 1- 2 b] \\
K_{\nicefrac{t}{b}\, ,\, \rho(1-t, \, b)}(x), & t \in (1- 2 b, 1],
\end{cases}
\end{align*}
where $\rho(t, \, b) := 2 b^2 + 2.5 - \sqrt{4 b^4 + 6 b^2 + 2.25 - t^2 - \nicefrac{t}{b}}$
and $K_{p, \, q}$ denotes the density function of a $\mathrm{Beta} (p, q)$-distribution.
We slightly modified the original definition of the estimator $\hat{f}^\ast_b$ by including weighting coefficients $w_l$
%like in~(\ref{kernel_standard})
to match the setting in our application.
\citet{chen1999} uses equal weights, i.e., $w_l = \frac1{N}$ for all $l = 1, \ldots , N$.
Note that the resulting estimator usually does not integrate to one as the functions $K^\ast_{t, b} (x)$ are only probability density functions in $x$ but not in~$t$. Therefore, a normalization is necessary to get the estimated density\footnote{As $\hat{f}^\ast_b$ and $\hat{f}_b$ are proportional, they are $\propto$-%$(=_{\Bcal})$-
equivalent $\lambda$-densities with $\lambda$ denoting the Lebesgue measure. But in accordance to usual probability density functions, we use the density as representative that integrates to one.}:
\begin{align}
\hat{f}_b(t) := \frac{\hat{f}^\ast_b(t)}{\int_0^1 \hat{f}^\ast_b(t) \, \dt} \,. \label{kernel_chen}
\end{align}

The optimal bandwidth $b$ can be chosen with unbiased cross-validation (e.g., \citealp{scott2015}).
This is also the default to choose the bandwidth for asymmetric kernels in the \texttt{R} package \texttt{kdensity} \citep{kdensity}, where both beta-kernels and Gaussian copula kernels are implemented, amongst others.
%, see \citet[Chapter~6.5.1.3]{scott2015}.

%The density values of $f \in B^2(\mu)$ at the boundary values are obtained straightforwardly as
%%$f (0) = p_0$ and $f(1) = p_1$, where $p_0$ and $p_1$ are the relative frequencies for a share of $0$ and $1$, respectively.
%the (weighted) relative frequencies for a share of $0$ and $1$, denoted by $p_0$ and $p_1$, respectively.
%To estimate the density values for $s \in (0, 1)$, we compute continuous densities
%%$f_{(0, 1)}: (0, 1) \ra \Rbb^+$ using weighted kernel density estimation with beta-kernels, which preserve the predetermined support $(0, 1)$,
%based on only those couples where both partners have labor income, and multiply it by $p_{(0, 1)} = 1 - p_0- p_1$, which is the relative frequency for a share in $(0, 1)$.
In our application, for each unique combination of covariate values we compute a density $f_{(0, 1)}: (0, 1) \ra \Rbb^+$ using beta-kernels based on dual-earner households.
To determine the bandwidth, we calculate the optimal bandwidth for each of the $552$ densities with unbiased cross-validation and choose the minimal resulting bandwidth as final bandwidth for all densities, yielding a value of $0.02$.
%, see Figure~\ref{kernel-functions} for an illustration of the resulting beta-kernels.
Selecting a smaller bandwidth prevents us from over-smoothing, which may disguise possible effects.
Furthermore, a small bandwidth allows for steep gradients, which indicate a possible discontinuity\footnote{\citet{bertrand2015} consider the share of the wife's income in a couple's total income for married couples in the U.S. and infer that there is a sharp discontinuous drop to the right of~$0.5$. This is in general not confirmed by our data, but we chose a small bandwidth to ensure flexibility of density estimation to capture such a decline.}.
Using the estimated densities $f_{(0, 1)}$ on $(0,1)$, we obtain the response densities on $[0,1]$ as
\begin{align}
&f: [0,1] \ra \Rbb^+ & & s \mapsto \begin{cases}
p_0 , & s=0\\
p_{(0, 1)} \, f_{(0, 1)}(s) , & s \in (0, 1)\\
p_1, & s= 1,
\end{cases}  \label{estimated_densities}
\end{align}
where $p_0$ and $p_1$ are the relative frequencies for a share of $0$ and $1$, respectively, and~$p_{(0, 1)} = 1 - p_0- p_1$ is the relative frequency for a share in $(0, 1)$.
%The estimated densities $fc$ and the relative frequencies only refer to observations in the related \emph{region}, \emph{c\_age}, and \emph{year} and include the \emph{weight} of each household.
% The underlying observation numbers per \emph{year} and \emph{c\_age} for $s = 0$, $s \in (0, 1)$, and~$s=1$ for each \emph{region} are provided in the appendix.

%%%%%%%%%%%%%%%%%%%%%%%%%%%%%%%%%%%

\subsection{Sensitivity Check for varying base-learner degrees of freedom}\label{appendix_soep_sensitivity_check}
In this section, we give some insights leading to the decision to use a model which is theoretically unfair regarding base-learner selection.
First, we perform a sensitivity check comparing it with a model that is fair in the sense that the % marginal
\emph{West\_East} effect base-learner does have the same number of degrees of freedom as other base-learners in the model.
Afterwards, we compare the resulting predictions with the response densities, revealing that the unfair model shows a better fit to the data than the fair one.
Note that both models are estimated with the \texttt{R} package \texttt{FDboost}, which uses effect coding. 
To improve interpretability, we converted those to reference coding for the application.
However, base-learner selection is performed by \texttt{FDboost} on effect coded level, thus we consider effect coding in the following.
For simplicity, we still use the denotation $\betah_{\ldots}, \hat{g}_{\ldots} (year)$ even though these effects are not identical to the reference coded effects denoted like this in the remaining paper.
%This means, that the estimated effects 

To ensure a fair selection process within the gradient boosting algorithm, each base-learner should ideally have the same number of degrees of freedom.
In our model~(10), % \eqref{soep_model}, 
this is not possible for the covariate effects, as the flexible nonlinear effects need a minimum of 2 degrees of freedom, while the intercept $\beta_0$ and $\beta_{\text{\emph{West\_East}}}$ only allow for a maximal value of 1. %, resulting in an unfair model.
%We set the degrees of freedom to 2 for all effects but the intercept $\beta_0$ and $\beta_{\text{\emph{West\_East}}}$, which are set to 1.
Regarding base-learner selection, $\beta_{\text{\emph{West\_East}}}$ thus is theoretically at a disadvantage compared to the other main effects.
To study the severity of this disadvantage, we compare our model with another model, which is fair regarding base-learner selection. % (but yields bad predictions, especially at the boundary values over time).
This is reached by dividing the degrees of freedom in direction of the share in half for all effects but $\beta_0$ and $\beta_{\text{\emph{West\_East}}}$, in both, the continuous and discrete model.
%\begin{itemize}
%\item
%unfair model: intercept and old/new-effect have less degrees of freedom (each 1) than the other effects (each 2); the timeformula (share-direction) is the same in the Kronecker-product with each effect (df\_d = 2 for discrete model; df\_c = 53 for continuous model)
%\item
%fair model: to compensate the limitation in degrees of freedom for the intercept and old/new-effect, we keep the number of degrees of freedom the same in the Kronecker-multiplied timeformula (df\_d = 2; df\_c = 53) for these two effects, but divide it in half for all remaining effects (df\_d = 1; df\_c = 26.5).
%\end{itemize}
Apart from that, the models are specified identically to the ones presented in the main manuscript.
Again, we determine the stopping iterations based on $25$ bootstrap samples, respectively, resulting in $490$ for the continuous and $735$ for the discrete model.
For simplicity, we refer to the resulting models as \emph{fair} models in contrast to the \emph{unfair} models of choice in the following.
%In all comparisons of our sensitivity check, we consider the continuous and discrete models separately. % , as they were estimated separately.
In our sensitivity check, we first compare the selection frequencies, the crucial parameter for the fairness of a model.
For further insights, we also consider the in-sample risk reduction and the estimated effects for $\beta_{\text{\emph{West\_East}}}$ in the fair vs. the unfair models. % to see how drastic the differences are.

%\vspace{-0.3cm}
\begin{figure}[H]
\begin{center}
\includegraphics[width=0.49\textwidth]{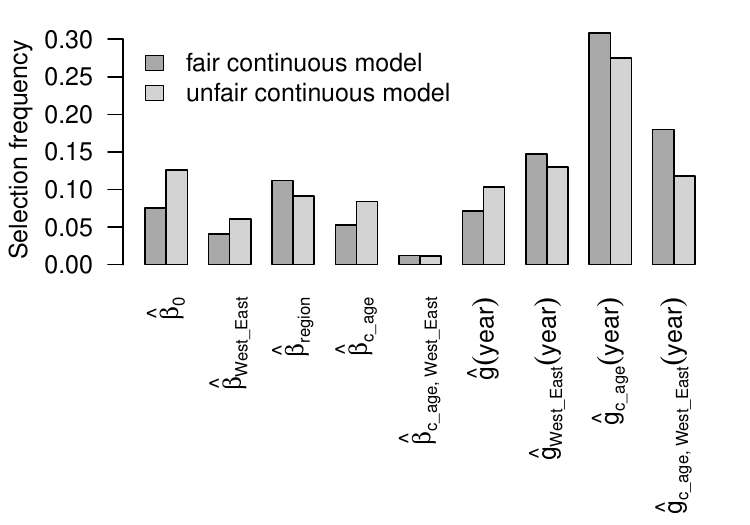}
\includegraphics[width=0.49\textwidth]{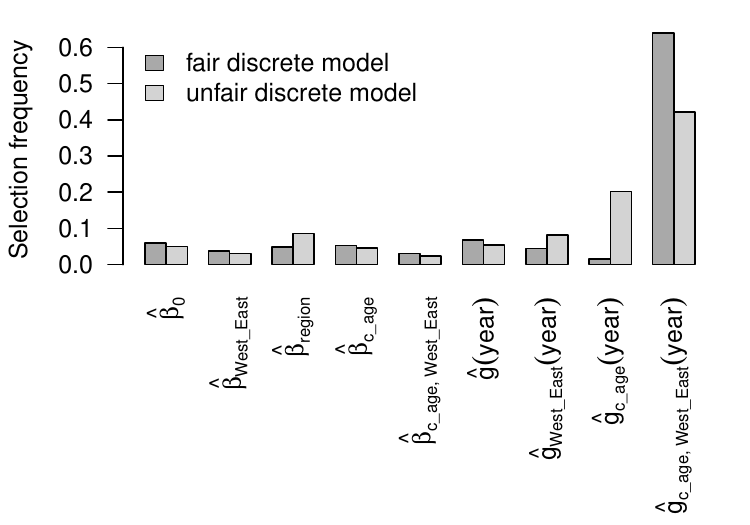}
\end{center}
%\vspace{-0.8cm}
\caption{Selection frequencies of the different (effect coded) effects for fair vs. unfair models for continuous [left] and discrete [right].\label{selfreq}}
\end{figure}
Figure~\ref{selfreq} shows the selection frequencies of each effect in the continuous and discrete models comparing the fair with the unfair models, respectively.
The left side shows the continuous models.
Here, $\beta_{\text{\emph{West\_East}}}$ gets selected even more often in the unfair model -- where it is theoretically disadvantaged -- than in the fair model.
Considering the discrete models (right), $\beta_{\text{\emph{West\_East}}}$ is selected slightly less often than in the fair model, but the difference does not seem severe.

%\pagebreak
%\section{Risk reduction}

%\vspace{-0.3cm}
%\begin{figure}[H]
%\begin{center}
%\includegraphics[width=0.49\textwidth]{Images/Risk_reduction_fair_vs_unfair_c.pdf}
%\includegraphics[width=0.49\textwidth]{Images/Risk_reduction_fair_vs_unfair_d.pdf}
%\end{center}
%\vspace{-0.8cm}
%\caption{Risk reduction of the different effects for fair (left) vs. unfair (right) models for continuous (top) and discrete (bottom).\label{risk}}
%\end{figure}

\vspace{-0.3cm}
\begin{figure}[H]
\begin{center}
\includegraphics[width=0.49\textwidth]{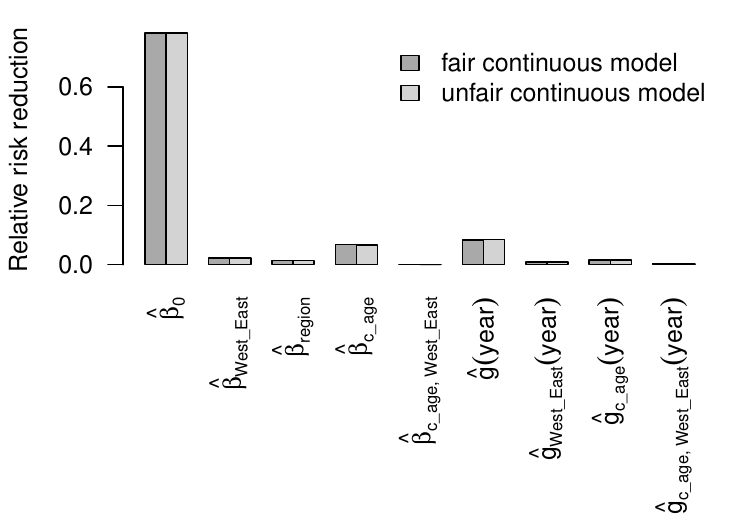}
\includegraphics[width=0.49\textwidth]{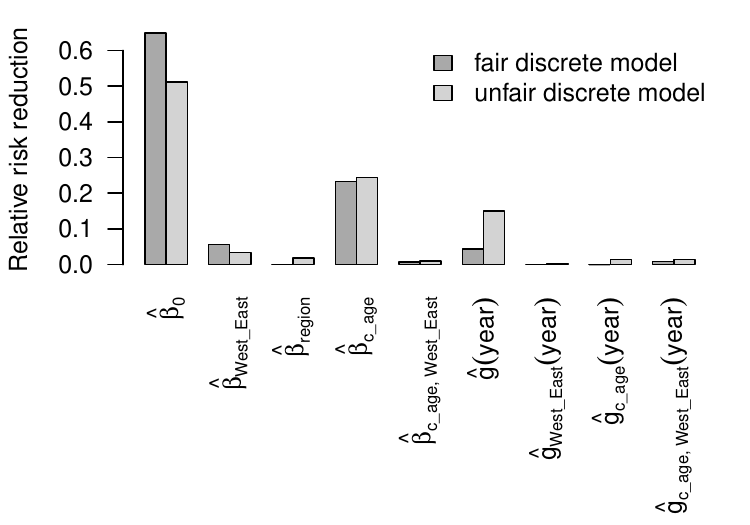}
\end{center}
\vspace{-0.8cm}
\caption{Relative in-sample risk reduction of the different (effect coded) effects for fair vs. unfair models for continuous [left] and discrete [right].\label{relrisk}}
\end{figure}

%The fairness of the selection process is only based on the selection frequencies.
%Thus, we already can conclude that the disadvantage of $\beta_{\text{\emph{West\_East}}}$ is not severe in our case. However, to
% To get a more complete impression, we also consider
The relative in-sample risk reduction of the effects in the different models is illustrated in Figure~\ref{relrisk}.
For the continuous models (left), the risk reductions per effect are almost identical in both models, which indicates that there is no disadvantage for $\beta_{\text{\emph{West\_East}}}$ in the unfair model.
For the discrete models (right), $\beta_{\text{\emph{West\_East}}}$ again deems more important in the fair model than in the unfair one.

%\pagebreak
%\section{Estimated effects}

%\vspace{-0.3cm}
\begin{figure}[H]
\begin{center}
\includegraphics[width=0.49\textwidth]{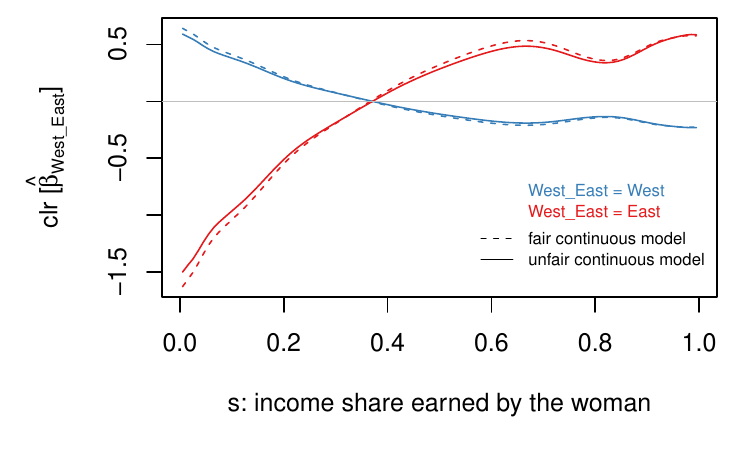}
\includegraphics[width=0.49\textwidth]{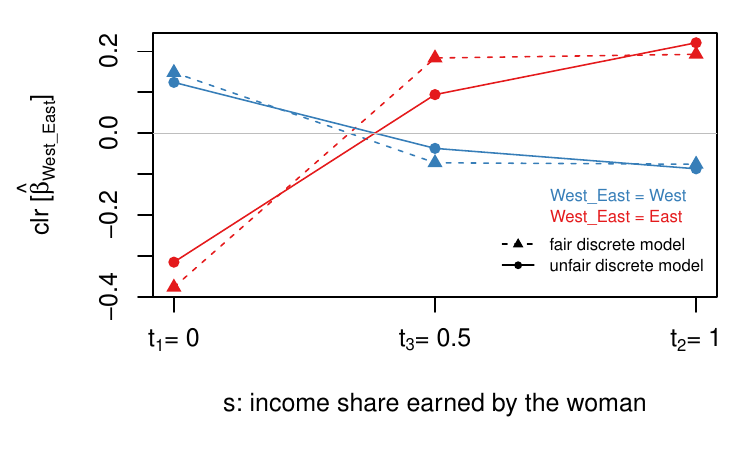}
%\\
%\includegraphics[width=0.3\textwidth]{Images/legend_old_new.pdf}
\end{center}
\vspace{-0.8cm}
\caption{Clr transformed estimated (effect coded) effects of \emph{West\_East} for fair vs. unfair models for continuous [left] and discrete [right].\label{effects}}
\end{figure}

Finally, we compare the clr transformed estimated effects of $\beta_{\text{\emph{West\_East}}}$ in the different models in Figure~\ref{effects}. % for completeness.
While this does not allow conclusions about the fairness of the models, it might be disconcerting, if the estimated effects were completely different.
However, this is not the case.
We obtain very similar effects in the continuous models (left).
Regarding the discrete models (right), the values differ more (relatively), but the trend is the same.

%\section{Conclusion}
In summary, we observe almost no differences in the continuous models between a fair and unfair model specification.
In contrast, there are slight differences in the discrete models.
However, they are not too severe, so that $\beta_{\text{\emph{West\_East}}}$ does not seem to be at a large disadvantage. % compared to the other effects.

%%%%%%%%%%%%%%%%%%%%%%%%%%%%%%%%%%%

We decided to prefer the unfair model to the fair one because of the fit to the data.
Figure~\ref{figure_predictions_fair} shows the predicted densities resulting from the fair model, Figure~\ref{figure_densities} the response densities, and Figure~\ref{figure_predictions} the predicted densities resulting from the unfair model.
All three figures are structured as follows.
In the upper part, they illustrate the respective densities for all six \emph{regions} and all three \emph{c\_age} groups. 
The densities are shown in one panel for all \emph{years}, respectively, with a color gradient and different line types indicating the \emph{year}.
The density values at the boundaries $0$ and~$1$ %(corresponding to $p_0$ and $p_1$) 
are represented as dashes, shifted slightly outwards for better visibility.
The lower part of the figures show their development over time more clearly.
For the response densities (Figure~\ref{figure_densities}), they are represented as dashes again (green and red, respectively), while the relative frequency %$p_{(0, 1)}$, which equals the Lebesgue integral of $f$, 
of dual-earner households is illustrated via blue circles.
For the predicted densities (Figures~\ref{figure_predictions_fair} and \ref{figure_predictions}), the smooth trend over time is shown by different types of lines, but using the same colors as for the response densities.

First, we compare the predictions from the fair model, i.e., Figure~\ref{figure_predictions_fair}, with the response densities, i.e., Figure~\ref{figure_densities}.
In general, the shapes of the predicted densities for $s \in (0, 1)$ match the ones of the response densities for the different \emph{regions} and values of \emph{c\_age} (upper parts of the figures):
The densities corresponding to \emph{regions} in \emph{West} Germany (\emph{northeast}, \emph{west}, \emph{southwest}, \emph{south}) show more probability mass at smaller income shares for couples with minor children (\emph{0-6} and~\emph{7-18}) compared to couples without minor children (\emph{other}), while the densities for \emph{East} Germany (\emph{east}, \emph{northeast}) show more symmetric distributions regardless of the age of the youngest child.
However, the absolute values of the predicted densities resulting from the fair model are at the same level for couples with children aged 0-6 years as for couples with children aged 7-18 years.
Regarding the response densities, this is not the case.
Here, the absolute values of the densities corresponding to \emph{0-6} are lower than the ones for \emph{7-18}.
Furthermore, the trend over the years is not covered well, especially in the discrete model, which shows in the relative frequencies (lower part of the figures):
For the predicted densities resulting from the fair models, we expect an increase of non-working women ($p_0$) and a decrease of dual-earner households ($p_{(0, 1)}$) with time in all regions and for all values of \emph{c\_age}.
For the response densities, these developments are the other way around: $p_0$ tends to decrease, while $p_{(0, 1)}$ tends to increase!
In contrast, comparing the predicted densities resulting from the unfair model (Figure~\ref{figure_predictions}) with the response densities (Figure~\ref{figure_densities}), these issues do not appear, while the shapes of the predicted densities in $s \in (0, 1)$ are still matched nicely.
Finally, we consider the sum of squared errors (SSE) as defined in~(4) % \eqref{sse} 
for both models.
It also leads to the decision to prefer the unfair model as its SSE is only $1436$ and thus smaller than the SSE of the fair model, which is $1704$.
%$\SSE(h_{\text{unfair}}) = 1436 < 1704 = \SSE(h_{\text{fair}})$

% Relative Risk
%> reliMSE_pred_fair
%$reliMSE
%[1] 0.228449
%$num
%[1] 1703.745
%$den
%[1] 7457.878
%
%> reliMSE_pred
%$reliMSE
%[1] 0.1926089
%$num
%[1] 1436.453
%$den
%[1] 7457.878

%Finally, we decided to discard the fair model and use the unfair one.
Apparently, the fair model is not flexible enough to fit the data well due to the reduced degrees of freedom for the basis over $(0,1)$ for the continuous model and over $\{0, 1, 0.5\}$ for the discrete one.
Thus, we decided to discard the fair model and keep the unfair one instead.

\vspace{-1cm}
\begin{figure}[H]
%\begin{center}
\begin{minipage}{0.93\textwidth}
\includegraphics[width=\textwidth]{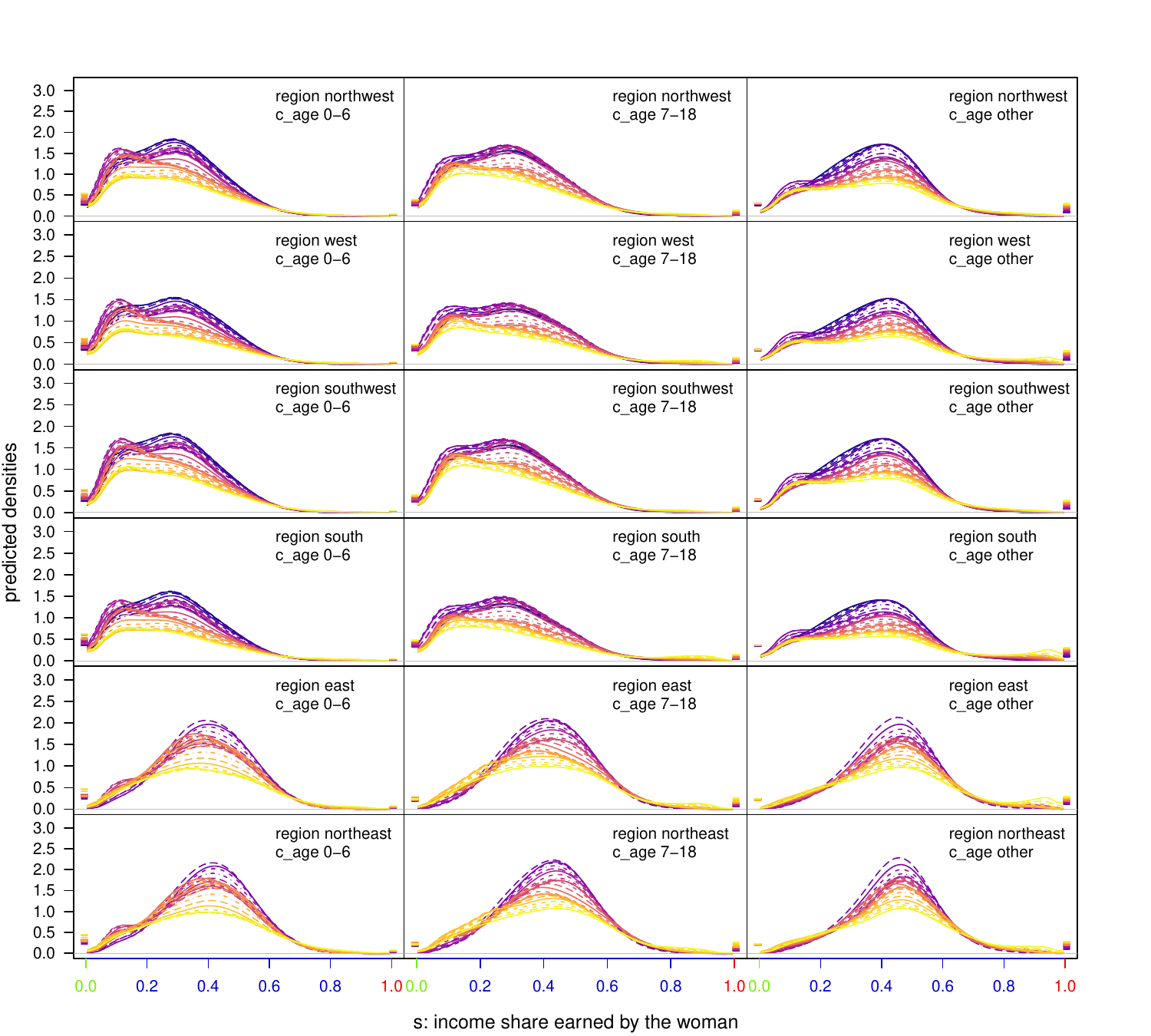} %\\[2mm]
\end{minipage}
\begin{minipage}{0.059\textwidth}
\includegraphics[width=1.2\textwidth]{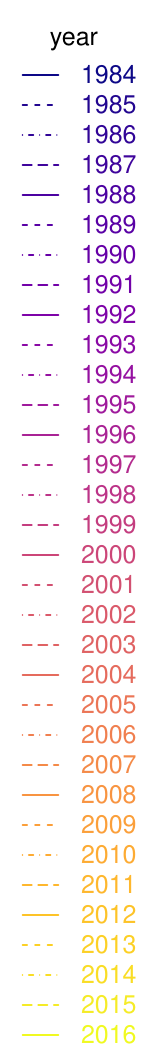}
\end{minipage}
\begin{minipage}{0.93\textwidth}
\vspace*{2mm}
\hspace*{0.8cm}
\includegraphics[width=0.87\textwidth]{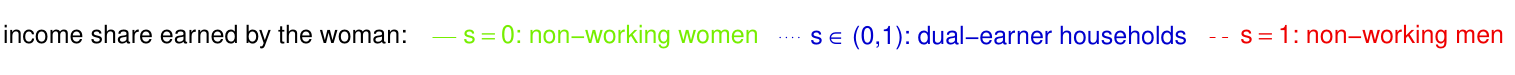} \\[-9mm]
\includegraphics[width=\textwidth]{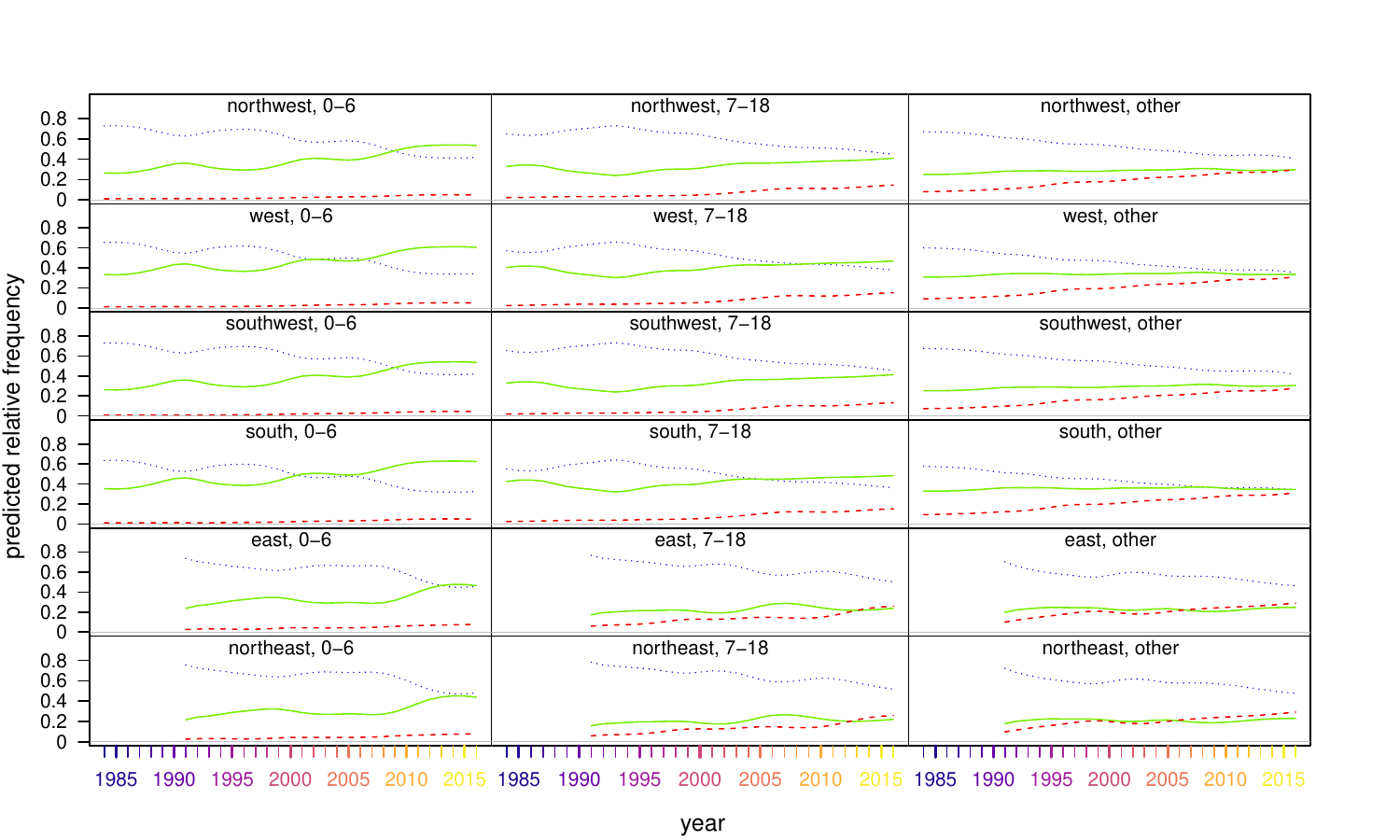}
\end{minipage}
\begin{minipage}{0.059\textwidth}
~
\end{minipage}
\caption{Predicted densities [upper $6 \times 3$ panels] and corresponding relative frequencies [lower $6 \times 3$ panels] resulting from finally discarded fair models for all \emph{regions} [rows] for all three values of \emph{c\_age} [columns].\label{figure_predictions_fair}}
\end{figure}

%%%%%%%%%%%%%%%%%%%%%%%%%%%%%%%%%%%

%\subsection{Densities}
\vspace{-1cm}
\begin{figure}[H]
%\begin{center}
\begin{minipage}{0.93\textwidth}
\includegraphics[width=\textwidth]{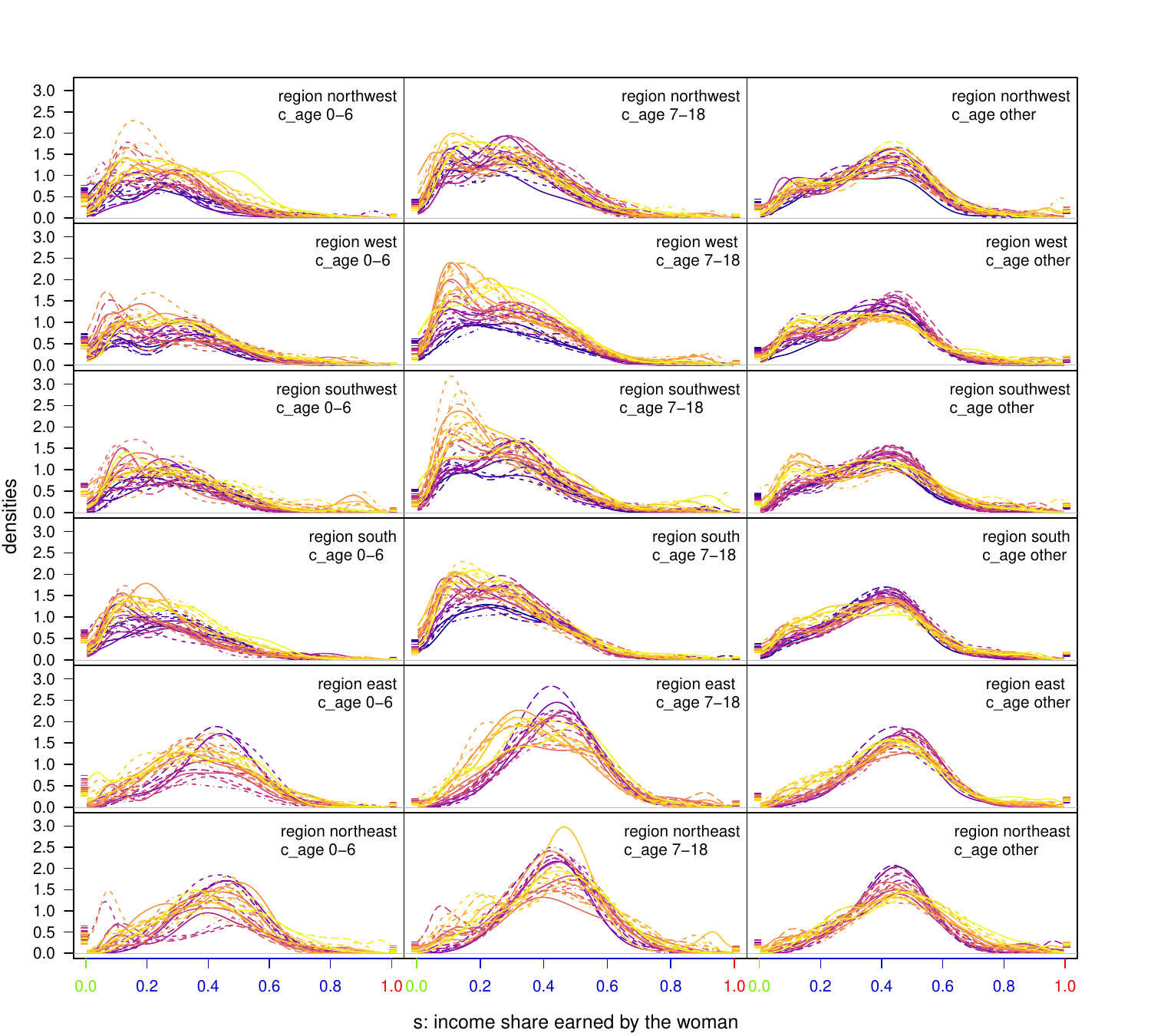} %\\[2mm]
\end{minipage}
\begin{minipage}{0.059\textwidth}
\includegraphics[width=1.2\textwidth]{Images/year_legend_v_1.pdf}
\end{minipage}
\begin{minipage}{0.93\textwidth}
\vspace*{2mm}
\hspace*{0.8cm}
\includegraphics[width=0.87\textwidth]{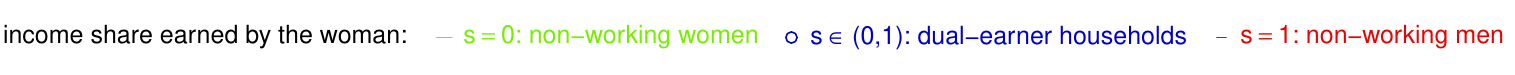} \\[-9mm]
\includegraphics[width=\textwidth]{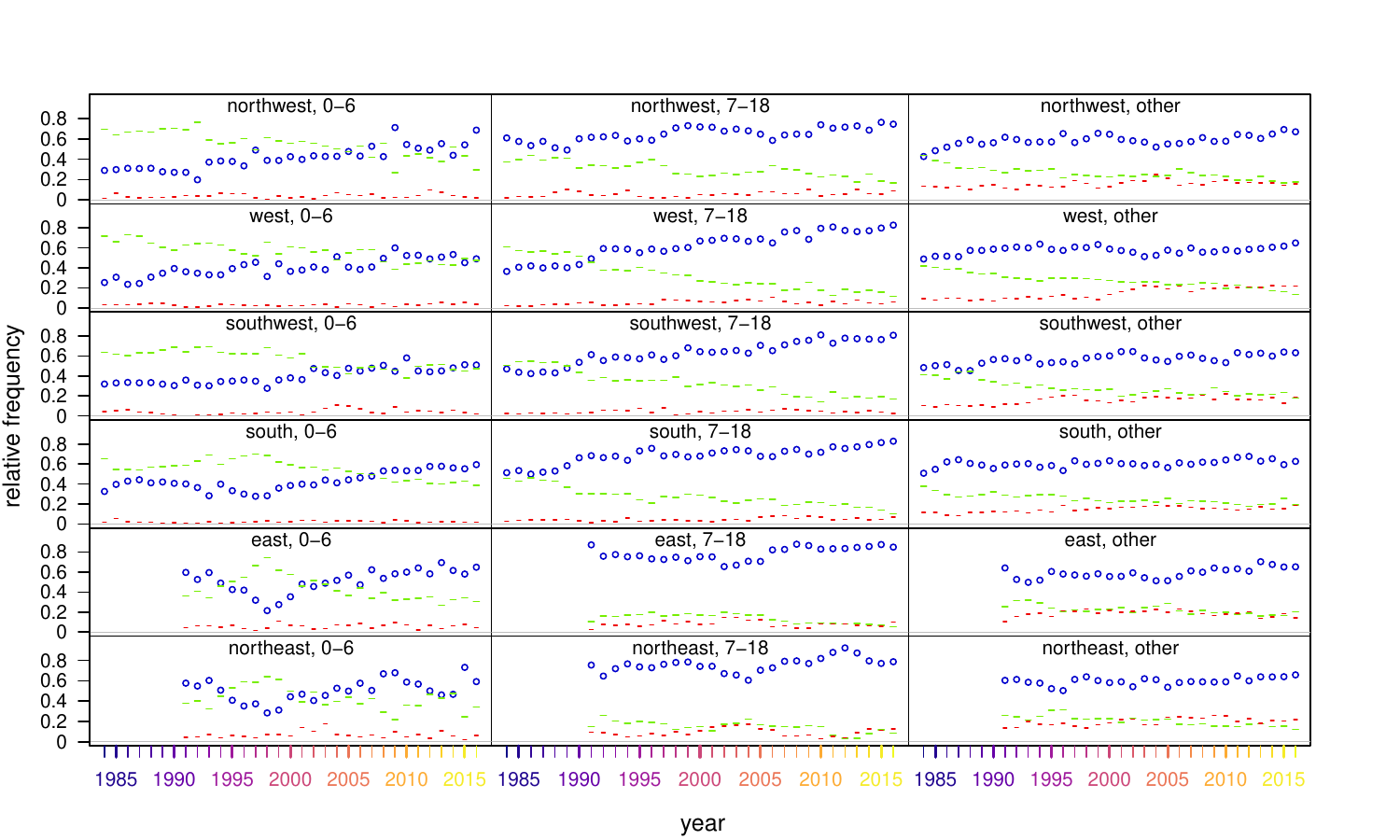}
\end{minipage}
\begin{minipage}{0.059\textwidth}
~
\end{minipage}
\caption{Response densities [upper $6 \times 3$ panels] and corresponding relative frequencies [lower $6 \times 3$ panels] for all \emph{regions} [rows] for all three values of \emph{c\_age} [columns].\label{figure_densities}}
\end{figure}

\vspace{-1cm}
\begin{figure}[H]
%\begin{center}
\begin{minipage}{0.93\textwidth}
\includegraphics[width=\textwidth]{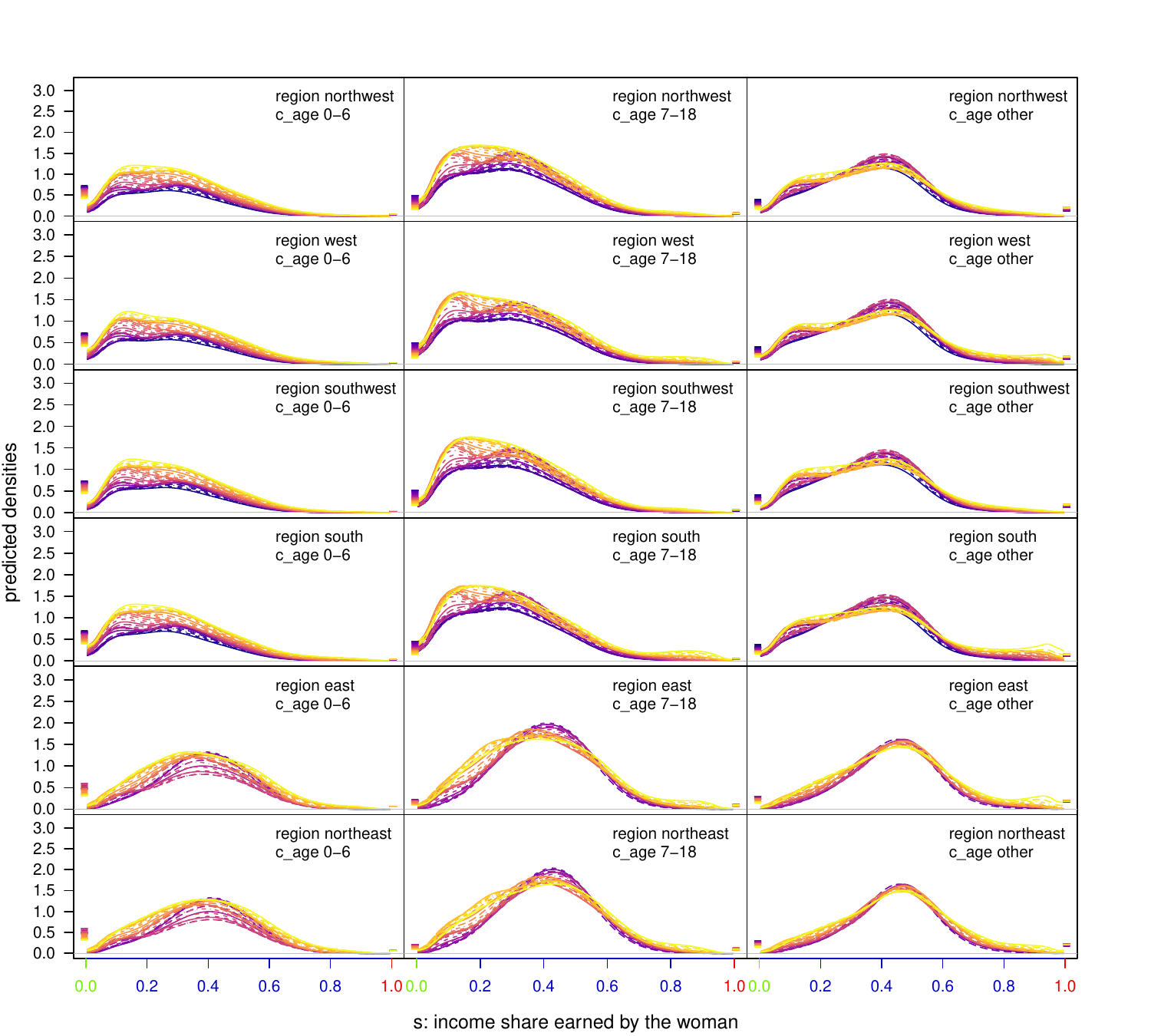} %\\[2mm]
\end{minipage}
\begin{minipage}{0.059\textwidth}
\includegraphics[width=1.2\textwidth]{Images/year_legend_v_1.pdf}
\end{minipage}
\begin{minipage}{0.93\textwidth}
\vspace*{2mm}
\hspace*{0.8cm}
\includegraphics[width=0.87\textwidth]{Images/legend_pred_hist_years_h_long.pdf} \\[-9mm]
\includegraphics[width=\textwidth]{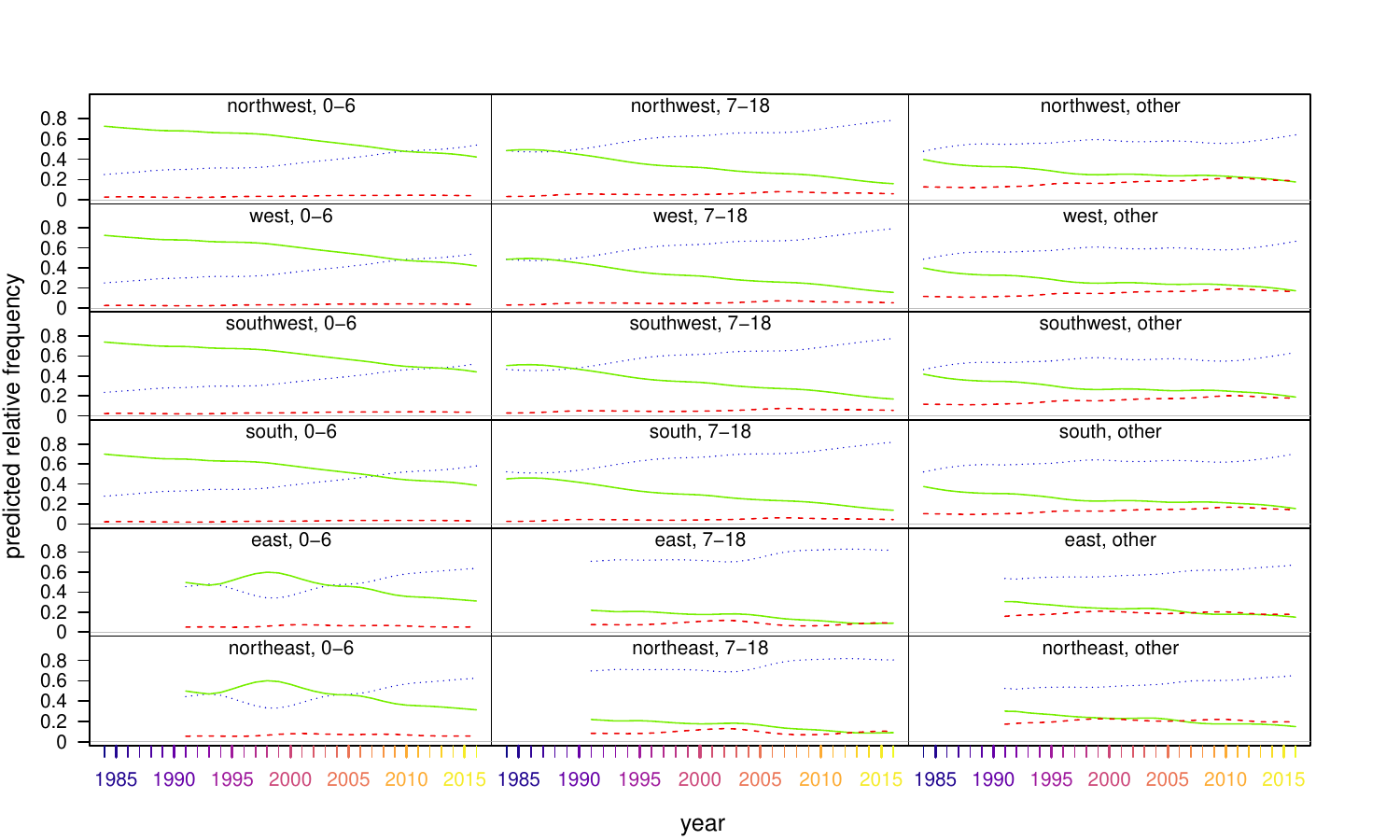}
\end{minipage}
\begin{minipage}{0.059\textwidth}
~
\end{minipage}
\caption{Predicted densities [upper $6 \times 3$ panels] and corresponding relative frequencies [lower $6 \times 3$ panels] resulting from finally used unfair models for all \emph{regions} [rows] for all three values of \emph{c\_age} [columns].\label{figure_predictions}}
\end{figure}

\subsection{Estimated Effects}\label{appendix_estimated_effects}
This section shows all estimated effects of model~(10) % (\ref{soep_model})
with  Figures~\ref{appendix_estimated_old_new}-\ref{appendix_estimated_year_old_new_cgroup} structured similar to %Figures~\ref{estimated_old_new}-\ref{estimated_year}.
Figure~2. % \ref{estimated_child_group}.
The left side shows the perturbation of the intercept with the respective effect and other reasonable effects (e.g., the main effects for interaction effects).
The circles at $0.5$ correspond to the Lebesgue integral of the respective function, i.e., the expected relative frequency of dual-earner households.
On the right side, we illustrate the clr transformed effects to easily allow their interpretation via (log) odds ratios as described in Section~3.2. % \ref{chapter_interpretation}.
Example interpretations are given for Figures~\ref{appendix_estimated_old_new} and~\ref{appendix_estimated_year}.
%The circles at $t_3 = 0.5$ correspond to the Lebesgue integral of the respective function.
%In $\B$, i.e., for the plots on the left, this corresponds to the expected relative frequency of dual-earner households.
%In $\Ln$, i.e., for the plots on the right, they form the discrete component of the clr transformed effect together with its values at $t_1 = 0$ and $t_2 = 1$, see~(\ref{decomposition_discrete_clr}).

% old_new
\begin{figure}[H]
\begin{center}
\includegraphics[width=0.49\textwidth]{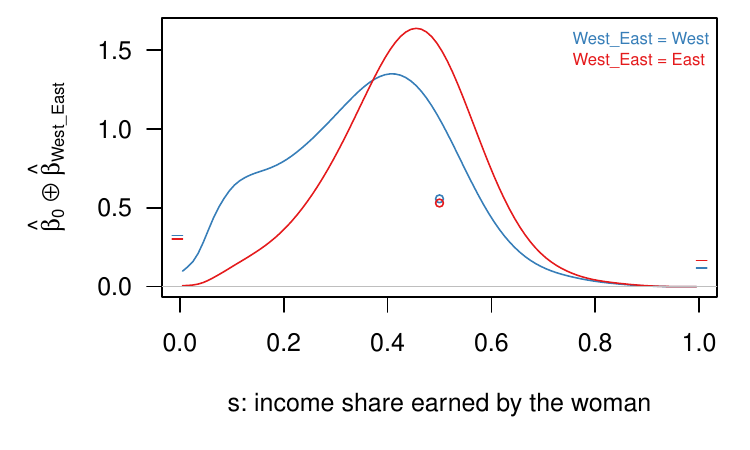}
\includegraphics[width=0.49\textwidth]{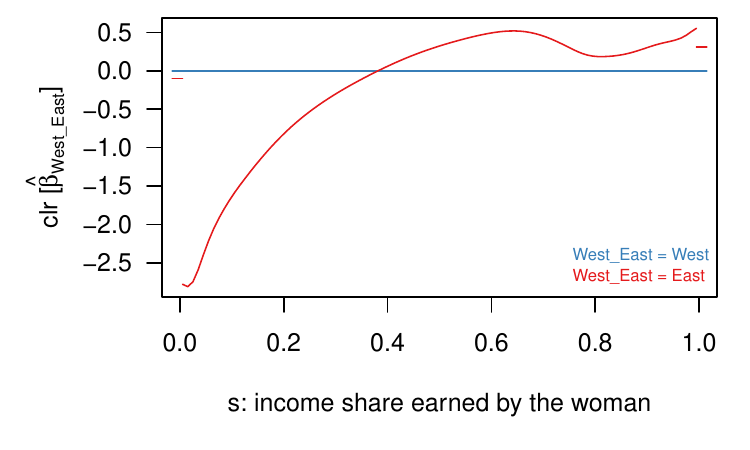}
\end{center}
\vspace{-0.5cm}
\caption{Expected densities for couples without minor children in 1991 for \emph{West} vs. \emph{East} Germany [left] and clr transformed estimated effects of \emph{West\_East} [right]. \label{appendix_estimated_old_new}}
\end{figure}

Figure~\ref{appendix_estimated_old_new} illustrates the estimated effect of \emph{West\_East}.
As \emph{West} is the reference category, we have $\betah_0 \oplus \betah_{West} = \betah_0$ and $\clr [\betah_{West}] = 0$.
The left part of the figure shows the expected densities for couples living in \emph{West} %Germany in comparison to 
versus \emph{East} Germany for the reference, i.e., couples without minor children %(\emph{c\_age other}) 
in 1991.
For \emph{West} Germany, the expected density over $(0, 1)$ has a smaller mode %, a second maximum further left 
and probability mass shifted to the left compared to \emph{East} Germany.
Non-working women ($s = 0$) are more frequent in \emph{West} than in \emph{East} Germany, while dual-earner households (circles at $s = 0.5$) and single-earner women ($s = 1$) are more frequent in \emph{East} Germany.
Alternatively, we can interpret the log odds ratio of $\betah_{East}$ and $\betah_{West}$ for $s$ compared to~$t$ for any $s, t \in [0, 1]$ of interest (right).
% Due to $\clr [\betah_{West}] = 0$ it simplifies to 
It %is given by 
equals the log odds of $\betah_{East}$, i.e., $\clr [\betah_{East}](s) - \clr [\betah_{East}] (t)$, corresponding to vertical differences in the red curve.
%I.e., the log odds ratio of $\betah_{East}$ and $\betah_{West}$ equals the log odds of $\betah_{East}$ for $t$ compared to~$s$.
First, we compare the boundary values, %corresponding to 
i.e., single-earner households.
The log odds ratio for $s = 1$ compared to $t = 0$ is $0.31 - (-0.44) = 0.75$, which means that the odds for single-earner %women 
versus non-working women in \emph{East} Germany are $\exp(0.75) \approx 2.12$ times the odds in \emph{West} Germany.
To compare dual-earner households with non-working women, consider the log odds ratio for $s \in (0, 1)$ and $t = 0$, which is negative for $s < 0.23$ and positive otherwise.
E.g., the log odds ratio for $s = 0.5$ compared to $t = 0$ is $ 0.53 -(-0.44) = 0.97$, i.e., the odds for equal earning couples versus non-working women in \emph{East} Germany are $\exp(0.97) \approx 2.64$ times the odds in \emph{West} Germany.
The log odds ratio for $s = 1$ (single-earner women) compared to $t \in (0, 1)$ (dual-earner households) is positive for $t < 0.42$ and negative for larger $t$.
E.g., for $t = 0.5$, the log odds ratio is $0.31 - 0.53 = -0.22$, i.e., the odds for single-earner women versus equal earning couples in \emph{East} Germany are $\exp(-0.22) \approx 0.8$ times the odds in \emph{West} Germany.
%To compare single-earner with dual-earner households, consider $t \in (0, 1)$.
%For both, $s = 0$ and $s = 1$, the log odds ratio is positive for small $t$ and negative otherwise, with the threshold at $t = 0.23$ and $t = 0.42$, respectively.
%E.g., the log odds ratio for $s = 0$ compared to $t = 0.5$ is $-0.44 - 0.53 = -0.97$, while for $s = 1$ it is $0.31 - 0.53 = -0.22$.
%Thus, the odds for non-working women/single-earner women versus equal earning couples in \emph{East} Germany are $\exp (-0.97) \approx 0.38$ times/$\exp(-0.22) \approx 0.8$ times the odds in \emph{old} ones.
%
Within dual-earner households, i.e., for $s, t \in (0, 1)$, the log odds ratio of $\betah_{East}$ and $\betah_{West}$ for $s$ compared to $t$ is mostly positive for $t < s$ as $\clr [\betah_{East}]$ increases monotonically %(with an exception 
(except between $0.7$ and $0.8$).
Thus, the odds for a larger %income share 
versus a smaller income share are larger in \emph{East} than in \emph{West} Germany.

% region
\begin{figure}[H]
\begin{center}
\includegraphics[width=0.49\textwidth]{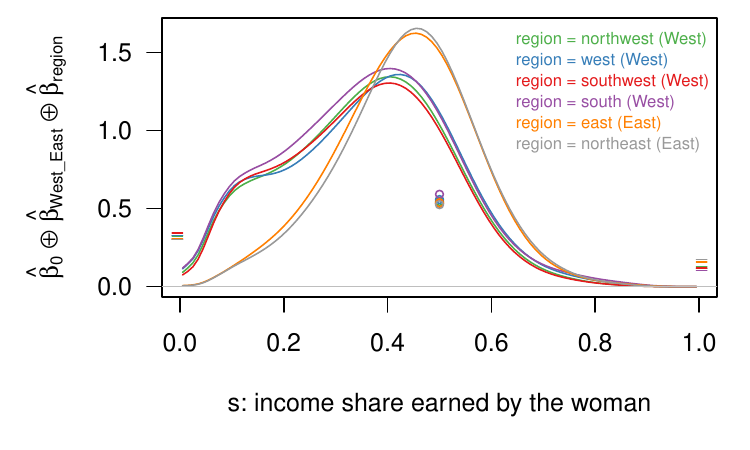}
\includegraphics[width=0.49\textwidth]{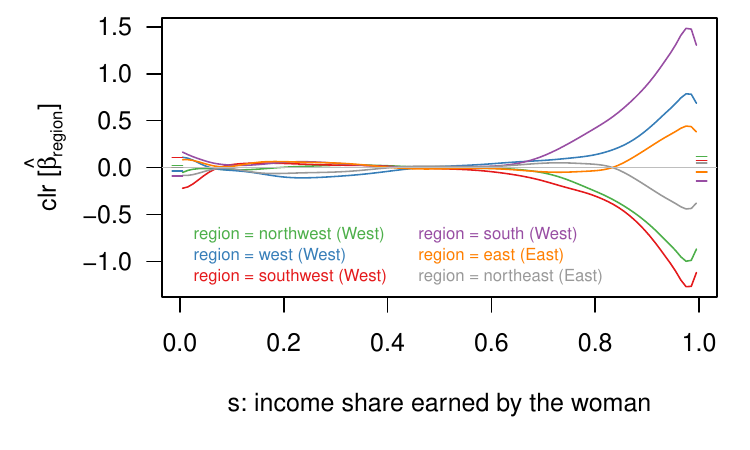}
\end{center}
\vspace{-0.5cm}
\caption{Expected densities for couples without minor children in 1991 living in the different \emph{regions} [left] and clr transformed estimated effects of \emph{region} [right].}
\end{figure}

% c_age
\begin{figure}[H]
\begin{center}
\includegraphics[width=0.49\textwidth]{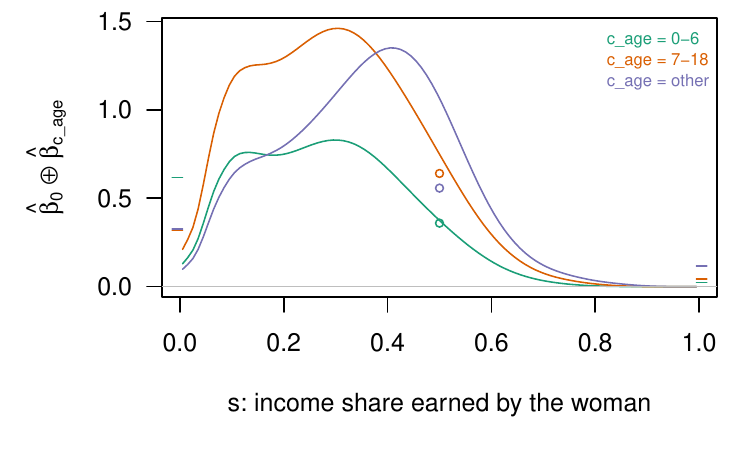}
\includegraphics[width=0.49\textwidth]{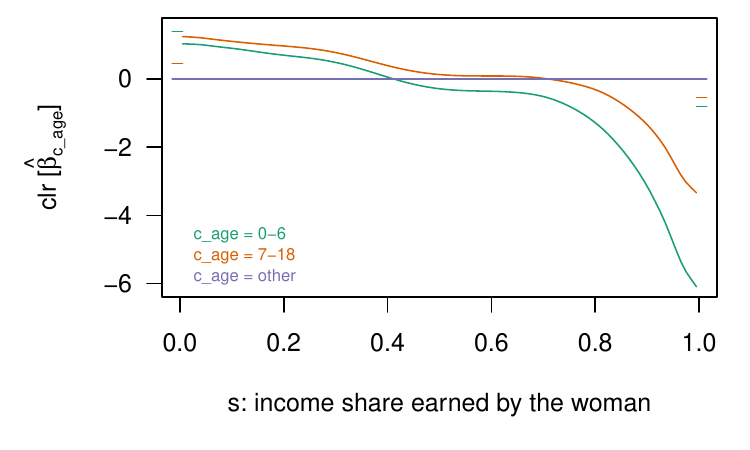}
\end{center}
\vspace{-0.5cm}
\caption{Expected densities for couples living in \emph{West} Germany in 1991 for all three values of \emph{c\_age} [left] and clr transformed estimated effects of \emph{c\_age} [right].}
\end{figure}

% old_new, c_age
\begin{figure}[H]
\begin{center}
\includegraphics[width=0.49\textwidth]{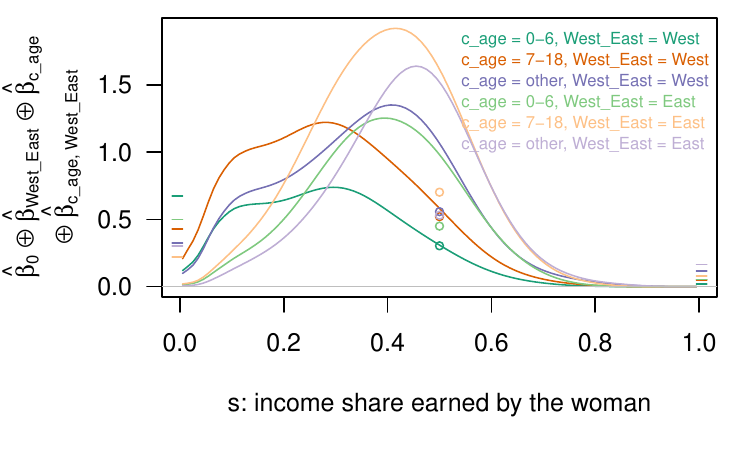}
\includegraphics[width=0.49\textwidth]{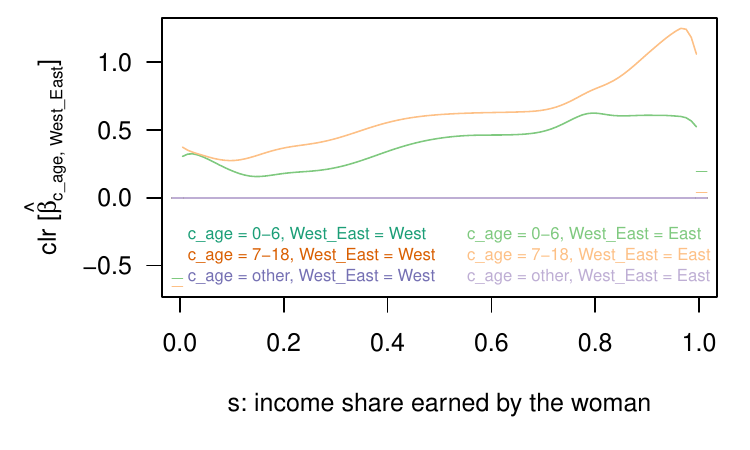}
\end{center}
\vspace{-0.5cm}
\caption{Expected densities for couples in 1991 for all three values of \emph{c\_age} living in \emph{West} vs. \emph{East} Germany [left] and clr transformed estimated interaction effects of \emph{c\_age} and \emph{West\_East} [right].}
\end{figure}

% year
\begin{figure}[H]
\begin{center}
\includegraphics[width=0.49\textwidth]{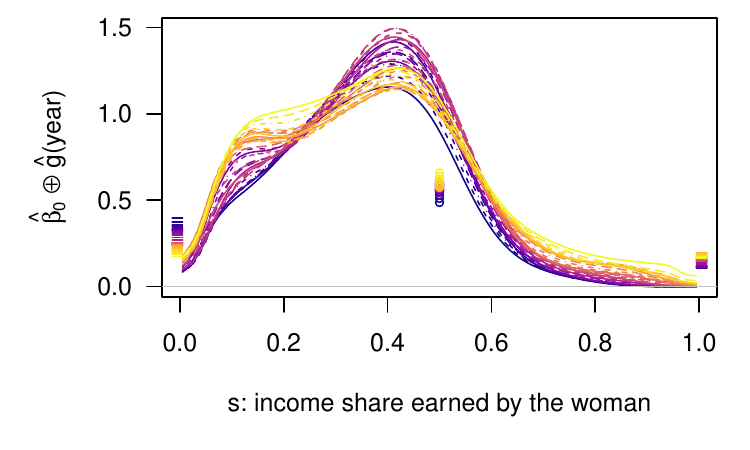}
\includegraphics[width=0.49\textwidth]{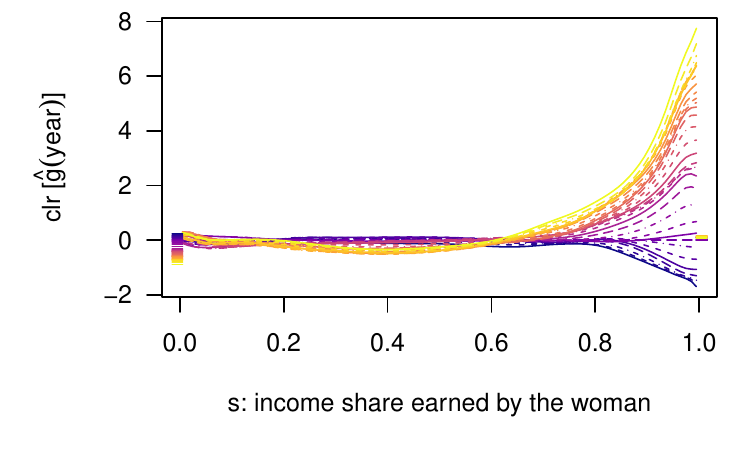}
\includegraphics[width=0.7\textwidth]{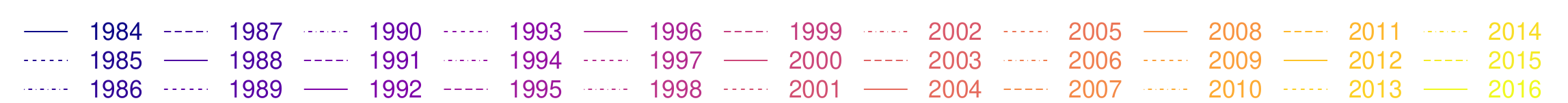}
\end{center}
\vspace{-0.5cm}
\caption{Expected densities for couples without minor children living in \emph{West} Germany over time [left] and clr transformed estimated effects of \emph{year} [right].\label{appendix_estimated_year}}
\end{figure}

Figure~\ref{appendix_estimated_year} shows the flexible nonlinear effect of \emph{year}. %, using the same color gradient and line types as in Figure~\ref{original_densities_south_east}. 
Here, we observe a clear temporal trend towards more dispersed distributions of shares in~$(0, 1)$.
In the left panel, this is clearly visible.
The mode of the expected densities for couples without minor children living in \emph{West} Germany % (\emph{c\_age other})
stays approximately the same (about $0.4$) with probability mass shifting outwards over time.
In more recent years, the expected densities tend to have a second maximum further left and a heavier tail on the right.
Furthermore, the expected relative frequency of non-working women ($s = 0$) decreases with time, while the frequency of single-earner women ($s = 1$) increases to now more similar levels.
%
%Consider the right part of the figure. % and note that $\clr[\gh(1991)] = 0$ as 1991 is the reference for \emph{year}.
The clr transformed effects (right) support our finding of dispersing densities on~$(0, 1)$.
%On $(0, 1)$, the log odds ratios of $\gh(year)$ and $\gh(1991)$ for all $t$ in a range $I_0$ of small income shares (e.g., $I_0 = (0, 0.3)$) compared to $s = 0.4$ are negative for earlier \emph{year} and positive for later \emph{year}, i.e., the odds of $\gh(year)$ for earlier \emph{year} are smaller than the odds of $\gh(year)$ for later \emph{year}.
%On~$(0, 1)$, 
Before 1991, the clr transformed effects tend to be smaller for low and high income shares (e.g., for $s \in A = (0, 0.3) \cup (0.6, 1)$) than for income shares in between (e.g., for $t \in B = (0.35, 0.45)$).
After 1991, this reverses. % with clr transformed effects being larger for outer income shares ($A$) % ($I_s$) 
% than for more central ones ($B$) % ($I_t$).
Thus, using Proposition~3.1 (a), % \ref{thm:oddsratio} (\ref{thm:odds_probabilities_inequality}), % \eqref{odds_probabilities_inequality}, 
%the odds of $\gh(year)$ for $t$ compared to $s$ are smaller for earlier \emph{years} than for later \emph{years} for all $t \in I_t$ and all $s \in I_s$.
%This implies that 
the odds of the probabilities for the outer region $A$ versus the more central region $B$ are smaller for earlier \emph{years} than in later \emph{years}. %, see Section~\ref{chapter_interpretation}.
We can conclude that the probability of $A$ increases and/or the probability of $B$ decreases with time.
The clr transformed effects get particularly large for high income shares $s < 1$, which is not visible on the level of the original densities, where the absolute values of the corresponding densities in this area are small (left).
This is due to the multiplicative effect structure, for which small (absolute) differences can correspond to large relative differences within the densities.

% year, old_new
\begin{figure}[H]
\begin{center}
\includegraphics[width=0.49\textwidth]{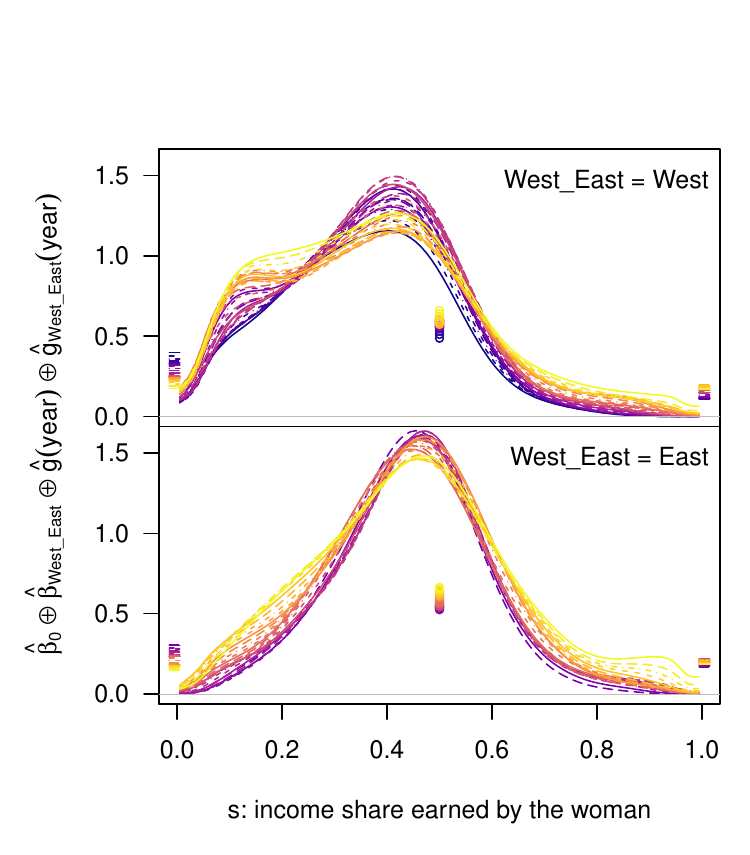}
\includegraphics[width=0.49\textwidth]{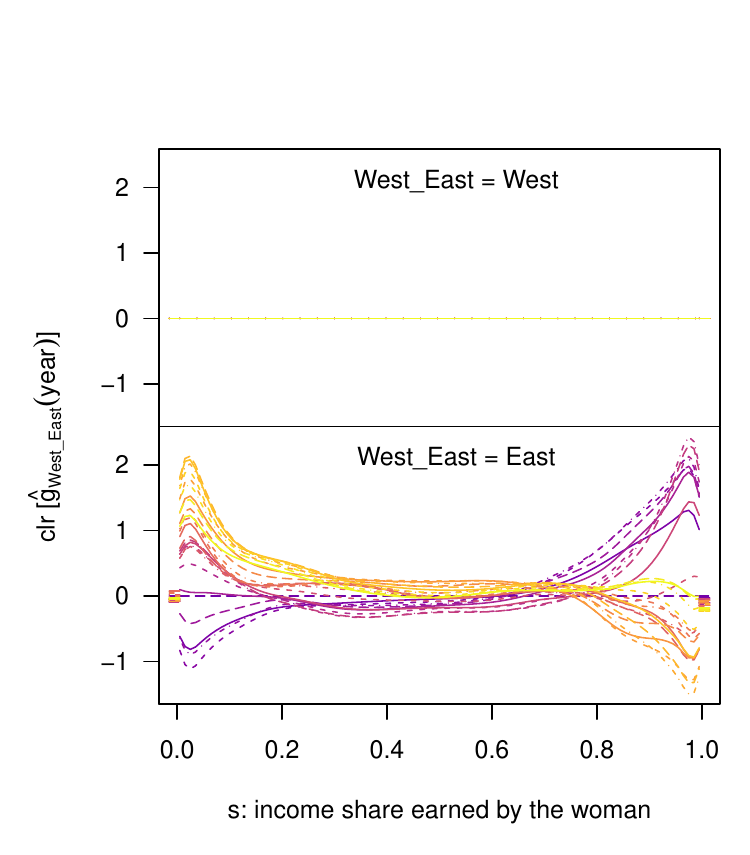}
\includegraphics[width=0.7\textwidth]{Images/year_legend.pdf}
\end{center}
\vspace{-0.5cm}
\caption{Expected densities for couples without minor children living in \emph{West} vs. \emph{East} Germany over time [left] and clr transformed estimated interaction effects of \emph{West\_East} and \emph{year} [right].}
\end{figure}

% year, c_age
\vspace{-1cm}
\begin{figure}[H]
\begin{center}
\includegraphics[width=0.49\textwidth]{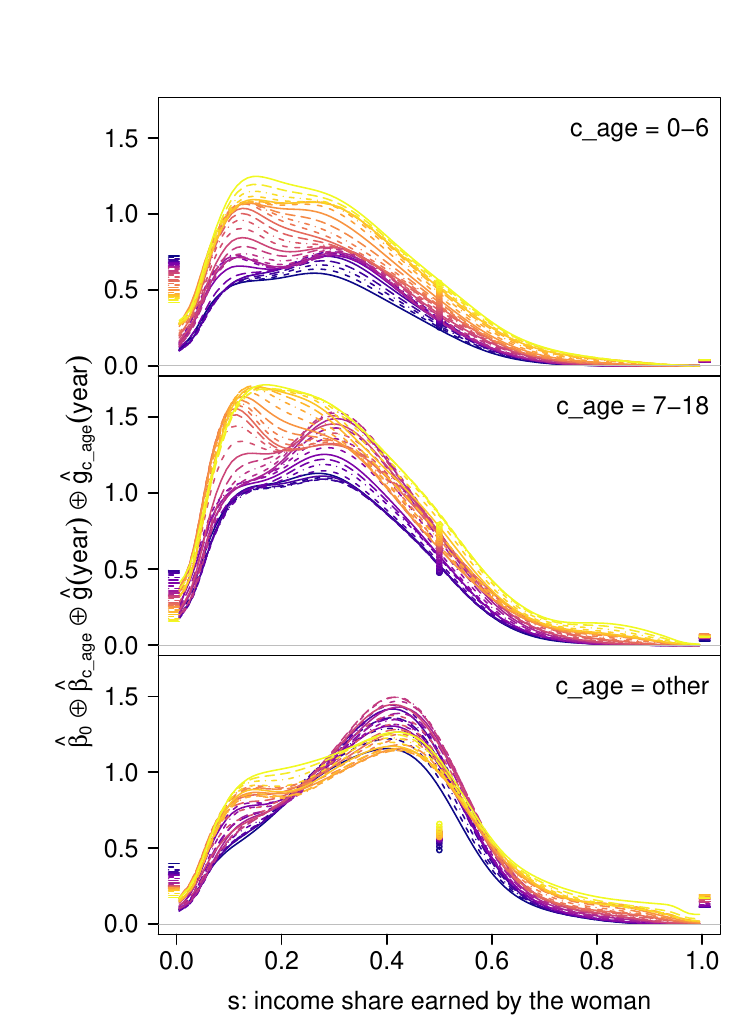}
\includegraphics[width=0.49\textwidth]{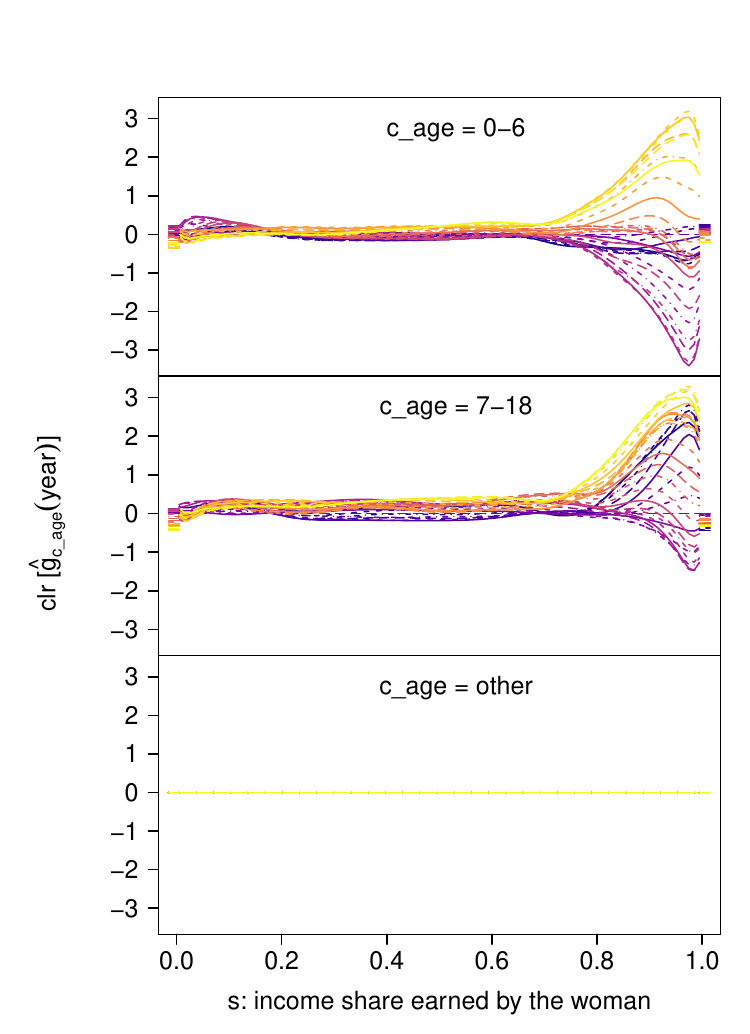}
\includegraphics[width=0.7\textwidth]{Images/year_legend.pdf}
\end{center}
\vspace{-0.5cm}
\caption{Expected densities for couples living in \emph{West} Germany for all three values of \emph{c\_age} over time [left] and clr transformed estimated interaction effects of \emph{c\_age} and \emph{year} [right]. \label{appendix_old_cage_year}}
\end{figure}
%\vspace{-1cm}
%\begin{figure}[H]
%\begin{center}
%\includegraphics[width=\textwidth]{Images/estimated_year_child_group_clr.pdf}
%%
%\includegraphics[width=0.7\textwidth]{Images/year_legend.pdf}
%\end{center}
%\vspace{-0.5cm}
%\caption{Clr transformed estimated group-specific smooth effects of \emph{year} for \emph{c\_age}.}
%\end{figure}
%

% year, old_new, c_age
\vspace{-1cm}
\begin{figure}[H]
\begin{center}
\includegraphics[width=0.49\textwidth]{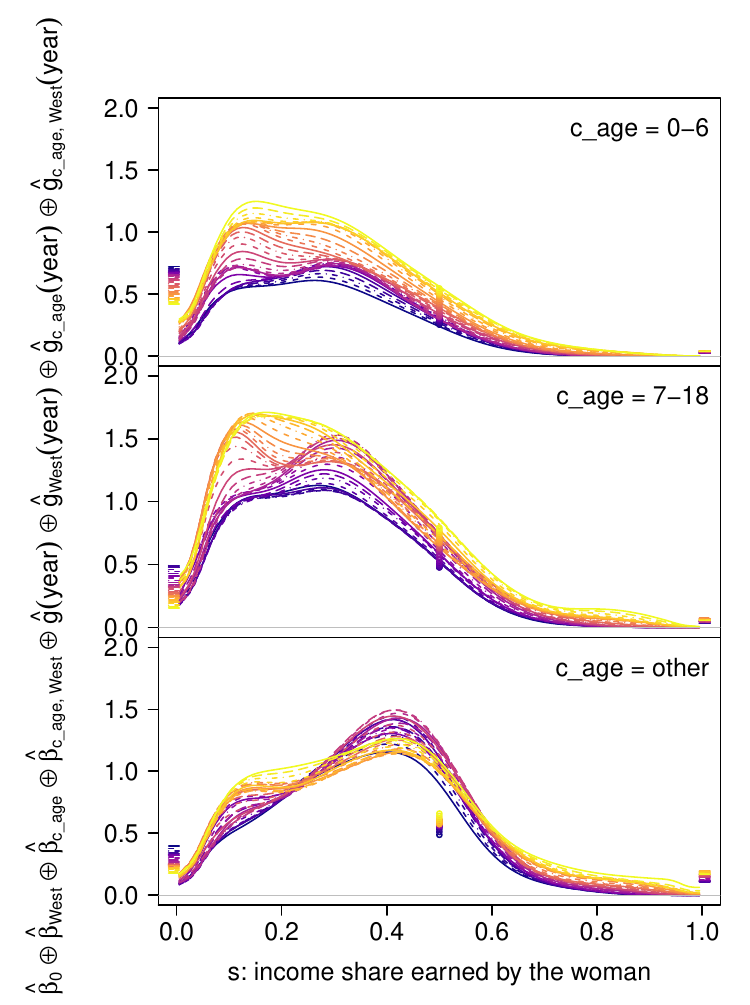}
\includegraphics[width=0.49\textwidth]{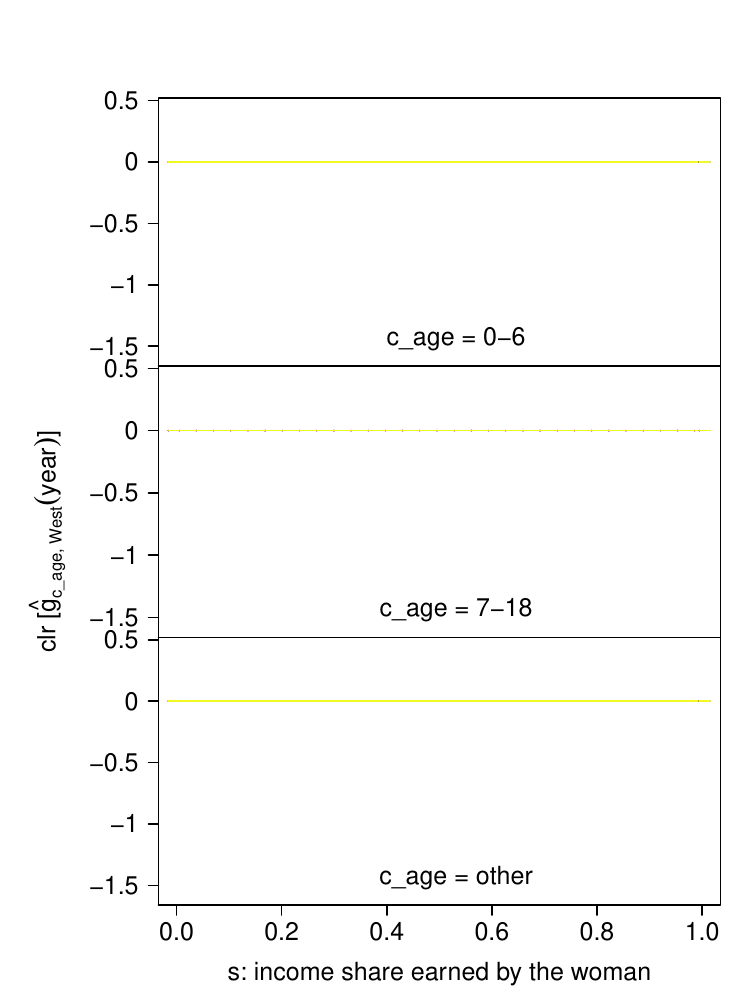}
\includegraphics[width=0.49\textwidth]{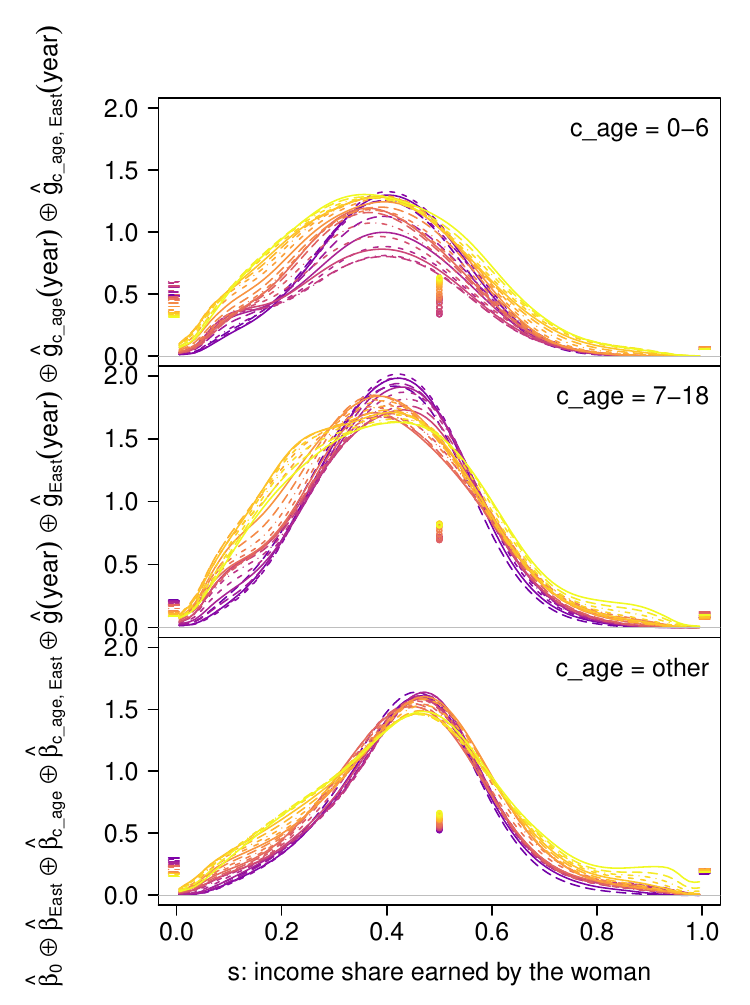}
\includegraphics[width=0.49\textwidth]{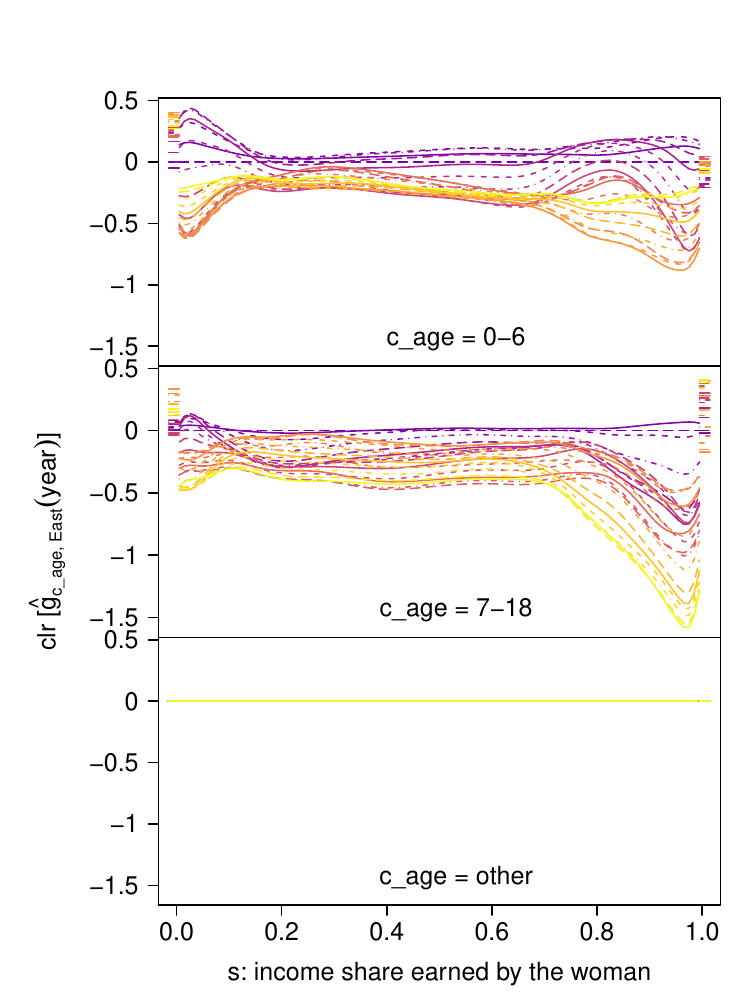}
\includegraphics[width=0.7\textwidth]{Images/year_legend.pdf}
\end{center}
\vspace{-0.5cm}
\caption{Expected densities for couples living in \emph{West} [top] vs. living in \emph{East} Germany [bottom] for all three values of \emph{c\_age} over time [left] and clr transformed estimated interaction effects of \emph{c\_age}, \emph{West\_East} and \emph{year} [right].\label{appendix_estimated_year_old_new_cgroup}}
\end{figure}

\section{Simulation study}\label{appendix_simulation}

\subsection{Definition of relMSE}\label{appendix_simulation_relmse_def}
Consider the setting of our simulation study in Section~5. % \ref{chapter_simulation}.
There, we use the relative mean squared error (relMSE) motivated by \citet{brh2015} to evaluate the goodness of the estimation results. %, where the mean squared error is standardized with respect to the global variability of the true density.
For predictions and estimated partial effects it is defined as
\begin{align*}
\rMSE (\eh)
:= \frac{\frac{1}{\upsilon (\Ycal)} \, \int_{\Ycal} \Vert E(y) \ominus \eh(y) \Vertb^2 \, \dups (y)}{\frac{1}{\upsilon (\Ycal)} \, \int_{\Ycal} \Vert E(y) \ominus \bar{E} \Vertb^2 \, \dups (y)}
= \frac{\int_{\Ycal} \Vert E(y) \ominus \eh(y) \Vertb^2 \, \dups (y)}{\int_{\Ycal} \Vert E(y) \Vertb^2 \, \dups (y)},
\end{align*}
where $\Ycal$ denotes the set $\{1, \ldots, 552\}$ for predictions, the set of possible values for categorical covariates (group-specific effects), e.g., $\{ \emph{West}, \emph{East} \}$ for the covariate \emph{West\_East}, or the observed range for scalar covariates (linear/flexible effects), e.g., $[1984, 2016]$ for \emph{year}.
For effects depending on more than one covariate, $\Ycal$ is the Cartesian product of the appropriate sets.
The measure $\upsilon$ is the counting measure, the Lebesgue measure, or a product measure thereof, respectively.
The estimated densities are denoted by $\eh (y) \in \B$ for $y \in \Ycal$, corresponding to $\fh_i = \fh (i), i \in \Ycal$ for predictions or $\hh_j(\xf), \xf \in \Ycal$ for estimated effects.
Analogously, the true densities are denoted by $E(y)$.
Their overall mean, $\bar{E} := 1 / \upsilon (\Ycal) % \frac{1}{\upsilon (\Ycal)}
\int_{\Ycal} \int_\Tcal E(y) \, \dmu \, \dups (y)$, is $0 \in \B$ as a constant.

\subsection{RelMSEs and MSEs for all effects}\label{appendix_simulation_relmse}
Figure~\ref{relimse_effects} shows the complete simulation results.
The left side illustrates the relMSEs (see Section~5) % \ref{chapter_simulation}) 
for the predictions and all partial effects.
The boxplots on the right correspond to the respective mean squared errors (MSEs), i.e., the numerators of the relMSEs.
Furthermore, the denominators, i.e., the mean squared norms of the true effects, are added in form of a blue ``x''.
The right side shows that larger relMSEs, in particular for $\betah_{region}$, $\betah_{\text{\emph{c\_age, West\_East}}}$, $\gh_{\text{\emph{West\_East}}} (year)$, $\gh_{c\_age} (year)$, and $\gh_{\text{\emph{c\_age, West\_East}}}(year)$, arise from the mean squared norm of the true effects for the respective effects being small.
This means, the relative mean squared errors are large, because the true effects are small but not because the errors are large.
\begin{figure}[H]
\begin{center}
\includegraphics[width=0.49\textwidth]{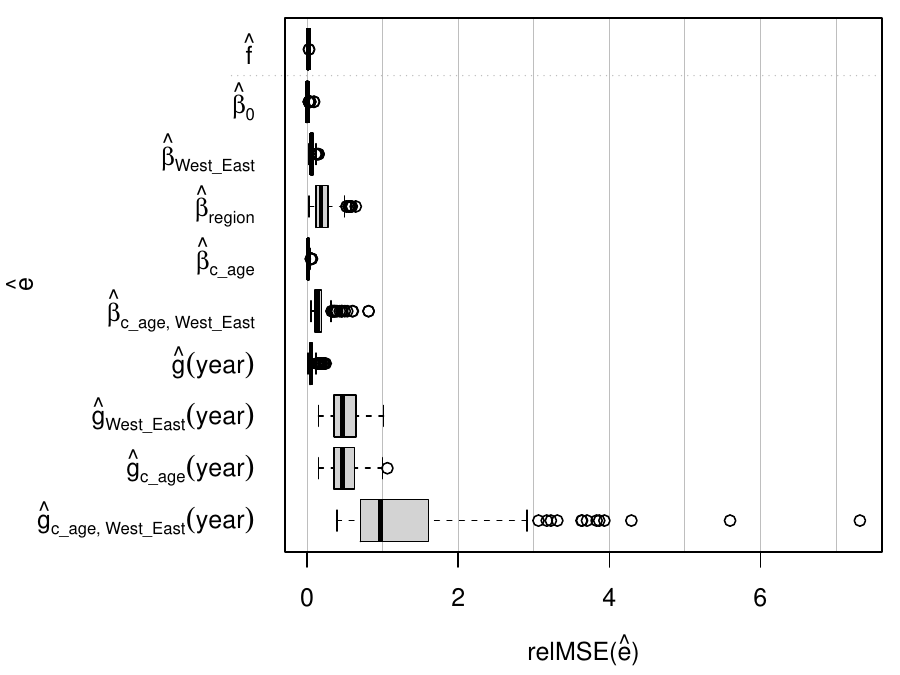}
\includegraphics[width=0.49\textwidth]{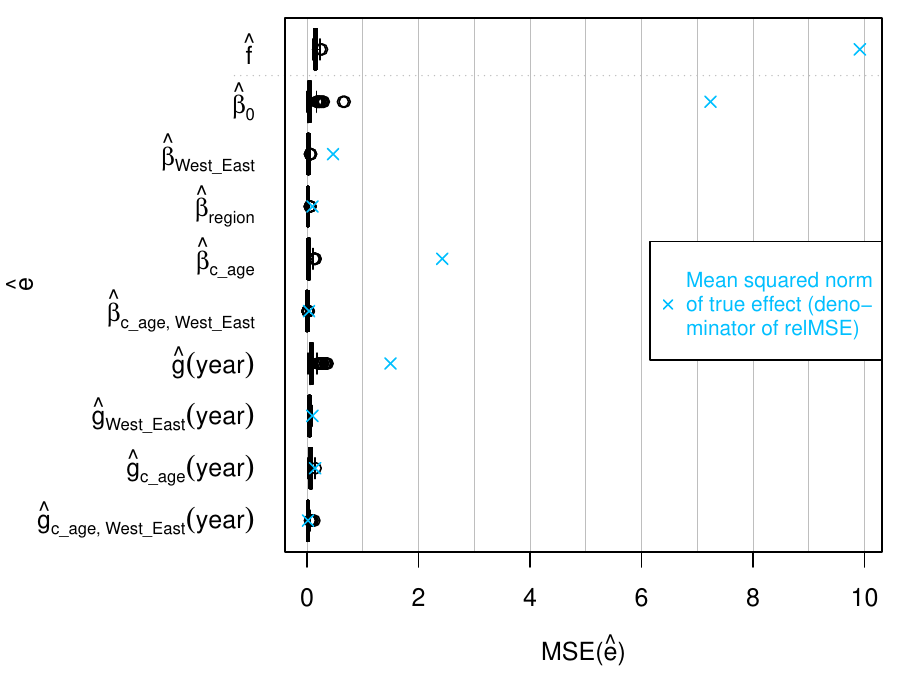}
\end{center}
\vspace{-0.5cm}
\caption{RelMSE [left] and MSE [right] for predictions [top] and all partial effects [bottom]. \label{relimse_effects}}
\end{figure}

\subsection{Model selection}\label{appendix_simulation_selected_effects}
Table~\ref{table_not_selcted_effects} summarizes how many times effects are not selected over the $200$ simulation runs.
It contains the counts for the separately estimated continuous and discrete models, as well as for the final combined model in the last three columns, each of which sums up to 200 (total number of simulation runs).
The rows of the table are grouped by the number of effects that are not selected in a simulation run, ranging from no effects (i.e., all effects are selected) to three effects.
The table contains all effects (second column) that are not selected in at least one simulation run in either the continuous or the discrete model.
These are exactly the four interaction effects.
In particular, the main effects are selected in all simulation runs in both models (continuous and discrete).
%%One effect can thus appear several times
%To get the total number of simulation runs, where a specific effect was not selected, one has to add up all numbers, for which the respective effect appears in the second column.
%%There is no overlap in the counts, i.e., the number of simulation runs, where two or three effects were not selected does not include the number of simulation runs, where one of those was not selected.
Note that as soon as one effect is selected in either the continuous or the discrete model, it is also selected in the combined model.
Or, put differently, for an effect to be not selected in the combined model, it must not be selected in neither the continuous nor the discrete model.
%Thus, the number of simulation runs, where all effects are selected (first row) is higher for the combined model than for the continuous or discrete models, separately.
%Analogously, the number of simulation runs, where at least one effect is not selected (remaining rows) is lower for the combined model.
This explains that in the combined model, there are only few simulation runs, where an effect is not selected at all ($4$ in total), while for the separate models the numbers are noticeably higher.
Most remarkably, in the continuous model, $\betah_{\text{\emph{c\_age, West\_East}}}$ is not selected in $131$ simulation runs in total (including simulation runs, where additional effects are not selected).
%If an effect is only selected in the discrete model, this means that there is no effect on the distribution within dual-earner households.
%In contrast, if an effect is only selected in the continuous model, it does influence the distribution within dual-earner households, but not the relative frequencies of single-earner and dual-earner households.

% Effects not selected in continuous models:
%    8         8 9     9   7 8     4     7   4 7 8   4 8
%  118    60     8     5     3     2     2       1     1
% Effects not selected in discrete models:
%       4   4 7   7   9
% 163  33     2   1   1
% Effects not selected in joint models:
%      4  4 7   7   9
%196   1    1   1   1
% Effect number correspondence:
% 4: old_new * year
% 7: c_age * year
% 8: c_age * old_new
% 9: c_age * old_new * year

\ifnum\value{aoas}=1
{
%  \spacingset{1}
\begin{table}[H]
\caption{Counts of effects not selected over the 200 simulation runs.
\label{table_not_selcted_effects}}
\begin{center}
\begin{tabular}{l|l|c|c|c}
\hline
& \textbf{Effect(s) not selected} & \multicolumn{3}{c}{\textbf{Number of simulation runs}}\\
\cline{3-5}
& & continuous & discrete & combined \\
& & model & model & model \\
\hline
All effects selected & & 60 & 163 & 196 \\
 \hline
One effect not & $\betah_{\text{\emph{c\_age, West\_East}}}$ & 118 & 0 & 0 \\
\hspace{0.1cm} selected & $\gh_{\text{\emph{West\_East}}} (year)$ & 2 & 33 & 1\\
& $\gh_{\text{\emph{c\_age}}} (year)$ & 2 & 1 & 1 \\
& $\gh_{\text{\emph{c\_age, West\_East}}}(year)$ & 5 & 1 & 1 \\
\hline
Two effects not & $\betah_{\text{\emph{c\_age, West\_East}}}$,  & \multirow{2}{*}{1} & \multirow{2}{*}{0} & \multirow{2}{*}{0} \\
\hspace{0.1cm} selected & \hspace{0.1cm} $\gh_{\text{\emph{West\_East}}} (year)$ & & & \\
 & $\betah_{\text{\emph{c\_age, West\_East}}}$, & \multirow{2}{*}{3} & \multirow{2}{*}{0} & \multirow{2}{*}{0} \\
& \hspace{0.1cm} $\gh_{\text{\emph{c\_age}}} (year)$ & & & \\
& $\betah_{\text{\emph{c\_age, West\_East}}}$,  & \multirow{2}{*}{8} & \multirow{2}{*}{0} & \multirow{2}{*}{0} \\
& \hspace{0.1cm} $\gh_{\text{\emph{c\_age, West\_East}}}(year)$ & & & \\
& $\gh_{\text{\emph{West\_East}}} (year)$, & \multirow{2}{*}{0} & \multirow{2}{*}{2} & \multirow{2}{*}{1} \\
& \hspace{0.1cm} $\gh_{\text{\emph{c\_age}}} (year)$ & & & \\
\hline
Three effects not & $\betah_{\text{\emph{c\_age, West\_East}}}$, & \multirow{3}{*}{1} & \multirow{3}{*}{0} & \multirow{3}{*}{0} \\
\hspace{0.1cm} selected & \hspace{0.1cm} $\gh_{\text{\emph{West\_East}}} (year)$, & & & \\
 & \hspace{0.1cm} $\gh_{\text{\emph{c\_age}}} (year)$ & & & \\
\hline
\end{tabular}
\end{center}
\end{table} 
} \else
{
\begin{table}[H]
\begin{center}
\begin{tabular}{l|l|c|c|c}
& \textbf{Effect(s) not selected} & \multicolumn{3}{c}{\textbf{Number of simulation runs}}\\
\cline{3-5}
& & continuous & discrete & combined \\
& & model & model & model \\
\hline
All effects selected & & 60 & 163 & 196 \\
 \hline
One effect not & $\betah_{\text{\emph{c\_age, West\_East}}}$ & 118 & 0 & 0 \\
\hspace{0.1cm} selected & $\gh_{\text{\emph{West\_East}}} (year)$ & 2 & 33 & 1\\
& $\gh_{\text{\emph{c\_age}}} (year)$ & 2 & 1 & 1 \\
& $\gh_{\text{\emph{c\_age, West\_East}}}(year)$ & 5 & 1 & 1 \\
\hline
Two effects not & $\betah_{\text{\emph{c\_age, West\_East}}}$,  & \multirow{2}{*}{1} & \multirow{2}{*}{0} & \multirow{2}{*}{0} \\
\hspace{0.1cm} selected & \hspace{0.1cm} $\gh_{\text{\emph{West\_East}}} (year)$ & & & \\
 & $\betah_{\text{\emph{c\_age, West\_East}}}$, & \multirow{2}{*}{3} & \multirow{2}{*}{0} & \multirow{2}{*}{0} \\
& \hspace{0.1cm} $\gh_{\text{\emph{c\_age}}} (year)$ & & & \\
& $\betah_{\text{\emph{c\_age, West\_East}}}$,  & \multirow{2}{*}{8} & \multirow{2}{*}{0} & \multirow{2}{*}{0} \\
& \hspace{0.1cm} $\gh_{\text{\emph{c\_age, West\_East}}}(year)$ & & & \\
& $\gh_{\text{\emph{West\_East}}} (year)$, & \multirow{2}{*}{0} & \multirow{2}{*}{2} & \multirow{2}{*}{1} \\
& \hspace{0.1cm} $\gh_{\text{\emph{c\_age}}} (year)$ & & & \\
\hline
Three effects not & $\betah_{\text{\emph{c\_age, West\_East}}}$, & \multirow{3}{*}{1} & \multirow{3}{*}{0} & \multirow{3}{*}{0} \\
\hspace{0.1cm} selected & \hspace{0.1cm} $\gh_{\text{\emph{West\_East}}} (year)$, & & & \\
 & \hspace{0.1cm} $\gh_{\text{\emph{c\_age}}} (year)$ & & &
\end{tabular}
\caption{Counts of effects not selected over the 200 simulation runs.
\label{table_not_selcted_effects}}
\vspace{-0.5cm}
\end{center}
\end{table} 
} \fi

%%\ifnum\value{aoas}=1
%%{
%\end{appendix}
%%}
%%\else
%%\fi
% Literaturverzeichnis Appendix
\addcontentsline{toc}{section}{References}
\printbibliography[heading=bibliography]
\end{refsection}

\end{document}